\documentclass[aos]{imsart}
\RequirePackage{amsthm,amsfonts,amssymb}
\RequirePackage[numbers,sort]{natbib}
\RequirePackage[colorlinks,citecolor=blue,urlcolor=blue]{hyperref}
\RequirePackage{graphicx}

\usepackage[reqno]{amsmath}

\usepackage{enumerate}% allows enumerate with latin characters
\usepackage{dsfont}% allows doublestroke symbols with \mathds{}
\usepackage[usenames,dvipsnames]{color}
\usepackage{mathrsfs}
\usepackage{placeins}
\usepackage{graphicx}
\usepackage{tikz}
\usetikzlibrary{arrows}
\usepackage{url}
\usepackage{comment}
\usepackage{multicol}
\usepackage{cleveref}
\usepackage{stmaryrd}
\usepackage{calc}
\usepackage{xfrac}
\usepackage{tikz}
\tikzstyle{vertex}=[circle, draw, inner sep=0pt, minimum size=6pt]
\newcommand{\vertex}{\node[vertex]}
\usepackage{natbib}

\renewcommand{\thetable}{\arabic{section}.\arabic{table}}
\numberwithin{equation}{section}

\theoremstyle{plain} %default (text in italic)
\newtheorem{theorem}{Theorem}[section]
\newtheorem{lemma}[theorem]{Lemma}
\newtheorem{proposition}[theorem]{Proposition}
\newtheorem{corollary}[theorem]{Corollary}

\theoremstyle{remark}
\newtheorem{definition}[theorem]{Definition}
\newtheorem{dfn}[theorem]{Definition}
\newtheorem{example}[theorem]{Example}

\newtheorem{remark}[theorem]{Remark}

\def\tablename{\footnotesize Table}

\newcommand{\bthe}{\begin{theorem}}
\newcommand{\ethe}{\end{theorem}}

\newcommand{\ben}{\begin{enumerate}}
\newcommand{\een}{\end{enumerate}}

\newcommand{\bit}{\begin{itemize}}
\newcommand{\eit}{\end{itemize}}

\newcommand{\beq}{\begin{equation}}
\newcommand{\eeq}{\end{equation}}

\newcommand{\ble}{\begin{lemma}}
\newcommand{\ele}{\end{lemma}}

\newcommand{\DASe}{\begin{definition}\rm}
\newcommand{\ede}{\halmos\end{definition}}

\newcommand{\bco}{\begin{corollary}}
\newcommand{\eco}{\end{corollary}}

\newcommand{\bpr}{\begin{proposition}}
\newcommand{\epr}{\end{proposition}}

\newcommand{\brem}{\begin{remark}\rm}
\newcommand{\erem}{\end{remark}}

\newcommand{\bproof}{\begin{proof}}
\newcommand{\eproof}{\end{proof}}

\newcommand{\bexam}{\begin{example}\rm}
\newcommand{\eexam}{\end{example}}

\newcommand{\bfi}{\begin{fig}}
\newcommand{\efi}{\end{fig}}

\newcommand{\btab}{\begin{tab}}
\newcommand{\etab}{\end{tab}}

\newcommand{\beao}{\begin{eqnarray*}}
\newcommand{\eeao}{\end{eqnarray*}\noindent}

\newcommand{\balo}{\begin{align*}}
\newcommand{\ealo}{\end{align*}}

\newcommand{\balm}{\begin{align}}
\newcommand{\ealm}{\end{align}\noindent}

\newcommand{\beam}{\begin{eqnarray}}
\newcommand{\eeam}{\end{eqnarray}\noindent}

\newcommand{\barr}{\begin{array}}
\newcommand{\earr}{\end{array}}

\newcommand{\E}{\mathbb{E}}
\newcommand{\M}{\mathbb{M}}

\renewcommand\P{\mathbb{P}}
\newcommand{\R}{\mathbb{R}}
\newcommand{\MO}{\mathrm{MO}}

\newcommand{\VF}[1]{{\color{blue} #1}}

\def\MRV{\mathcal{MRV}}

\def\RV{\mathcal{RV}}
\def\PGC{\mathrm{P\text{-}GC}}
\def\PMOC{\mathrm{P\text{-}MOC}}

\def\binfty{\boldsymbol \infty}

\def\cB{\mathcal{B}}

\def\bzero{\boldsymbol 0}
\def\bone{\boldsymbol 1}
\def\bA{\boldsymbol A}

\def\bX{\boldsymbol X}
\def\bY{\boldsymbol Y}
\def\bZ{\boldsymbol Z}

\def\bx{\boldsymbol x}

\def\bz{\boldsymbol z}
\def\bv{\boldsymbol v}

\def\be{\boldsymbol e}

\def\bl{\boldsymbol l}

\def\bE{\mathbf E}

\newcommand{\eqd}{\stackrel{\mathrm{d}}{=}}

\newcommand{\tto}{{t\to\infty}}

\newcommand{\Var}{\text{VaR}}
\newcommand{\VaR}{\text{VaR}}

\newcommand{\CoVaR}{\text{CoVaR}}
\newcommand{\ECI}{\text{ECI}}
\newcommand{\rhom}{\rho^{\vee}}%{\rho_{\max}}

\newcommand{\ov}{\overline}

\newcommand{\vague}{\stackrel{\lower0.2ex\hbox{$\scriptscriptstyle
                    \it{v} $}}{\rightarrow}}
\newcommand{\weak}{\stackrel{\lower0.2ex\hbox{$\scriptscriptstyle
                    \it{w} $}}{\rightarrow}}
\newcommand{\what}{\stackrel{\lower0.2ex\hbox{$\scriptscriptstyle
                    \it{\hat{w}} $}}{\rightarrow}}
\newcommand{\eqdis}{\stackrel{\lower0.2ex\hbox{$\scriptscriptstyle
                    \mathrm{d}$}}{=}}
\newcommand{\distr}{\stackrel{\lower0.2ex\hbox{$\scriptscriptstyle
                    \it{d} $}}{\rightarrow}}
\allowdisplaybreaks %% display breaks for align
\definecolor{darkgreen}{RGB}{0,139,0}
\newcommand{\DAS}[1]{{\color{red} #1}}

\begin{document}

\begin{frontmatter}
\title{Measuring risk contagion in \vspace*{0.2cm} \\  financial networks with CoVaR}% %using a bipartite graph of overlapping portfolios}%: \
%\title{Measuring risk contagion in financial networks: \\ The effects of asymptotic independence}
%\title{Measuring risk contagion in financial networks: \\ The effect of asymptotically independent objects}
%\title{Systemic risks with asymptotic independence in bipartite networks}
%\title{Extremal events in financial networks: \\ The effects of asymptotic independence}
%\title{Financial risk measures in complex networks:\\ $\mbox{}$ \vspace*{-0.2cm} \\ The effect of asymptotic independence}
%\title{Risk contagion in financial networks: \\ The effects of asymptotic independence}
%\title{Systemic risk in financial networks: \\ heavy tails and asymptotic independence}
\runtitle{Risk contagion in networks}

\begin{aug}
  \author[A]{\fnms{Bikramjit} \snm{Das}\ead[label=e1]{bikram@sutd.edu.sg}\orcid{0000-0002-6172-8228}}
    \and
   \author[B]{\fnms{Vicky} \snm{Fasen-Hartmann}\ead[label=e2]{vicky.fasen@kit.edu}\orcid{0000-0002-5758-1999}}
%\thanksref{t1}\thankstext{t1}
 \address[A]{Engineering Systems and Design, Singapore University of Technology and Design  \printead[presep={,\ }]{e1}}

   \address[B]{Institute of Stochastics, Karlsruhe Institute of Technology\printead[presep={,\ }]{e2}}

%  \thanksref{T1}
%  \thankstext{T1}{}

  \runauthor{B. Das and  V. Fasen-Hartmann}
\end{aug}

\begin{abstract}
The stability of a complex financial system may be assessed by measuring risk contagion between various financial institutions with relatively high exposure. We consider a financial network model using a bipartite graph of financial institutions (e.g., banks, investment companies, insurance firms) on one side and financial assets on the other. Following empirical evidence, returns from such risky assets are modeled by heavy-tailed distributions, whereas their joint dependence is characterized by copula models exhibiting a variety of tail dependence behavior. We consider CoVaR, a popular measure of risk contagion and study its asymptotic behavior under broad model assumptions. We further propose the \emph{Extreme CoVaR Index} (ECI) for capturing the strength of risk contagion between risk entities in such networks, which is particularly useful for models exhibiting asymptotic independence. The results are illustrated by providing precise expressions of CoVaR and ECI when the dependence of the assets is modeled using two well-known multivariate dependence structures: the Gaussian copula and the Marshall-Olkin copula.

\end{abstract}

\begin{keyword}[class=AMS]
\kwd[Primary ]{62G32}  % Statistics of extreme values; tail inference
\kwd{91G45} %Financial networks (including contagion, systemic risk, regulation)
\kwd{91G70}  %Statistical methods; risk measures
\kwd[; Secondary ]{60G70} %Extreme value theory; extremal stochastic processes
\kwd{62H05} %Characterization and structure theory for multivariate probability distributions; copulas
\end{keyword}

%\begin{keyword}[class=JEL]
%\kwd{\emph{JEL subject classifications: } }
%\kwd[Primary ]{D81} %Criteria for Decision-Making under Risk and Uncertainty
%\kwd{C14} % semiparametric and non-parametric modeling
%\kwd[; Secondary ]{C61} % OPtimization Techniques; Programming Models; Dynamic Analysis
%\kwd{C63} % Computational Techniques
%\kwd{D81} %Criteria for Decision-Making under Risk and Uncertainty
%\kwd{G11} % Portfolio choice
%\kwd{G22} % insurance
%\end{keyword}

\begin{keyword}
%\kwd{bipartite graphs}
%\kwd{asymptotic independence}
\kwd{bipartite graph}
\kwd{copula models}
\kwd{CoVaR}
\kwd{financial network}
\kwd{heavy tails}
\kwd{Gaussian copula}
%\kwd{Marshall-Olkin copula}
%\kwd{multivariate regular variation}
%\kwd{mutual asymptotic independence}
%\kwd{risk measures}
\kwd{risk contagion}
%\kwd{regular variation}
%\kwd{networks}
\end{keyword}

\end{frontmatter}
%\maketitlehttps://arxiv.org/abs/math.PR/0000000

%{\rmfamily \tableofcontents}

\section{Introduction}\label{sec:intro} 

The global financial crisis of 2007-2009 brought into immediacy the need to understand risk contagion in order to assess the stability of a financial system. The high level of inter-connectivity between various financial institutions has been argued to be one of few key contributors to such systemic financial instability, see \cite{allen:gale:2000,eisenberg:noe:2001,acemoglu:2015systemic,gai:kapadia:2010,feinstein:rudloff:weber:2017,benoit:colliard:hurlin:2016} for some compelling arguments on such network effects.  Among the various frameworks proposed to capture risk contagion in financial systems, a popular one has been to model the financial system as a bipartite graph of financial institutions (e.g., banks, investment companies, insurance firms) on one side and overlapping financial assets where the banks invest, on the other. We will use this bipartite network as our model which has been observed in many financial markets; %(\citet{caccioli:shrestha:moore:farmer:2014,gao:2022systemic,pichler:2021systemic,poledna:etal:2021,fan:2024}); 
see the Chinese and Mexican financial system in \Cref{fig:netex}. 
%Following these examples % of the structures seen in various financial systems (see \Cref{fig:netex}), 
%we propose a bank-asset bipartite graphical network model of risk-sharing. 
This framework has been used not only for modeling bank-asset risk sharing (\cite{caccioli:shrestha:moore:farmer:2014,poledna:etal:2021,fan:2024}) 
but also bank-firm credit network (\cite{marotta:etal:2015}), claims in insurance markets  (\cite{kley:kluppelberg:reinert:2016,kley:kluppelberg:reinert:2017}), etc.

% The a bipartite bank-asset  network model has also been used   to study financial contagion (\citet{caccioli:shrestha:moore:farmer:2014, kley:kluppelberg:reinert:2016}) which we will further elaborate in Section \ref{subsec:network}.

\begin{figure}[ht]
\includegraphics[width=0.44\linewidth]{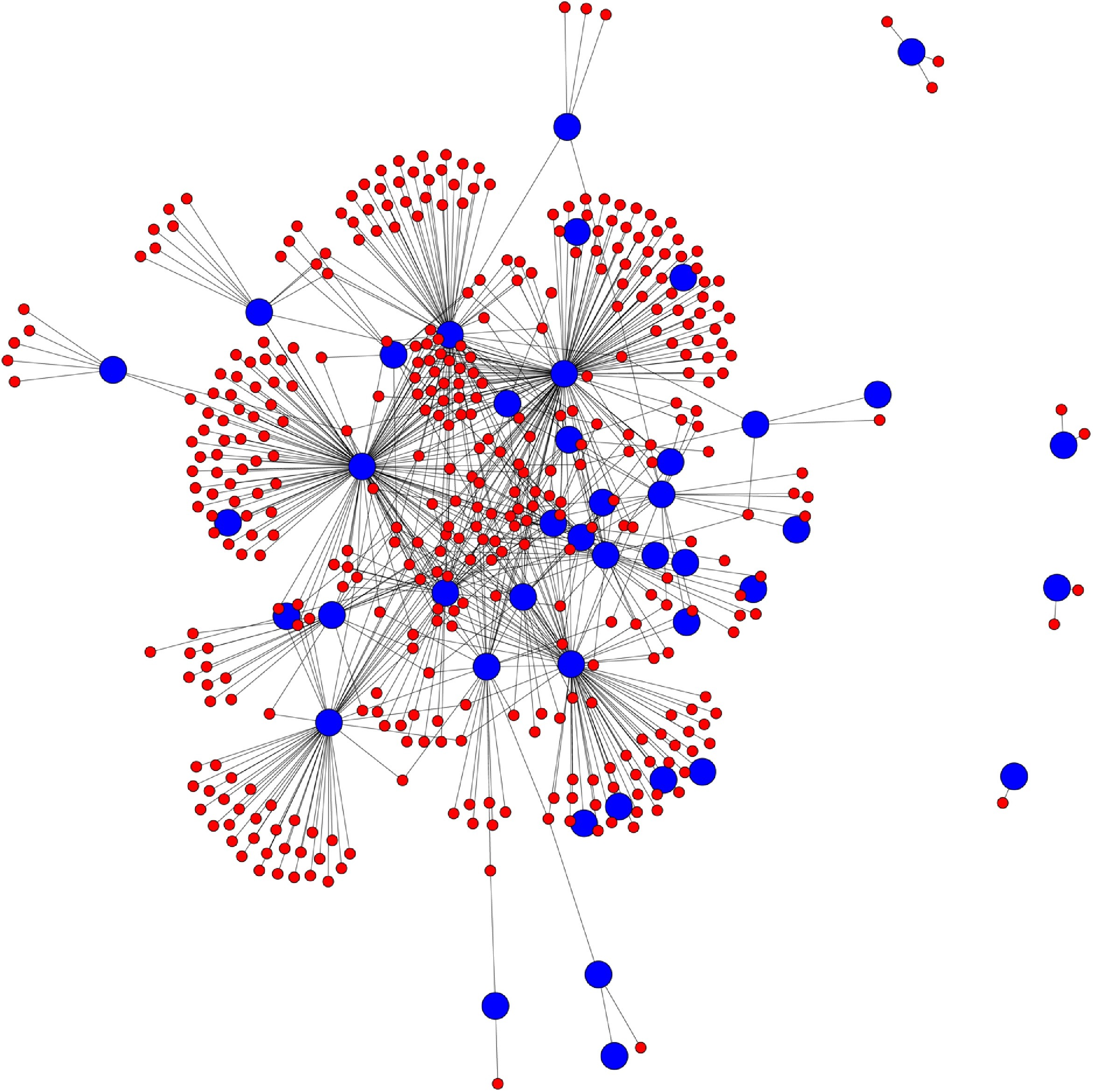} \;\;\;\quad\quad\quad \includegraphics[width=0.44\linewidth]{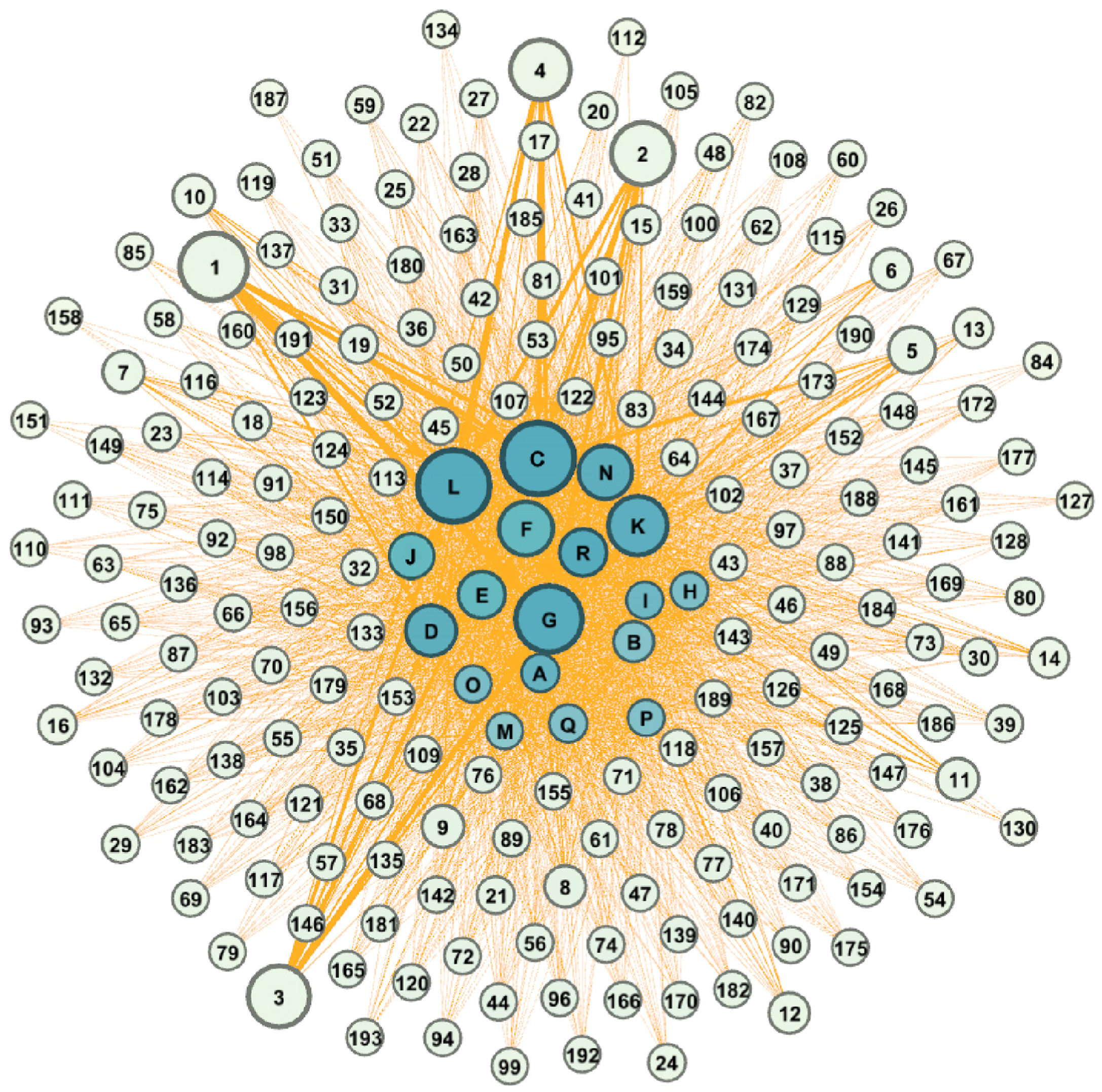}
\caption{(1) Left network: Bank-asset bipartite network of the Mexican financial system on a particular day. Nodes in the network represent banks (blue) and assets (red). Links between an asset and a bank exist if the bank holds the asset in its portfolio. (Courtesy: \citet{poledna:etal:2021}); (2) Right network: Bipartite network structure of common assets held by Chinese banks in 2021. The numbers and letters in the circles respectively represent identifiers
for banks and industries (Courtesy: \citet{fan:2024})}\label{fig:netex}
\end{figure} 

 To fix notations, consider a vertex set of banks/agents $\mathcal{A}=\mathbb{I}_q=\{1,\ldots,q\}$  and a vertex
set of assets/objects  $\mathcal{O}=\mathbb{I}_d=\{1,\ldots, d\}$. Denote by $Z_{j}$  the risk attributed to the $j$-th object and $\bZ=(Z_{1},\ldots, Z_{d})^{\top}$ forms the risk vector. 
%Assume that the graph creation process is independent of $\bZ$.
Each bank/agent $k \in \mathcal{A} $ is connected to a number of  assets/objects $j \in \mathcal{O}$
and these connections may be modeled in a stochastic manner following some probability distribution; see \Cref{fig:bipnet} for a representative example of such a bipartite network.
 A basic model assumes that $k$ and $j$ connect with probability
 $$\P(k\sim j) = p_{k j} \in [0,1], \quad k \in \mathcal{A},\, j \in \mathcal{O}.$$
The proportion of loss of asset/object $j$ affecting bank/agent $k$ is denoted by
\begin{align}\label{eq:wij}
 f_{k}(Z_{j}) = \mathds{1}{(k\sim j)}W_{k j} Z_{j},
\end{align}
where $W_{k j}>0$ denotes the proportional effect of the $j$-th object on the $k$-th agent.
Defining the $q \times d$ adjacency matrix $\bA$ by
\begin{align}\label{Aarbitrary}
    A_{k j} &= \mathds{1}{(k \sim j)}W_{k j},
\end{align}
the total exposure of the banks/agents (think of negative log-returns on market equity value) is given by $\bX=(X_{1}, \ldots, X_{q})^{\top}$, where $X_{k} = \sum_{j=1}^{d} f_{k}(Z_{j})$ can be represented as
\begin{align}\label{eqn:x=az}
\bX=\bA\bZ.
\end{align}
We may assume that the graph creation process is independent of $\bZ$, i.e., $\bA$ and $\bZ$ are independent. Now the behavior of $\bX$ will be governed by both the network represented by $\bA$ and the underlying distribution of $\bZ$, the risk related to the assets/objects.  The goal of this paper is to study risk contagion in terms of the extremal behavior of $\bX$ under reasonable assumptions on $\bZ$ and $\bA$.
%Our key interest will be in copulas exhibiting asymptotic independence in some form, and it turns out that the Gaussian copula and the Marshall-Olkin copula are well-suited for this purpose especially for general dimensions.

%
%the CAPM with risk factor with tails finer than  \citet{Zhou:2010b}\citet{Oordt:Zhou} \citet{Huang:Vries:2002} \citet{bradley:taqqu:2003} \citet{MAINIK2015} \citet{Huschens2000}}

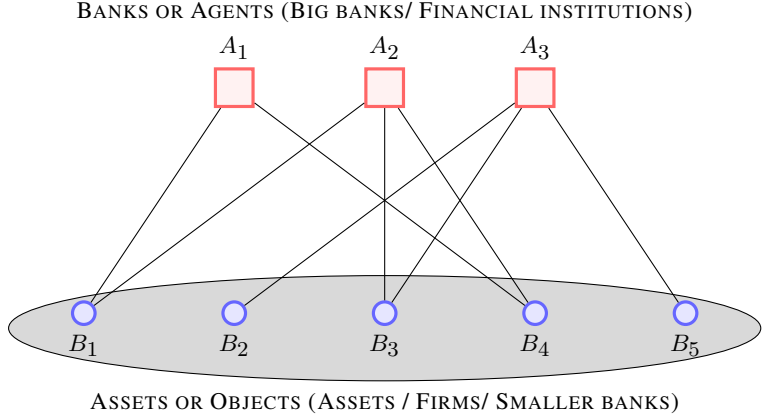
\begin{figure}[th]
\[\begin{tikzpicture}[
squarednode/.style={rectangle, draw=red!60, fill=red!5, very thick, minimum size=5mm},
roundnode/.style={circle, draw=blue!60, fill=blue!10, very thick, minimum size=3mm},
]
    \node at (4,3) {\sc{Banks or Agents (Big banks/ Financial institutions)}};
    \node at (4,-2.2) {\sc{Assets or Objects (Assets / Firms/ Smaller banks)}};
    \draw[fill=gray!30!white](4,-1.2)ellipse(5cm and .7cm);
	\vertex[squarednode] (A1) at (2,2) [label=above:$A_1$] {};
	\vertex[squarednode] (A2) at (4,2) [label=above:$A_2$] {};
	\vertex[squarednode] (A3) at (6,2) [label=above:$A_3$] {};
	\vertex[roundnode] (B1) at (0,-1) [label=below:$B_1$] {};
	\vertex[roundnode] (B2) at (2,-1) [label=below:$B_2$] {};
	\vertex[roundnode] (B3) at (4,-1) [label=below:$B_3$] {};
	\vertex[roundnode] (B4) at (6,-1) [label=below:$B_4$] {};
     \vertex[roundnode] (B5) at (8,-1) [label=below:$B_5$] {};
	\path
		(A1) edge (B1)
		(A1) edge (B4)
		(A2) edge (B1)
		(A2) edge (B3)
		(A3) edge (B2)
		(A3) edge (B3)
		(A2) edge (B4)
        (A3) edge (B5)
        ;
       % (B1) edge[dashed] (B2)
       % (B1) edge[dashed] (B3)	;
\end{tikzpicture}\]
\caption{A bipartite network with $q=3$ banks (or agents) and $d=5$ assets (or objects). We model inter-dependencies via a multivariate model for the assets and the bipartite network (which may be assumed to be random or fixed) connecting assets to banks.}\label{fig:bipnet}
\end{figure}

%\subsection{Modeling the underlying distribution}\label{subsec:multiheavy}
%\marginpar{\VF{mention new paper \citet{das:fasen:2024b}}}
{In order to model $\bZ$, the returns from the asset, we note that} financial returns have often been empirically observed to be heavy-tailed; see \cite{ibragimov:prokhorov:2017,resnickbook:2007, embrechts:kluppelberg:mikosch:1997,Adler:Taqqu,Rachev:2003}. For this paper, we will model the distribution of $\bZ$ to be heavy-tailed using the paradigm of multivariate regular variation (\cite{bingham:goldie:teugels:1989,embrechts:kluppelberg:mikosch:1997,resnickbook:2007}); further details are given in \Cref{sec:mrvwithasyind}.  To model the dependence between the components of $\bZ$, we resort to a few popular multivariate dependence structures, namely,  the asymptotically strongly dependent case, the Gaussian dependence, an erstwhile popular model for dependence in many domains including finance (\cite{fisher_etal:2009,malevergne:sornette:2003}) and the Marshall-Olkin dependence often used in reliability/failure modeling in large systems (\cite{yuge:maruyama:yanagi:2016,marshall:olkin:1967}). Our key results are more generally applicable, and we only provide explicit expressions for the above examples.
 Interestingly the latter dependence structures, i.e., the Gaussian copula and the Marshall-Olkin copula possess the property of (pairwise) asymptotic independence in the tails, i.e., extreme values are less likely to occur simultaneously; see \cite{das:fasen:2024b} for details. This property of asymptotic independence has also been empirically observed in international equity markets (\cite{Bradley:Taqqu:2004,poon:rockinger:tawn:2001, poon:rockinger:tawn:2003}). 

We note as well that an affine linear dependence structure similar to \eqref{eqn:x=az} is also found in certain factor models. For example, in the popular single-index factor model, namely the Capital Asset Pricing Model (CAPM),  the rate of return of an asset is divided into a market return and an idiosyncratic risk of the return, which are assumed to be independent and the weights are the $\beta$'s introduced in the CAPM (see \cite{Zhou:2010b,Huang:Vries:2002,Huschens2000} for details). If the idiosyncratic risks have heavier tails than the market returns of the assets, then the assets turn out to be asymptotically independent, which is again a modeling component of the present paper.

%It is well-known that tail risks between two such return values from different financial objects tend to exhibit asymptotic independence, i.e.,  \DAS{Include literature \citet{malevergne:sornette:2003}} extreme values are less likely to occur simultaneously. 
%Surprisingly, asymptotic independence in dimensions larger than two has received little attention in the literature. A key contribution of our work is to propose a notion of \emph{asymptotic independence} in a truly multivariate setting; see \Cref{def:mai}. We observe that such a definition allows us to understand multivariate tail dependence in a more nuanced manner, for example, we are able to capture the true dimensionality of extreme tail events of various order. This new definition becomes useful as well for modeling the underlying risks while computing systemic risk in networked systems as portrayed in \Cref{fig:bipnet}. 

{Finally, our key goal is to assess risk contagion in financial systems.} Regulatory bodies like Basel Committee on Banking Supervision, Solvency II and the Swiss Solvency Test for insurance regulations have regularly recommended monitoring of \emph{individual measures of risk} exposure like Value-at-Risk (\cite{jorion:2006}) and  Expected Shortfall (\cite{Artzner1999}) for financial and insurance institutions so that adequate capital is reserved to avoid catastrophic losses. It is apparent that for capturing \emph{risk contagion} we need to assess the potential loss for entities/institutions of interest conditioning on an extreme loss event pertaining to one or more other financial institution(s). In this regard,  conditional risk measures like CoVaR (\cite{adrian:brunnermeier:2016}, \cite{girardi:ergun:2013}), Marginal Expected Shortfall (MES) (\cite{acharya:petersen:philippon:richardson:2017}), Marginal Mean Excess (MME) (\cite{das:fasen:2018,das:fasen:2019}), Systemic risk (SRISK) (\cite{brownlees:engle:2016}) have become extremely popular in the years following the 2007-2009 crisis. In this paper, we focus on one particular measure of risk contagion, namely, CoVaR.  Naturally, computing such risks requires an appropriate modeling of joint tail risks which captures the marginal risk behavior as well as inter-dependencies between financial entities. 
%In the rest of this section, we define the risk measure of interest CoVaR, then set-up the bipartite network to model the financial system  and discuss the assumptions on distribution and dependence behavior of the underlying risk variables.

%\subsection{Contagion risk measure: CoVaR}\label{subsec:covar}
We properly define this particular measure of risk contagion, CoVaR, in the following.
For a   random variable $Y$, the \emph{Value-at-Risk} or $\VaR$ at level $1-\gamma\in (0,1)$ is defined as 
\begin{eqnarray}\label{def:var}
\VaR_{\gamma}(Y) := \inf\{y\in\R: \P(Y>y)\le \gamma\}=\inf\{y\in\R: \P(Y\le y) \geq 1-\gamma\}.
\end{eqnarray}
Note that we chose $\gamma$ as the subscript instead of the usual $1-\gamma$ for brevity and notational convenience. Now, given two   random variables $Y_1, Y_2,$ and  constant levels $1-\gamma_1,1-\gamma_2\in(0,1)$, we define the \emph{Co}nditional-VaR or \emph{Co}ntagion-VaR (CoVaR) of $Y_1$ (at level $1-\gamma_1$) given $Y_2$ (at level $1-\gamma_2$) as
    \begin{align}\label{def:covar}
           \text{CoVaR}_{\gamma_1| \gamma_2}(Y_1|Y_2)=\inf\{y\in\R:\P(Y_1>y|Y_2>\VaR_{\gamma_2}(Y_2))\leq \gamma_1\}.
    \end{align}
Here, $\text{CoVaR}_{\gamma_1| \gamma_2}(Y_1|Y_2)$ represents the VaR at level $1-\gamma_1\in(0,1)$ of $Y_1$ given that $Y_2$ is above its own  VaR at level $1-\gamma_2\in(0,1)$. Although we concentrate on CoVaR throughout the paper, it is clear that other conditional risk measures like MES/MME (\cite{acharya:petersen:philippon:richardson:2017,das:fasen:2018}) and SRISK (\cite{brownlees:engle:2016}) are viable options. %as well and may be explored in future work.

%\marginpar{\DAS{I am wondering if we should keep the other CoVaR examples here, i.e., with sum and max. But I agree it makes the intro bigger. Just thinking} }
 The risk measure CoVaR was introduced to capture risk contagion, as well as systemic risk by \cite{adrian:brunnermeier:2016} where they used the conditioning event to be $Y_2=\VaR_{\gamma_2}(Y_2)$; this was later modified by \cite{girardi:ergun:2013} to $Y_2>\VaR_{\gamma_2}(Y_2)$ with the restriction that $\gamma_1=\gamma_2$; this latter definition has been shown to have nicer properties for dependence modeling and is used widely for modeling and estimation (\cite{nolde:zhou:zhou:2022,Mainik:Schaaning,hardle:wang:yu:2016}). We will use the more general definition given in \eqref{def:covar} following \cite{girardi:ergun:2013}  but allowing $\gamma_1$ and $\gamma_2$ to not necessarily to be equal; see \cite{reboredo:ugolini:2015,kley:kluppelberg:reinert:2017,bianchietal:2023} for related work using this definition.

The computation of  CoVaR under a bipartite network framework has been addressed in \cite{kley:kluppelberg:reinert:2017} where the authors have particularly concentrated on distributions of risk exposures that are either asymptotically co-monotone or asymptotically independent. In addition, some of their proofs in the asymptotic independent case are based on the i.i.d. assumption, and more importantly, may often lead to null estimates for the relevant measures.
%The computation of  CoVaR under a bipartite network framework has been addressed in \cite{kley:kluppelberg:reinert:2017}; their work deals with underlying distributions of risk exposures which are either asymptotically comonotone or asymptotically independent but when addressed, the proofs are occasionally based on the i.i.d. assumption, and more importantly often leading to null estimates for the relevant measures.
%\marginpar{\DAS{I toned down the statement a bit, I think :)).}}
%Note that a computation of CoVaR under a bipartite set-up has been  addressed in \citet{kley:kluppelberg:reinert:2017} where the results assume asymptotically (fully) dependent risk exposures which for example a Gaussian dependence structure or i.i.d. variables would not satisfy; thus applied to an asymptotically independent framework the results on  \citet{kley:kluppelberg:reinert:2017} often lead to null estimates for the relevant measures. %Results on asymptotically independent models are not a key focus in \citet{kley:kluppelberg:reinert:2017}, and when addressed, the resulting conclusions are occasionally based on the i.i.d. assumption, and more importantly often leading to null estimates for the relevant measures.  %Thus, some related results in our paper capture a similar flavor as those of \citep{kley:kluppelberg:reinert:2017}, but 
In this paper, we provide an enhanced characterization of risk contagion using CoVaR, not only addressing the asymptotically strong dependent case (which includes the co-monotone case), but especially emphasizing on models with asymptotic independence which are natural and popular, and have not been particularly addressed before; here we use the characterization for tail asymptotics of random linear functions of regularly varying vectors developed in \cite{das:fasen:kluppelberg:2022}. 
%In our paper, we fill this gap and allow for a variety of strengths of dependence for modeling the underlying vectors with the emphasis on the asymptotic independent case. 
Our results show that the asymptotic behavior of CoVaR is affected by both the structure of the bipartite network, as well as the strength of dependence of the underlying distribution of risk factors.  %Next, we model the underlying distribution of $\bZ$ more precisely.

The paper is structured as follows. %In \Cref{sec:asyind}, we discuss the idea of \emph{asymptotic (tail) independence} and its effects on multivariate extremes. Since such a feature is often empirically observed, at least in a bivariate sense, we formalize a multivariate notion for this concept and use two parametric models for elaborating on the presence of this phenomenon, namely the Gaussian copula and the Marshall-Olkin copula. 
In \Cref{sec:mrvwithasyind}, we characterize multivariate heavy-tailed models and in particular investigate Gaussian  and Marshall-Olkin copula models with Pareto-type margins.
In \Cref{sec:systemicrisk}, we characterize the asymptotic behavior of $\CoVaR_{\gamma g(\gamma)|\gamma}(Y_1|Y_2)$ 
as $\gamma\downarrow 0$ for some function $g(\gamma)$ where $Y_1$ is heavy-tailed (note that both margins need not be heavy-tailed). We also introduce here the \emph{Extreme CoVaR Index} (ECI), which is a measure of the strength of risk contagion from $Y_1$ to $Y_2$ and is particularly useful when they are asymptotically independent. The bipartite network with asymptotically strongly dependent underlying assets is addressed in an example in this section as well.
%One particular example in this section is the asymptotic behavior of CoVaR for the bipartite network $\bX=\bA\bZ$ when the underlying assets are asymptotically strongly dependent. 
For the remainder of the paper, we investigate the more challenging case of asymptotically independent assets. %After having found tail probabilities of the underlying assets $\bZ$ exhibiting asymptotic independence in \Cref{sec:mrvwithasyind}, and the asymptotics of CoVaR in \Cref{sec:systemicrisk},  
In \Cref{sec:bipnet}, we compute joint and conditional probabilities for extreme events pertaining to the banks/agents given by $\bX=\bA\bZ$ and the asymptotic behavior of CoVaR as well as the ECI. This requires an appropriate understanding of transformations of various kinds of sets in the presence of the bipartite network model.  We conclude in \Cref{sec:concl} with indications for future work. All proofs of relevant results are relegated to the appendix.

%\marginpar{\VF{My point was more: the emphasis of random \emph{vector}. That we have to include tail equivalence for random \emph{variables} is clear.}}

\subsection*{Notations}\label{subsec:notations}
 The following notations are used throughout the paper. We denote by \linebreak $\mathbb{I}_d=\{1,\ldots,d\}$ an index set with $d$ elements, the subscript is dropped when  evident from the context.  
%For sequences $\bw_t=(w_{t,1},\ldots,w_{t,d}),\bz_t=(z_{t,1},\ldots,z_{t,d})\in\R^d$, $t>0$, the notation $\bz_t=\bw_t+o(t)$ as $t\to\infty$ means that $z_{t,j}-w_{t,j}=o(t)$ as $t\to\infty$ for every $j\in\mathbb{I}_d$.
%\marginpar{\VF{Do we use the notation that a random vector is tail equivalent?} \DAS{We have in Theorem 4.1.}
 For a given vector $\bz\in\R^d$ and $S\subseteq \mathbb{I}_d$, we denote by  $\bz^{\top}$ the transpose of $\bz$ and by $\bz_S\in\R^{|S|}$ the vector obtained by deleting the components of $\bz$ in $\mathbb{I}_d\backslash S$. Similarly, for non-empty $I,J\subseteq\mathbb{I}_d$, $\Sigma_{IJ}$ denotes the appropriate sub-matrix of a given matrix $\Sigma\in\R^{d\times d}$; and we write $\Sigma_I$ for $\Sigma_{II}$.
Vector operations are  understood component-wise, e.g., for vectors $\bv=(v_1,\ldots,v_d)^\top$ and $\bz=(z_1,\ldots,z_d)^\top$, $\bv\le \bz$ means $v_j\le z_j \forall j$. %The order statistics of $\bz=(z_1,\ldots,z_d)$ are written as $z_{(1)}\geq z_{(1)}\geq\ldots\geq z_{(d)}$.
We also have the following notations for vectors on $\R^d$: $\bzero=(0,\ldots,0)^\top, \bone=(1,\ldots,1)^\top,$ $\binfty=(\infty,\ldots,\infty)^\top$ and $e_j =(0,\ldots,1,\ldots,0)^{\top}$, $j\in\mathbb{I}_d$, where $e_j$ has only one non-zero entry 1 at the $j$-th co-ordinate. 

For a random vector  $\bZ=(Z_1,\ldots,Z_d)^\top$, we write $\bZ\sim F$ if $\bZ$ has distribution function $F$; moreover, we understand that marginally $Z_j\sim F_j$ for $j\in \mathbb{I}_d$. We call the random vector $\bZ\sim F$ to be \emph{tail equivalent} if for all $j,\ell \in \mathbb{I}_d, j\neq \ell$, we have $\lim_{t\to\infty} (1-F_j(t))/(1-F_\ell(t))=c_{j\ell}$ for some $c_{j\ell}>0$; moreover,  if $c_{j\ell}=1,  \forall\, j,\ell \in \mathbb{I}_d,$ then we say $\bZ$ is \emph{completely tail equivalent}. {Analogously we call the associated random variables $Z_j,Z_\ell$ tail equivalent or completely tail equivalent.} For functions $f,g$, we write $f(t)\sim g(t)$ as $t\to\infty$ if $\lim_{t\to\infty} f(t)/g(t)=1$. 
 The cardinality of a set $S\subseteq \mathbb{I}_d$ is denoted by $|S|$.
  The indicator function of an event $A$ is denoted by $\mathds{1}_{A}$ and for a constant $t>0$ and a set $A\subseteq \R_+^{d}$, we denote by $tA:= \{t\bz: \bz\in A\}$. 

\section{Multivariate heavy tails}\label{sec:mrvwithasyind}
%Financial returns are often empirically observed to be heavy-tailed in nature, 
%{Heavy tailed models in finance are widespread, see \citet{Adler:Taqqu,Rachev:2003}  for an overview.}
In this paper, the framework that we use for modeling heavy tails is that of multivariate regular variation.
 % which is quite amenable to the computation of tail probabilities in the presence of asymptotic independence. 
We start with a brief primer on multivariate regular variation, and then explicitly derive necessary model parameters, limit measures and their supports for our two primary model examples, the Gaussian and the Marshall-Olkin copula model.

%\subsection{Preliminaries} We recall key concepts from multivariate regular variation and copulas next.
\subsection{Preliminaries: multivariate regular variation}\label{subsec:mrv}
Regular variation is a popular theoretical framework for modeling heavy-tailed distributions; for the models in this paper, we assume this property for all our marginal risk variables.

A measurable function $f: \R_{+} \to \R_{+}$ is \emph{regularly varying} (at $+\infty$) with some fixed $\beta\in \R$ if $\lim_{t\to\infty} f(tz)/f(t) =z^{\beta}, \forall\,z>0$.  We write $f\in\RV_{\beta}$ and if $\beta=0$, then we call $f$ a \emph{slowly varying} function. A real-valued random variable $Z\sim F$  is regularly varying (at $+\infty$) if the tail   $\overline{F}:= 1-F \in \RV_{-\alpha}$ for some $\alpha> 0$. Alternatively, for $Z\sim F$, $\overline{F}\in \RV_{-\alpha}$ is equivalent to the existence of a measurable function $b:\R_+\to\R_+$ with $b(t)\to \infty$ as $t\to\infty$ such that
\[t\,\P(Z>b(t)z) = t\,\overline{F}(b(t)z) \xrightarrow{t \to \infty} z^{-\alpha}, \quad \forall\, z>0.\]
Consequently, $b\in \RV_{1/\alpha}$ and a canonical choice for $b$ is
$b(t)= F^{\leftarrow}(1-1/t)=\overline{F}^{\leftarrow}(1/t)$ where $F^{\leftarrow}(z)= \inf\{y\in\R:F(y)\ge z\}$ is the generalized inverse of $F$. Well-known distributions like Pareto, L\'evy, Fr\'ech\'et, Student's $t$ which are used to model heavy-tailed data, are all in fact regularly varying; see \cite{embrechts:kluppelberg:mikosch:1997,bingham:goldie:teugels:1989,resnickbook:2007,resnickbook:2008} for further details.

For multivariate risks, our interest is in high (positive) risk events, and hence we concentrate on characterizing tail behavior on the positive quadrant $\R^d_+:=[0,\infty)^d$; risk events in any other quadrant can be treated similarly. Multivariate regular variation on such \emph{cones} and their subsets are discussed in detail in \cite{hult:lindskog:2006a,das:mitra:resnick:2013,lindskog:resnick:roy:2014,das:fasen:2023}. We briefly introduce the ideas, notations and tools here. We focus particularly on subcones of $\R^d_+$ of the following form:
\begin{align*}
%\E^{(1)}_{d}  := &\;\R^{d}_+ \setminus \CA_d^{(0)} =  \R_+^{d} \setminus \{\bzero\} =   \{\bz\in \R^d_+: z_{(1)}>0\}, \label{def:ed1}\\
\E^{(i)}_{d}  := &\;  \R^d_+ \setminus \{\bz\in \R^d_+: z_{(i)}=0\} = \; \{\bz\in \R^d_+: z_{(i)}>0\}, \quad i\in\mathbb{I}_d, \label{def:edd}
\end{align*}
where $z_{(1)}\geq z_{(2)}\geq \ldots\geq z_{(d)}$ is the decreasing order statistic of $z_1,\ldots,z_d$. Here $\E_d^{(i)}$ represents the subspace of $\R^d_+$ with all $(i-1)$-dimensional co-ordinate hyperplanes removed; the $0$-dimensional hyperplane is the point $\{\bzero\}$.
 %We call the subsets $\E_d^{(i)}$ \emph{co-ordinate subcones} since they are cones obtained from $\R_+^d$ by removing particular co-ordinate hyperplanes.
% Here $\E^{(1)}_{d}$ is the {positive quadrant} with $\{\bzero\}$ removed, $\E^{(2)}_{d}$ is the {positive quadrant} with all one-dimensional co-ordinate axes removed, $\E^{(3)}_{d}$ is the {positive quadrant} with all two-dimensional co-ordinate hyperplanes removed, and so on.
Clearly,
\begin{equation} \label{eq2.1}
\E^{(1)}_{d} \supset \E^{(2)}_{d} \supset \ldots \supset \E^{(d)}_{d}.
\end{equation}
%Note that for any $1\le i \le d$, the set $\{\bz\in \R^d_+: z_{(i)}=0\}\subseteq \R_+^d$
%is a closed cone in $\R_+^d$ containing $\bzero$.
The type of convergence we use here is called $\M$-convergence of measures (\cite{hult:lindskog:2006a,das:mitra:resnick:2013}),
%. It can be discussed in the more general context of multivariate regular variation on subcones \linebreak $\R_+^d\setminus\C_0$ of $\R_+^d$, where $\C_0$ is a closed cone containing $\bzero$.
of which multivariate regular variation on $\E_d^{(i)}$ is only a special example; see \cite{lindskog:resnick:roy:2014, das:fasen:2023}.  For a subspace $\E_d^{(i)}$, let $\cB(\E_d^{(i)})$ denote the collection of Borel sets contained $\E_d^{(i)}$.  % and let $\mathbb{I}_d=\{1,\ldots,d\}$.
%More details can be found in \citet{hult:lindskog:2006a,das:mitra:resnick:2013,lindskog:resnick:roy:2014,das:fasen:2023}.
%Moreover, for our particular sets of interest of the form $A$ as  in \eqref{eq:tailset}, we have $A\subset \E_d^{(i)}$.

\begin{dfn}\label{def:mrv}
Let $i\in\mathbb{I}_d$. A random vector $\bZ \in \R^{d}$ is \emph{multivariate regularly varying} (MRV) on $\E^{(i)}_{d}$ if there exists a regularly varying function
$b_i \in \RV_{1/\alpha_i}$, $\alpha_i >0$,  and a non-null (Borel) measure $\mu_i$ which is finite on Borel sets bounded away from $\{\bz\in \R^d_+: z_{(i)}=0\}$
{such that}
\begin{equation}\label{eq:RegVarMeasmui}
    \lim_{t\to\infty}t\,\P\left( \frac{\bZ}{b_i(t)} \in B \right)= \mu_i(B)
\end{equation}
for all sets $B \in \cB(\E^{(i)}_{d})$ which are {\it bounded away from\/} $\{\bz\in \R^d_+: z_{(i)}=0\}$ with $\mu_i(\partial B)=0$. We write $\bZ \in \MRV(\alpha_i, b_i, \mu_i, \E^{(i)}_{d})$, and some parameters may be  dropped for convenience.
\end{dfn}
The limit measure $\mu_i$ defined in \eqref{eq:RegVarMeasmui} turns out to be  homogeneous of order $-\alpha_i$, i.e., $\mu_i(\lambda B)=\lambda^{-\alpha_i}\mu_i(B)$ for all $\lambda>0$. Moreover, if   $\bZ \in \MRV(\alpha_i, b_i, \mu_i, \E^{(i)}_{d})$, $\forall\, i\in\mathbb{I}_d$,
then a direct conclusion from \eqref{eq2.1} is that
\beam
    \alpha_1\leq \alpha_2\leq \cdots\leq\alpha_d
\eeam
implying that as $i\in \mathbb{I}_d$ increases, the rate of decay of probabilities of tail sets in $\E_d^{(i)}$ either remains the same or becomes faster.

\subsection{Regular variation  under Gaussian copula}\label{subsec:gcrv}
Probabilities of tail subsets of $\R^d_+$ for heavy-tailed models using the popular Gaussian dependence structure have been discussed in detail in \cite{das:fasen:2024}. Here, we concentrate on such models where the tails are asymptotically power-law (Pareto-like), this allows for explicit computations and insights into this model. %especially with respect to asymptotic independence behavior. 
First, to fix notations, if $\Phi_\Sigma$ denotes the distribution function of a $d$-variate  normal distribution with all marginal means zero, variances one and positive semi-definite correlation matrix $\Sigma\in\R^{d\times d}$, and $\Phi$ denotes a standard normal distribution function, then $$C_{\Sigma}(u_1,\ldots,u_d)= \Phi_\Sigma(\Phi^{-1}(u_1),\ldots, \Phi^{-1}(u_d)), \quad \quad 0< u_1,\ldots, u_d < 1,$$
denotes the Gaussian copula with correlation matrix $\Sigma$. Next, we define the following model analogous to the RVGC model defined in \cite{das:fasen:2024} {and analyze it's MRV behavior}.
\begin{definition}\label{def:pgc}
An $\R^d$-valued random vector $\bZ=(Z_1,\ldots,Z_d)^\top \sim F$ follows a \emph{Pareto-tailed distribution with Gaussian copula}, index $\alpha>0$, scaling parameter $\theta>0$ and correlation matrix $\Sigma$, if the following holds:
\begin{enumerate}[(i)]
    \item The  marginal distributions $F_j$ of $Z_j$ are continuous and strictly increasing  with tail $ \ov F_j(t):=1-F_j(t) \sim \theta t^{-\alpha}$,  $\forall\,j\in\mathbb{I}_d$ and some $\theta,\alpha>0$. 
  %  \item  The function $b:\R_+\to\R_+$ is defined as $b(t)=(\theta t)^{\frac{1}{\alpha}}$.
    \item The joint distribution function $F$ of $\bZ$ is given by  $$F(\bz) = C_{\Sigma}(F_1(z_1),\ldots,F_d(z_d)), \quad \bz=(z_1,\ldots,z_d)^\top\in\R^d,$$
      where $C_{\Sigma}$ is the Gaussian copula with positive-definite correlation matrix  $\Sigma\in\R^{d\times d}$.
\end{enumerate}
We write $\bZ \in \PGC(\alpha,\theta,\Sigma)$   where some parameters may be dropped for convenience.
\end{definition}

\begin{remark} The structure assumed allows for a wide variety of distributional behavior.
\begin{itemize}
    \item[(a)] The tails of Pareto, L\'evy, Student's $t$ and Fr\'ech\'et distributions satisfy $\ov F_j(t)\sim \theta t^{-\alpha}$   (\cite{embrechts:kluppelberg:mikosch:1997, nair:wierman:zwart:2015}) thus, allowing a few popular heavy-tailed marginals to be chosen from.
    \item[(b)] A special case here is when the correlation matrix $\Sigma=I_d$ and all marginals $Z_1,\ldots,Z_d$ are identically Pareto distributed, thus covering the case with i.i.d. marginals as well.
\item[(c)]   If the Gaussian dependence is defined by an \emph{equicorrelation} matrix given by
\begin{align}\label{def:sigmarho}
\Sigma_\rho & :=\begin{pmatrix}
1 & \;\;\; \rho  & \;\;\; \ldots & \;\;\; \rho\\
\rho & \;\;\; 1  & \;\;\; \ldots & \;\;\; \rho\\
\vdots & \;\;\; \vdots & \;\;\; \vdots & \;\;\; \vdots \\
 \rho & \;\;\; \ldots & \;\;\; \ldots & \;\;\;  1
\end{pmatrix}
\end{align}
with $-\frac{1}{d-1}<\rho <1$ (making $\Sigma_\rho$ positive definite),  we write $\bZ\in \PGC(\alpha, \theta,  \Sigma_{\rho})$. This correlation matrix is naturally the only choice if $d=2$.
\item[(d)]  If $\bZ \in \PGC(\alpha,\theta,\Sigma)$   then any sub-vector $\bZ_S $, $S\subseteq \mathbb{I}_d$, satisfies $\bZ_S \in \PGC(\alpha,\theta,\Sigma_S)$. %\marginpar{\VF{I deleted the proof of this comment. I think it should be clear}}
\end{itemize}
\end{remark}
For the $\PGC$ model, although we have only assumed a Pareto-like tail for the marginal distributions, it turns out that along with the Gaussian dependence, this suffices for the random vectors to admit \emph{multivariate regular variation} on the various subcones $\E_d^{(i)} \subset \R_+^d$; this was derived in further generality in \cite{das:fasen:2024}.
We recall and reformulate some of the related results. 
%and then provide exact tail asymptotics and some insights into the asymptotic independence behavior under this model.
%\subsubsection{Regular variation in P-GC models}
%\subsubsection{Pareto-tailed Gaussian copula models and regular variation on cones}
%The results in this section are adapted from \citet{das:fasen:2024}.
%e require some notation before we are able to present the results.
The following lemma is from \cite[Proposition 2.5 and Corollary 2.7]{hashorva:husler:2002} and we require that notation for the presentation of regular variation in P-GC models.
\begin{lemma}\label{lem:quadprog}
Let $\Sigma\in \R^{d\times d}$ be a positive-definite correlation matrix. Then the quadratic programming problem
\begin{align}\label{eq:quadprog}
  \mathcal{P}_{\Sigma^{-1}}: \min_{\{\bz\ge \bone\}} \bz^\top \Sigma^{-1}\bz
\end{align}
  has a unique solution $\be^*=\be^*(\Sigma)\in\R^d$ such that
    \begin{align*}
    \gamma:= \gamma(\Sigma):= \min_{\{\bz\ge \bone\}} \bz^\top \Sigma^{-1} \bz=\be^{*\top} \Sigma^{-1}\be^*>1.
\end{align*}
% If $\Sigma^{-1}\bone \geq \bzero$ then  $\be^*=\bone.$
Moreover, there exists a unique non-empty index set  $I:=I(\Sigma)\subseteq \{1,\ldots, d\}=:\mathbb{I}_d$ with \linebreak $J:=J(\Sigma):=\mathbb{I}_d\setminus I$ such that  the unique solution $\be^*$ is given by
$$\be^*_I=\bone_I, \qquad \text{ and } \qquad\be^*_J=-[\Sigma^{-1}]_{JJ}^{-1}[\Sigma^{-1}]_{JI}\bone_I=\Sigma_{JI}(\Sigma_I)^{-1}\bone_I\geq \bone_J,$$
and $ \bone_{I}\Sigma_{I}^{-1}\bone_I=\be^{*\top} \Sigma^{-1}\be^*=\gamma>1$ as well as $\bz^{\top}\Sigma^{-1}\be^*=\bz_I^{\top}\Sigma_{I}^{-1}\bone_I $ for any $\bz\in\R^d$. 
Also defining   $h_i:=h_i(\Sigma):=e_i^\top \Sigma^{-1}_{I}\bone_I>0$ for $i\in \mathbb{I}_d$, we have $h_i>0$ for $i\in I$.  If $\Sigma^{-1}\bone \geq \bzero$ then  $2\le |I| \le d$ and $\be^*=\bone.$
%\end{enumerate}
\end{lemma}
With the notations and definitions of \Cref{lem:quadprog}, we can state the main result on MRV for $\PGC$ models which summarizes \cite[Theorems 3.1 and 3.4]{das:fasen:2024}.
\begin{comment}
    \begin{lemma}\label{lem:qp}
Let $\Sigma\in \R^{d\times d}$ be a positive-definite correlation matrix. Then the quadratic programming problem
\begin{align}\label{eq:quadprog}
  \mathcal{P}_{\Sigma^{-1}}: \min_{\{\bz\ge \bone\}} \bz^\top \Sigma^{-1}\bz
\end{align}
  has a unique solution $\be^*=\be^*(\Sigma)\in\R^d$ such that
    \begin{align}
    \gamma:= \gamma(\Sigma):= \min_{\{\bz\ge \bone\}} \bz^\top \Sigma^{-1} \bz=\be^{*\top} \Sigma^{-1}\be^*>1.
\end{align}
Moreover, the following hold:
\begin{enumerate}[(a)]
\item If $\Sigma^{-1}\bone \geq \bzero$ then  $\be^*=\bone.$

\item There exists a unique non-empty index set  $I:=I(\Sigma)\subseteq \{1,\ldots, d\}=:\mathbb{I}_d$ with $J:=J(\Sigma):=\mathbb{I}_d\setminus I$ such that  the unique solution $\be^*$ is given by
$$\be^*_I=\bone_I, \qquad \text{ and } \qquad\be^*_J=-[\Sigma^{-1}]_{JJ}^{-1}[\Sigma^{-1}]_{JI}\bone_I=\Sigma_{JI}(\Sigma_I)^{-1}\bone_I\geq \bone_J,$$
and moreover,
$$ \bone_{I}\Sigma_{I}^{-1}\bone_I=\be^{*\top} \Sigma^{-1}\be^*=\gamma>1.$$
 Finally,   $h_i:=h_i(\Sigma):=e_i^\top \Sigma^{-1}_{I}\bone_I>0$ for $i\in I$ and for any $\bz\in\R^d$ the following equality holds: $$\bz^{\top}\Sigma^{-1}\be^*=\bz_I^{\top}\Sigma_{I}^{-1}\bone_I.$$
\end{enumerate}
\end{lemma}
\Cref{lem:qp}  is taken from \citet[Proposition~2.1]{hashorva:2005}  and \citet[Proposition 2.5 and Corollary 2.7]{hashorva:husler:2002}.
Note that, if  $\Sigma^{-1}\bone \geq \bzero$ then $|I|\geq 2$ but it is not necessarily that $|I|=d$.
\end{comment}

%\marginpar{\VF{I changed Proposition to Theorem}}
\begin{proposition}\label{prop:premain}
    Let  $\bZ \in \PGC(\alpha, \theta, \Sigma)$ where $\Sigma$ is positive definite. Fix a non-empty set $S\subseteq \mathbb{I}_d$ with $|S|\ge 2$.
    \begin{itemize}
        \item[(i)] Let $\gamma_S:=\gamma(\Sigma_S)$, $I_S:=I(\Sigma_S)$, $\be_S^*:=\be^*(\Sigma_S)$, and  $h_{s}^S:=h_s(\Sigma_S), s\in I_S$, be defined as in  \Cref{lem:quadprog}.
        \item[(ii)]   Let  $J_S:=\mathbb{I}_d\setminus I_S$. Define
 $\bY_{J_S}\sim \mathcal{N}(\bzero_{J_S}, \Sigma_{J_S}-\Sigma_{J_SI_S}\Sigma_{I_S}^{-1}\Sigma_{I_SJ_S})$ if $J_S\not=\emptyset,$ and \linebreak $\bY_{J_S}=\bzero_{J_S}$ if $J_S=\emptyset$.
 \item[(iii)] If $J_S\not=\emptyset$ define $\bl_{S}:=\lim_{t\to\infty}t(\bone_{J_S}-\be_{J_S}^*)$, a vector in $\R^{|J_S|}$ with components either  $0$ or $-\infty$,  and if  $J_S=\emptyset$ define $l_S:=\bzero_S$.
    \end{itemize}
 Let  $ \Gamma_{\bz_S} = \{\bv\in \R_+^d: v_s > z_s, \forall s\in S\}$ for $\bz_S=(z_s)_{s\in S}$  with $z_s>0, \forall s\in S$. Then as $t\to\infty$,
    \begin{align*} %\label{eq:mainlimitpgcAzs}
        \P(\bZ\in t \Gamma_{\bz_S}) = (1+o(1))\Upsilon_S({2\pi})^{\frac{\gamma_S}2} \theta^{\gamma_S}t^{-\alpha\gamma_S} (2\alpha\log (t))^{\frac{\gamma_S-|I_S|}{2}}  \prod_{s\in I_S} z_s^{-\alpha h_s^S},
    \end{align*}
    where $\Upsilon_S=    (2\pi)^{-|I_S|/2}| \Sigma_{I_S}|^{-1/2}\prod_{s\in S}(h_s^S)^{-1}  \P(\bY_{J_S}\geq \bl_S).$ Moreover, the following holds.
   % \begin{align*}%\label{eq:kappaSigma}
   %     \Upsilon_S=
  %      \frac{\P(\bY_{J_S}\geq \bl_S)}{(2\pi)^{|I_S|/2}|     \Sigma_{I_S}|^{1/2}\prod_{s\in S}h_s^S }.
   % \end{align*}
    % Also define $b(t)=(\theta t)^{1/\alpha}, t>0$.
    \begin{itemize}
  \item[(a)] We have $\bZ\in \MRV(\alpha_1,b_1,\mu_1,\E_d^{(1)})$ with $\alpha_1=\alpha, b_1(t)=(\theta t)^{1/\alpha}$, $t>0$ and \linebreak
     $\mu_1([\bzero,\bz]^c) = \sum_{j=1}^d z_{j}^{-\alpha},$ $\forall\; \bz\in \R_+^d. $
\item[(b)] Let $2\le i \le d$.
Define
\begin{align*}
   \mathcal{S}_i&:=\left\{ S\subseteq \mathbb{I}_d: \; |S|\geq i, \bone_{I_S}\Sigma^{-1}_{I_S}\bone_{I_S}=\min_{\widetilde S\subseteq \mathbb{I}_d, |\widetilde S|\ge i} \bone_{I_{\widetilde S}}\Sigma^{-1}_{I_{\widetilde S}}\bone_{I_{\widetilde S}}\right\},\\
    %=\{S_1,\ldots,S_{m_i}\},\\
    I_i&:=\arg \min_{S\in\mathcal{S}_i}|I_{S}|,
\end{align*}
where $I_i$ is not necessarily unique.
Then $\bZ\in \MRV(\alpha_i,b_i,\mu_i,\E_d^{(i)})$ where
 \begin{align}
    \gamma_i&=\gamma(\Sigma_{I_i})=
    \bone_{I_i}^{\top}\Sigma_{I_i}^{-1}\bone_{I_i}=\min_{S\subseteq \mathbb{I}_d, |S| \ge i} \min_{\bz_S \ge \bone_S} \bz_S^{\top} \Sigma_S^{-1} \bz_S, \nonumber\\
   \alpha_i&= %\alpha \cdot =\alpha \cdot \min_{S\subseteq \mathbb{I}_d, |S|\ge i} \bone_{I_S}^{\top} \Sigma_{I_S}^{-1} \bone_{I_S} =
   \alpha\gamma_i, \quad\quad   b_i^{\leftarrow}(t) = ({2\pi})^{-\frac{\gamma_i}2}\theta^{-\gamma_i}t^{\alpha\gamma_i} (2\alpha\log (t))^{\frac{|I_i|-\gamma_i}{2}}, \nonumber\\
%\end{align*}
% For sets $\Gamma_{\bz_S}=\{\bv\in \R^d_+: v_s>z_s, \forall\,s\in S\}$ with $z_s>0, \forall s\in S\subseteq \mathbb{I}_d, \,|S|\ge i$, we have
%\begin{align}
\label{def:muiAzs}
    \mu_i(\Gamma_{\bz_S})&= \begin{cases}
                   \Upsilon_{S}  \prod_{s\in I_S} z_s^{-\alpha h_s^S},
    & \text{ if }  S\in \mathcal{S}_i \text{ and } |I_S|=|I_i|, \\
                   0,  &\text{otherwise}. \end{cases}
\end{align}
\end{itemize}
\end{proposition}

%\begin{remark}
%    In \citet{das:fasen:2024}, conclusions similar to \Cref{prop:premain} hav been derived with weaker assumptions on the marginal random variables; in this paper we restrict to Pareto-like tails for explicit computations and cleaner exposition for characterizing asymptotic independence.
%\end{remark}
%\citet{das:fasen:2024} proved the theorem in the more general context where $Z_j$ is allowed to have a regularly varying tail with weak assumption on the slowly varying function. For that the reason the results of the present paper hold as well under their setup. However, for the ease of notation we restrict to the P-GC model.

First, we note that in this model, the parameters or indices of regular variation in each subspace $\E_d^{(i)}$ are indeed all different. The proof of this result and the following results of this section are given in \Cref{se2:proofs}.

\begin{proposition} \label{prop:bi}
  Let $\bZ\in \PGC(\alpha, \theta, \Sigma)$ where $\Sigma$ is positive definite and the assumptions and notations of \Cref{prop:premain} hold.
    Then, with $\bZ\in \MRV(\alpha_i,b_i,\mu_i,\E_d^{(i)}), \,\forall\, i\in \mathbb{I}_d$, we have $$\alpha_1 < \alpha_2 < \cdots < \alpha_d$$ and in particular, $b_i(t)/b_{i+1}(t)\to \infty$ as $t\to\infty$ for $i=1,\ldots,d-1$.
\end{proposition}

Thus, the rate of convergence of tail sets in $\E_d^{(i)}$ for different $i\in\mathbb{I}_d$ are different.
Next, we investigate the support of the limit measure $\mu_i$. %\DAS{The support of the measure $\mu_i$ is not necessarily $\E_d^{(i)}$, as we see in the next example, and this may have an effect on the regularly varying parameter on the \VF{cones}, see \citet{das:fasen:2024}.} \VF{I am not sure about this comment: Which cone? $\E_d^{(i)}$ is not fitting in my opinion. \DAS{ok, I change }}

\begin{example}~ \label{Example 2.7}
 For any positive-definite correlation matrix $\Sigma = (\rho_{j\ell})$, we have $$\bone_{\{j,\ell\}}^{\top}\Sigma^{-1}_{\{j,\ell\}}\bone_{\{j,\ell\}}=\frac{2}{1+\rho_{j \ell }}, \quad j\neq \ell \in \mathbb{I}_d,$$
 and the right-hand side above will not have the same value for all $j\not=\ell\in \mathbb{I}_d$ unless $\Sigma$ is an equicorrelation matrix. Hence, if $\Sigma$ is not an equicorrelation matrix, using notation from \Cref{prop:premain}, there exists a set $S=\{j^*, \ell^*\}\notin \mathcal{S}_2$, and, $$\mu_{2}(\{\bz\in\R_+^d:z_j>0\,\forall\, j\in  S, z_\ell=0\,\forall\, \ell\in\mathbb{I}_d\backslash  S\})=0.$$   
%see as well \Cref{rem:nullgauss} for an example.
\end{example}

In the following, we present a general result characterizing the support of the limit measure on each of the Euclidean subcones for the $\PGC$ model.
%Finally, we are able to characterize the support of $\mu_i$.

\begin{proposition}\label{prop:supportpgc}
  Let the assumptions and notations of \Cref{prop:premain} hold and fix some $i\in \mathbb{I}_d$.
\begin{itemize}
    \item[(a)] Suppose that for all $S\subseteq \mathbb{I}_d$ with $|S|=i$ we have $\Sigma^{-1}_S\bone_S >\bzero_S$ and $\bone_S^{\top}\Sigma^{-1}_S\bone_S=\gamma_i$. Then the support of the limit measure $\mu_i$ on $\E_d^{(i)}$ as defined in \eqref{def:muiAzs} is
    $$\bigcup_{\genfrac{}{}{0pt}{}{S\subset\mathbb{I}_d}{|S|=i}}\{\bz\in\R_+^d:z_j>0\,\forall\, j\in S, z_\ell=0\,\forall\, \ell\in\mathbb{I}_d\backslash S\}.$$ 
    %and, in fact, $\mu_i(\{\bz\in\R_+^d:z_j>0\,\forall\, j\in S, z_l=0\,\forall\, l\in\mathbb{I}_d\backslash S\})>0$ for all $S\subseteq \mathbb{I}_d$ with $|S|=i$.
    \item[(b)] Suppose that for  some $S\subseteq \mathbb{I}_d$ with $|S|=i$,  $\Sigma^{-1}_S\bone_S >\bzero_S$ and $\gamma_i\not=\bone_S^{\top}\Sigma^{-1}_S\bone_S$.   Then  $$\mu_{i}(\{\bz\in\R_+^d:z_j>0\,\forall\, j\in  S, z_\ell=0\,\forall\, \ell\in\mathbb{I}_d\backslash  S\})=0.$$
 \end{itemize}
% In particular, in both cases 
% $\alpha_1<\alpha_2<\cdots<\alpha_d$.
\end{proposition}
%\marginpar{\VF{I have changed the proposition, but I am also fine to skip it}}

%The proof is given in \Cref{se2:proofs}.

\subsection{Regular variation under Marshall-Olkin copula}

Another important example of multivariate tail risk modeling is the Marshall-Olkin copula. The Marshall-Olkin distribution is often used in reliability theory to capture the dependence between the failure of subsystems in an entire system and hence, is a candidate model for measuring systemic risk. We consider a particular type of Marshall-Olkin survival copula; cf. \cite{lin:li:2014} and \cite[Example 2.14]{das:fasen:2023}. Assume that for every non-empty set $S\subseteq \mathbb{I}_d$  there exists a parameter $\lambda_S>0$ and $\Lambda:=\{\lambda_S: \emptyset\neq S\subseteq\mathbb{I}_d\}$. Then the generalized Marshall-Olkin survival copula with rate parameter $\Lambda$ is given by
\begin{align}\label{eq:scop:mo}
   \widehat{C}^{\MO}_{\Lambda}(u_1,\ldots,u_d) = \prod_{i=1}^d \prod_{|S|=i} \bigwedge_{j\in S} u_j^{\eta_j^S}, \quad 0 < u_j <1,
\end{align}
where
\begin{align}\label{eq:MO:eta}
    \eta_j^S = \frac{\lambda_S}{\sum\limits_{J \supseteq \{j\} } \lambda_J}, \quad j\in S \subseteq \mathbb{I}_d.
\end{align}
Similar to the $\PGC$ model, we define a Pareto-Marshall-Olkin copula ($\PMOC$) model next.

\begin{definition}\label{def:pmoc}
An $\R^d$-valued random vector $\bZ=(Z_1,\ldots,Z_d)^\top \sim F$ follows a \emph{Pareto-tailed distribution with Marshall-Olkin copula} with index $\alpha>0$, scaling parameter $\theta>0$ and rate parameters $\Lambda= \{\lambda_S: \emptyset \neq S\subseteq\mathbb{I}_d\}$, if the following holds:
\begin{enumerate}[(i)]
    \item The  marginal distributions $F_j$ of $Z_j$ are continuous and strictly increasing  with tail $ \ov F_j(t):=1-F_j(t) \sim \theta t^{-\alpha}$,  $\forall\,j\in\mathbb{I}_d$ and some $\theta,\alpha>0$.
  %  \item  The function $b:\R_+\to\R_+$ is defined as $b(t)=(\theta t)^{\frac{1}{\alpha}}$.
   % \item The joint distribution function $F$ of $\bZ$ is given by  $$F(\bz) = C^{\MO}_{\Lambda}(F_1(z_1),\ldots,F_d(z_d)), \quad \bz=(z_1,\ldots,z_d)^\top\in\R^d,$$  where $C^{\MO}_{\Lambda}$ denotes the Marshall-Olkin copula with rate parameter set $\Lambda$.
     \item {The joint survival distribution function $\ov{F}$ of $\bZ$ is given by $$\overline{F}(\bz) = \P(Z_1>z_,\ldots,Z_d>z_d)= \widehat{C}^{\MO}_{\Lambda}(\ov{F}_1(z_1),\ldots,\ov{F}_d(z_d)), \quad \bz=(z_1,\ldots,z_d)^\top\in\R^d,$$
      where $\widehat{C}^{\MO}_{\Lambda}$ denotes the  Marshall-Olkin survival copula with rate parameter set $\Lambda$.}
\end{enumerate}
We write $\bZ \in \PMOC(\alpha,\theta,\Lambda)$   where some parameters may be dropped for convenience.
\end{definition}

It is also possible to show multivariate regular variation for any $\bZ \in \PMOC(\alpha,\theta,\Lambda)$  on the cones $\E_d^{(i)}$, but finding the exact parameters and limit measures requires some involved combinatorial computations; hence, we concentrate on two specific choices of $\Lambda$, cf. \cite[Example 2.14]{das:fasen:2023}:
\begin{enumerate}[(a)]
    \item Equal parameter for all sets: Here, $\lambda_S=\lambda$ for all non-empty sets $S\subseteq\mathbb{I}_d$ where $\lambda>0$. We denote this model by $\PMOC(\alpha,\theta,\lambda^{=})$.
    \item Parameters proportional to the cardinality of the sets: Here, $\lambda_S=|S|\lambda$ for all non-empty sets $S\subseteq\mathbb{I}_d$ where $\lambda>0$.  We denote this model by $\PMOC(\alpha,\theta,\lambda^{\propto})$.
    \end{enumerate}
    
    Note that in both of the cases the Marshall-Olkin copula and hence, the $\PMOC$ model do not depend on the value of $\lambda$.
\begin{remark}
If $\bZ\in \PMOC(\alpha, \theta, \lambda^{=})$ then any sub-vector $\bZ_S$ with $S\subseteq\mathbb{I}_d$ also satisfies  $\bZ_S\in \PMOC(\alpha, \theta, \lambda^{=})$ implying a nested structure across dimensions. However, in the case of a Marshall-Olkin copula with proportional parameters, i.e., $\bZ\in \PMOC(\alpha, \theta, \lambda^{\propto})$,   such a nested property does not hold anymore.
\end{remark}
    
    In each of these cases, we can explicitly compute all the relevant parameters of the multivariate regular variation, and in fact, the regular variation limit measures have positive mass on all feasible support regions in these cases. The result given next is adapted from \cite[Example 2.14]{das:fasen:2023} and characterizes multivariate regular variation for the $\PMOC$ models for the choices of \emph{equal rate parameters} and \emph{proportional rate parameters}. These are also known as Caudras-Auge copulas \cite{cuadras:auge:1981} and have been used in L\'evy frailty models for survival analysis.

\begin{proposition}\label{prop:pmocmrv} The following statements hold.
\begin{enumerate}[(i)]
    \item Let $\bZ\in \PMOC(\alpha, \theta, \lambda^{=})$. Then $\bZ\in \MRV(\alpha_i, b_i, \mu_i, \E_d^{(i)})$ for $i\in \mathbb{I}_d$ where
    \begin{align*}
  & \alpha_i =(2-2^{-(i-1)})\alpha, \quad \quad \quad {b_i(t) = \theta^{\frac{1}{\alpha}} t^{\frac{1}{\alpha_i}}}, 
  \intertext{and for sets $\Gamma_{\bz_S}^{(d)}=\{\bv\in \R^d_+: v_s>z_s, \forall\,s\in S\}$ with $z_s>0, \forall s\in S\subseteq \mathbb{I}_d, \,|S|\ge i$, we have}
    &\mu_i(\Gamma_{\bz_S}^{(d)})  = \begin{cases}
                  \prod\limits_{j=1}^i \left({{z}_{(j)}}\right)^{-\alpha2^{-(j-1)}},
    & \text{ if }  |S|=i, \\
                   0,  &\text{otherwise}, \end{cases}
\end{align*}
where ${z}_{(1)}\ge \ldots \ge {z}_{(i)}$ denote the decreasing order statistic of $(z_j)_{j\in S}$.
 \item Let $\bZ\in \PMOC(\alpha, \theta, \lambda^{\propto})$. Then $\bZ\in \MRV(\alpha_i, b_i, \mu_i, \E_d^{(i)})$ for $i\in \mathbb{I}_d$ where 
    \begin{align*}
  & \alpha_i =\frac{\alpha}{d+1}\left(2d-\frac{d-i}{2^{i-1}}\right),%\alpha_i \frac d{d+1} +   \frac{i\alpha}{(d+1)2^{{i-1}}}},
  \quad  \quad \quad {b_i(t) = \theta^{\frac{1}{\alpha}} t^{\frac{1}{\alpha_i}}}, \quad 
  \intertext{and for sets $\Gamma_{\bz_S}^{(d)}=\{\bv\in \R^d_+: v_s>z_s, \forall\,s\in S\}$ with $z_s>0, \forall s\in S\subseteq \mathbb{I}_d, \,|S|\ge i$, we have}
    &\mu_i(\Gamma_{\bz_S}^{(d)})  = \begin{cases}
                   \prod\limits_{j=1}^i \left({z}_{(j)}\right)^{-\alpha\left(1-\frac{j-1}{d+1}\right)2^{-(j-1)}},
    & \text{ if }  |S|=i, \\
                   0,  &\text{otherwise}. \end{cases}
\end{align*}
%where ${z}_{(1)}\ge \ldots \ge {z}_{(i)}$ denote the decreasing order statistic of $(z_j)_{j\in S}$.
\end{enumerate}
{In both cases we have $\alpha_1 < \alpha_2 < \cdots < \alpha_d$ and in particular, $b_i(t)/b_{i+1}(t)\to \infty$ as $t\to\infty$ for $i=1,\ldots,d-1$.}
\end{proposition}

It is easy to check that if $\bZ\in\PMOC(\alpha,\theta,\Lambda)$, where the rate parameters are either \emph{all equal} or are \emph{proportional} (to the size of the sets) as in \Cref{prop:pmocmrv}, and the marginals are identically distributed, then $\bZ$ is an exchangeable random vector (\cite{durrett:2010}). This justifies the fact that in these two cases $\mu_i$ actually puts positive mass on $$\{\bz\in\R_+^d:z_j>0\,\forall\, j\in S, z_\ell=0\,\,\forall\,\, \ell\in\mathbb{I}_d\backslash S\}$$
for all $S\subseteq\mathbb{I}_d$ with $|S|=i$ for any fixed $i\in\mathbb{I}_d$.

%\DAS{(Another section commenting on possible multivariate models? Archimedean copula?)}
%\marginpar{\DAS{Do we need to write more?}\VF{I think this is sufficient}}

\section{Measuring CoVaR}\label{sec:systemicrisk}

{Risk contagion} is often assessed using conditional measures of risk, and in this regard, CoVaR has turned out to be both reasonable and popular; cf. \Cref{sec:intro}, also see \cite{kley:kluppelberg:reinert:2017,bianchietal:2023} for computations of CoVaR in various set-ups. By definition it measures the effect of severe stress of one risk factor, say $Y_2$, on the risk behavior of another factor, say $Y_1$. To facilitate the computation of CoVaR in the bipartite network set-up as described previously, first, we provide a result on its asymptotic behavior for a bivariate random vector assuming appropriate multivariate tail behavior, i.e., multivariate regular variation. Proofs of the results in this section are given in \Cref{sec:proofofsec5}.

\begin{theorem} \label{thm:bivCoVaRmain}
 {Let   $\bY=(Y_1, Y_2)^\top$ be a bivariate random vector with margins $F_1$ and $F_2$, respectively. Suppose
$\bY':=(Y_1,F_1^{\leftarrow}\circ F_2(Y_2))^\top\in\MRV(\alpha_i,b_i,\mu_i',\E_2^{(i)})$ for  $i=1,2$. }  Define the functions $h, h_{\gamma}:(0,\infty)\to(0,\infty)$ as
\begin{align*}
 h(y) & :=\mu_2'(\left(y,\infty\right)\times(1,\infty)),\\
    h_{\gamma}(y) &:= \gamma b_2^{\leftarrow}(\VaR_{\gamma}(Y_1))\P\left(Y_1>y\VaR_{\gamma}(Y_1)|Y_2>\VaR_{\gamma}(Y_1))\right)
\end{align*}
for fixed $0<\gamma<1$.  By the MRV assumptions  we have $h_{\gamma}(y)\to h(y)$  as $\gamma\downarrow 0$ for $y\in (0,\infty)$ and $h$ is decreasing with $\lim_{y\uparrow \infty}h(y)=0$. Further, assume the following:
\begin{enumerate}[(i)]
    \item Let $\lim_{y\downarrow 0}h(y)=r\in\left(0,\infty\right]$ and for some $l\geq 0$, we have  $h:(l,\infty)\to(0,r)$
is strictly decreasing and continuous with inverse $h^{-1}$.
\item Let $g:(0,1)\to(0,\infty)$ be a measurable function with $g(\gamma)\gamma b_2^{\leftarrow}(\VaR_{\gamma}(Y_1))\in(0,r)$.
\end{enumerate}
%Let $\lim_{y\downarrow 0}h(y)=r\in\left(0,\infty\right]$ and suppose $h:(l,\infty)\to(0,r)$
%is strictly decreasing and continuous with inverse $h^{-1}$ for some $l\geq 0$. Finally, let $g:(0,1)\to(0,1)$ be a decreasing map
%
%Suppose $h$ is strictly decreasing and continuous  with $\lim_{y\downarrow 0}h(y)=a\in\left(0,\infty\right]$ \marginpar{\VF{$a$ is not necessarily $\infty$, e.g. strongly dependent case}} and 
%$\lim_{y\downarrow \infty}h(y)=0$  such that $h:(0,\infty)\to(0,a)$
%is bijective with inverse $h^{-1}$. 
% Let $g:(0,\infty)\to(0,\infty)$ be a map satisfying 
%with $g(\gamma)\gamma b_2^{\leftarrow}(b_1(\gamma^{-1}))\in(0,r)$. 
Now let one of the following conditions hold:
\begin{itemize}
    \item[(a)] Either,
        \begin{align}\label{cond:bivcova}
            0<\liminf_{\gamma\downarrow 0}g(\gamma)\gamma b_2^{\leftarrow}(\VaR_{\gamma}(Y_1))\leq \limsup_{\gamma\downarrow 0}g(\gamma)\gamma b_2^{\leftarrow}(\VaR_{\gamma}(Y_1))< r,
        \end{align}
       
    \item[(b)] or, 
        $$\lim_{\gamma\downarrow 0}g(\gamma)\gamma b_2^{\leftarrow}(\VaR_{\gamma}(Y_1))=0,$$ and $h_{\gamma}^{\leftarrow}(v)/h^{-1}(v)\to 1$ uniformly on $\left(0,R\right]$ as $\gamma\downarrow 0$ for some $0<R<r$,
    \item[(c)] or,  
        $$\lim_{\gamma\downarrow 0}g(\gamma)\gamma b_2^{\leftarrow}(\VaR_{\gamma}(Y_1))=r,$$ and $h_{\gamma}^{\leftarrow}(v)/h^{-1}(v)\to 1$ uniformly on $\left[L,r\right)$ as $\gamma\downarrow 0$ for some $0<L<r$. 
       
 %       and in particular, if $\alpha_2<2\alpha_1$
 %       \begin{eqnarray*}
 %           \text{CoVaR}_{\upsilon \gamma,\gamma}(Y_1|Y_2)\sim  b_1(\gamma^{-1})h^{\leftarrow}(\upsilon \gamma^2 b_2^{\leftarrow}(b_1(\gamma^{-1}))), \quad \text{as }\gamma\downarrow 0.
 %       \end{eqnarray*}
\end{itemize}
    %we additionally assume that for any $a>0$,
%\begin{eqnarray}\label{eq:hggbyghunif}
%   \lim_{\gamma\downarrow 0} \sup_{x\ge a} \Bigg| \frac{h_{\gamma}(x)}{h(x)} - 1 \Bigg| = 0. 
%\end{eqnarray}
% Also assume that $h$ is continuous with inverse given by $h^{\leftarrow}$ and
%\begin{eqnarray}\label{star}
%    \lim_{\delta\downarrow 0}\limsup_{u\to 0}
%    \frac{h^{\leftarrow}(u)}{h^{\leftarrow}(u(1+\delta))}\leq 1.
%\end{eqnarray}
%and suppose $f^{\leftarrow}\in\mathcal{R}_{0,-\beta}$ for some $\beta >0$. \VF{In $0$ regularly varying} 

%Moreover, $\lim_{\gamma\downarrow 0} \text{CoVaR}_{\upsilon\gamma,\gamma}(Y_1|Y_2)=\infty$.
 Then for any $0<\upsilon<1$,
        \begin{eqnarray*}
            \text{CoVaR}_{\upsilon g(\gamma) |\gamma}(Y_1|Y_2)\sim  \VaR_{\gamma}(Y_1)h^{-1}(\upsilon g(\gamma)\gamma b_2^{\leftarrow}(\VaR_{\gamma}(Y_1))), \quad \gamma\downarrow 0.
        \end{eqnarray*}
%and equivalently,
%    \VF{\begin{eqnarray*}
  %          \text{CoVaR}_{\upsilon g(\gamma) |\gamma}(Y_1|Y_2)\sim  \VaR_{\gamma}(Y_1)\frac{\VaR_\gamma(Y_2)}{\VaR_\gamma(Y_1)}h^{-1}(\upsilon g(\gamma)\gamma b_2^{\leftarrow}(\VaR_{\gamma}(Y_1))), \quad \gamma\downarrow 0.
  %      \end{eqnarray*}
%\VF{In the case that $Y_1$ and $Y_2$ are tail-equivalent this reduces to
%\begin{eqnarray} \label{eq 4.2}
%            \text{CoVaR}_{\upsilon g(\gamma) |\gamma}%(Y_1|Y_2)\sim  \VaR_{\gamma}(Y_1)
 %           \left(\frac{\mu_1(\left[1,\infty\right)\times\R_+)}{\mu_1(\R_+\times \left[1,\infty\right))}\right)^{\frac{1}{\alpha}}
  %          h^{-1}(\upsilon g(\gamma)\gamma b_2^{\leftarrow}(\VaR_{\gamma}(Y_1)))  \nonumber\\
   %     \end{eqnarray}
    %    as $\gamma\downarrow 0$.}
        
\end{theorem}
\begin{comment}
\begin{remark}
    \begin{itemize}
        \item[(a)] (Strongly dependent case) Suppose $b_1=b_2$ and $\mu_1=\mu_2$. Then $a=\sup h(y)=1$ and $h_{\gamma}^{\leftarrow}(y)\to h^{-1}(y)$ uniformly on $(0,1)$.  Thus for $g(\gamma)\eqd 1$ and $g(\gamma)=\gamma$, the statement of 
        \Cref{thm:bivCoVaRmain} is valid.
        \item[(b)] (Asymptotic independent case) Suppose $g(\gamma)\eqd 1$ and $\alpha_2>\alpha_1$. Then $a=\infty$ and
        $g(\gamma)\gamma b_2^{\leftarrow}(b_1(\gamma^{-1}))\to \infty$ as $\gamma\downarrow 0$ such that we can apply  \Cref{thm:bivCoVaRmain}(a).
        \item[(c)] (Independent case) Suppose $Y_1,Y_2$ are independent and $Y_1$ is Pareto($\alpha$). Then $a=\infty$ and $h_\gamma^{\leftarrow}(x)=h^{-1}(x)$ such that for any function $g$ the assumption of \Cref{thm:bivCoVaRmain} are satisfied.
        \item[(d)] Mashall-Olkin: I think $h_\gamma(x)/h(x)=h_\gamma(1)/h(1)$ is independent of $x$ such that $h_\gamma(x)=h(x)h_\gamma(1)/h(1)$ is converging uniformly.
    \end{itemize}
\end{remark}
\end{comment}

\begin{remark}~ \label{Remark 3.2} A few remarks below may aid in understanding the assumptions and consequences of the above result.
\begin{enumerate}[(a)]
    \item If $F_2$ is continuous, then $F_2(Y_2)$ is uniform  on $(0,1)$. Hence the tail behavior of $Y_2$ has no influence on the asymptotic behavior of $\text{CoVaR}_{\upsilon g(\gamma) |\gamma}(Y_1|Y_2)$, as is expected.
    \item %\VF{If $Y_1$ and $Y_2$ are not tail-equivalent either $\mu_1(\R_+\times \left[1,\infty\right))$ or $\mu_1(\left[1,\infty\right)\times \R_+)$ are zero and hence, in this case, the asymptotic behavior \eqref{eq 4.2} is not well-defined.}
    If $\ov F_1\in\mathcal{RV}_{-\alpha}$, then the condition $(Y_1,F_1^{\leftarrow}\circ F_2(Y_2))^\top\in\MRV(\alpha_i,b_i,\mu_i',\E_2^{(i)}), i=1,2,$ can be formulated as a condition on the survival copula $\widehat C$ of $(Y_1,Y_2)$. Necessary and sufficient conditions on $\widehat C$ are given in \cite[Theorem 3.11 and Theorem 3.12]{das:fasen:2019}.
    \item Suppose $(Y_1,Y_2)^\top\in\MRV(\alpha_i,b_i,\mu_i,\E_2^{(i)})$ for  $i=1,2$ and $\P(Y_2>t)\sim K\P(Y_1>t)$ as $t\to\infty$ for some $K>0$.  Then
    $\bY':=(Y_1,F_1^{\leftarrow}\circ F_2(Y_2))^\top\in\MRV(\alpha_i,b_i,\mu_i',\E_2^{(i)})$ for $i=1,2$ and
    \begin{align} \label{def: h}
      h(y)  =\mu_2'(\left(y,\infty\right)\times(1,\infty))=\mu_2(\left(y,\infty\right)\times(K^{1/\alpha},\infty)).
     \end{align} 
      
    \item  The assumption $ 0<\upsilon<1$ is only sufficient and not necessary. Indeed, if $r=\infty$ and $g(\gamma)\to 0$ as $\gamma\downarrow 0$  (which is the standard case) then 
    for any $\upsilon\in(0,\infty)$ we have $\upsilon g(\gamma)\in(0,1)$ for small $\gamma$ and thus,  $\upsilon\in(0,\infty)$
    is allowed as well.
    \item   In general, we cannot guarantee that as $\gamma\downarrow 0$, we have  $h_\gamma^{\leftarrow}(v)/h^{-1}(v)$ converging uniformly to $1$ on bounded, yet non-compact intervals; nevertheless, such uniform convergence does hold for compact intervals; cf. \Cref{Lemma C.1}. The need for assuming uniform convergence on non-compact intervals is evident from the proof of \Cref{corollary 4.4}(a)(i), thus providing a justification for the additional assumptions in \Cref{thm:bivCoVaRmain}(b,c).
%On compact intervals  the uniform convergence of $h_{\gamma}^{\leftarrow}(v)/h^{-1}(v)$ holds as the following lemma shows. However, for non compact intervals this is in general not true (which we see nicely from the proof of \Cref{corollary 4.4}) although $h_\gamma^{\leftarrow}(v)/h^{-1}(v)$ converges uniformly to $1$  on half open intervals $\left[L,r\right)$. Therefore, the additional assumptions in \Cref{thm:bivCoVaRmain}(b,c) are necessary.
\end{enumerate}
\end{remark}

\begin{lemma} \label{Lemma C.1}
Let the assumptions of \Cref{thm:bivCoVaRmain} hold. Then for any closed interval $[a_1,a_2]\subset(0,r)$ we have as $\gamma\downarrow 0$,
\begin{eqnarray*}
    \sup_{v\in[a_1,a_2]}\left|\frac{h_\gamma^{\leftarrow }(v)}{h^{-1}(v)}-1\right|
    \to 0. 
\end{eqnarray*}
\end{lemma}

\begin{example}\label{ex:ggammastrind}
Prior to discussing further implications of \Cref{thm:bivCoVaRmain}, it is instructive to note the behavior of CoVaR when $Y_1=Y_2$ a.s. and when $Y_1, Y_2$ are independent. 
For convenience, assume that $Y_1,Y_2$ are identically   Pareto($\alpha_1$) distributed and $g:(0,1)\to(0,1)$ is a measurable function. Using the notations defined in \Cref{thm:bivCoVaRmain} we have the following.
\begin{itemize}
    \item[(a)] If {$Y_1=Y_2$ a.s.} then $(\alpha_1,b_1,\mu_1)=(\alpha_2,b_2,\mu_2)$ and   $h_\gamma(y)=h(y)=\max(y,1)^{-\alpha_1}$. For any 
   $0<\upsilon<1$ we have
    \begin{eqnarray*}
        \text{CoVaR}_{\upsilon g(\gamma)|\gamma}(Y_1|Y_2)= %(\upsilon\gamma)^{-\frac{1}{\alpha_1}} b_1(\gamma^{-1}) \sim 
        (\upsilon g(\gamma))^{-\frac{1}{\alpha_1}} \VaR_{\gamma}(Y_1)= (\upsilon \gamma g(\gamma))^{-\frac{1}{\alpha_1}}, \quad 0<\gamma<1.  
    \end{eqnarray*}
    In particular, if $g(\gamma)=\gamma$, then
    \begin{align*}\
       \text{CoVaR}_{\upsilon \gamma |\gamma}(Y_1|Y_2) = \Var_{\upsilon \gamma^2}(Y_1),  \quad 0<\gamma<1,
    \end{align*} 
    and if $g(\gamma)=1$ then 
    \begin{align}\label{eq:ggammaind}
       \text{CoVaR}_{\upsilon |\gamma}(Y_1|Y_2) = \Var_{\upsilon \gamma}(Y_1) = \upsilon^{-1/\alpha_1}\Var_{\gamma}(Y_1), \quad 0<\gamma<1.
    \end{align} 
    %Now if $Y_1, Y_2$ follow standard Pareto$(\alpha)$ then $\text{CoVaR}_{\upsilon\gamma,\gamma}(Y_1|Y_2) \sim (\upsilon\gamma^2)^{-1/\alpha}$ as $\gamma\downarrow 0$.
    \item[(b)] If $Y_1,Y_2$ are independent then $\alpha_2=2\alpha_1$, and a canonical choice is $b_2^{\leftarrow}(t)=(b_1^{\leftarrow}(t))^2$ making $h_\gamma(y)=h(y)=y^{-\alpha_1}$. Hence, for $0<\upsilon<1$,
    \begin{eqnarray*}
         \text{CoVaR}_{\upsilon g(\gamma) |\gamma}(Y_1|Y_2) = \Var_{\upsilon g(\gamma)}(Y_1),
        \quad 0<\gamma<1. 
    \end{eqnarray*}
    In particular, if $g(\gamma)=\gamma$, then
    \begin{align}\label{eq:ggammadep}
       \text{CoVaR}_{\upsilon \gamma |\gamma}(Y_1|Y_2) = \Var_{\upsilon \gamma}(Y_1)= \upsilon^{-1/\alpha_1}\Var_{\gamma}(Y_1), \quad 0<\gamma<1.
    \end{align} 
\end{itemize}
\end{example}
%Traditionally,  $\text{CoVaR}_{\gamma_1 |\gamma_2} (Y_1|Y_2)$  computations are implemented with either one or both of  $\gamma_1, \gamma_2$ fixed (\cite{girardi:ergun:2013, reboredo:ugolini:2015}) or $\gamma_1=O(\gamma_2)$ for $\gamma_1\downarrow 0$ (\cite{kley:kluppelberg:reinert:2017}). 
Comparing \eqref{eq:ggammaind} and \eqref{eq:ggammadep} we observe that for different levels of strength of dependence between $Y_1$ and $Y_2$ there exist different choices of $g(\gamma)$ allowing {for the following asymptotic behavior}: 
\begin{align*}
\text{CoVaR}_{\upsilon g(\gamma) |\gamma}(Y_1|Y_2) = \upsilon^{-1/\alpha_1}\VaR_{\gamma}(Y_1), \quad \gamma\downarrow0.
\end{align*}
Such a characterization may be obtained for a variety of dependence behavior and hence, we define the following quantity.
%When $Y_1, Y_2$ are asymptotically dependent, such choices seem natural; and assuming $\gamma_1=\gamma_2=\gamma)$ for $\gamma_1\downarrow 0$, we may obtain\[\text{CoVaR}_{\upsilon\gamma |\gamma}(Y_1|Y_2) = O (\VaR_{\upsilon\gamma^2}(Y_2)),\]
%see \citep{kley:kluppelberg:reinert:2017} for pertinent results under a bipartite network setting. In presence of asymptotic independence between $Y_1, Y_2$ such asymptotics may not hold true, \Cref{thm:bivCoVaRmain} aims to capture the asymptotic behavior of $\text{CoVaR}_{\gamma_1 |\gamma_2} (Y_1|Y_2)$ in terms of $\Var_{\gamma_2}(Y_2)$ for small values of $\gamma_1, \gamma_2$. 

 \begin{definition}\label{def:eci}
     Let   $\bY=(Y_1, Y_2)^\top$  be a bivariate random vector. Suppose for some measurable function $g:(0,1)\to(0,1)$, where $g(t^{-1})\in \RV_{-\beta}, \beta \ge0$, we have for $0<\upsilon<1$,
     \begin{align}\label{eqdef:ggamma}
\text{CoVaR}_{\upsilon g(\gamma) |\gamma}(Y_1|Y_2) = O(\VaR_{\gamma}(Y_1)), \quad \gamma\downarrow0.
\end{align}
Then we call $\beta^{-1}$ the \emph{Extreme CoVaR Index} of $Y_1$ given $Y_2$, or, $\text{ECI}(Y_1|Y_2)$ in short. 
 \end{definition}

\begin{remark}~ \label{Remark 3.6(a)} $\text{ECI}(Y_1|Y_2)$ provides a value to assess the strength of risk contagion from $Y_2$ to $Y_1$ and may take any value between 0 and $\infty$, including both values.
  \begin{enumerate}[(a)]  
    \item In \Cref{ex:ggammastrind}, we observed that for $Y_1$ and $Y_2$ strongly dependent,  the level required for $\CoVaR_{\upsilon g(\gamma)|\gamma}(Y_1|Y_2) = O(\Var_{\gamma}(Y_1))$ is in fact $g(\gamma)=\gamma^0=1$ and hence,    
    the ECI is $1/0=\infty$. On the other hand, if $Y_1$ and $Y_2$ are independent the required rate for having $\CoVaR_{\upsilon g(\gamma)|\gamma}(Y_1|Y_2) = O(\Var_{\upsilon \gamma}(Y_1))$ is $g(\gamma)=\gamma^1$ and hence, the $\text{ECI}(Y_1|Y_2)=1$.
  \item The ECI of $Y_1$ given $Y_2$ provides a measure of risk contagion between $Y_1$ and $Y_2$, in particular for systemic risks. Given $Y_2$ has a value larger than its $\VaR$ at level $\gamma$, ECI allows us to compute the level  $g(\gamma)=O(\gamma^{1/\text{ECI}})$ for $Y_1$ which will make the $\CoVaR_{\upsilon g(\gamma)|\gamma}(Y_1|Y_2)$ to be of the same order as $\VaR_{\gamma}(Y_1)$. 
  Lower values of ECI reflect lower values of $\CoVaR_{\gamma|\gamma}(Y_1|Y_2)$.

\end{enumerate}
\end{remark}

 If   $\bY':=(Y_1,F_1^{\leftarrow}\circ F_2(Y_2))^\top\in\MRV(\alpha_i,b_i,\mu_i,\E_2^{(i)})$ for  $i=1,2$,   the choice of a function $g(\cdot)$ satisfying \eqref{cond:bivcova}, must also satisfy $g(t^{-1})\in \RV_{-\beta}$  with $\beta=\alpha_2/\alpha_1-1$ and hence, 
\eqref{eqdef:ggamma} is satisfied. Let us summarize this result.

\begin{proposition}\label{Lemma 4.6}
    Let   $\bY=(Y_1, Y_2)^\top$ be  bivariate random vector with margins $F_1$ and $F_2$, respectively. Suppose
$\bY':=(Y_1,F_1^{\leftarrow}\circ F_2(Y_2))^\top\in\MRV(\alpha_i,b_i,\mu_i,\E_2^{(i)})$ for  $i=1,2$.
Then  $$\ECI(Y_1|Y_2)=\frac{\alpha_1}{\alpha_2-\alpha_1}$$ where $\alpha_1/0:=\infty$.
\end{proposition} 

\begin{remark}
  { Note that for $(Y_1,Y_2)^\top\in\MRV(\alpha_i,b_i,\mu_i,\E_2^{(i)})$, $i=1,2$, $\alpha_2$ is often called the parameter of hidden regular variation in bivariate regularly varying models (\cite{resnick:2002, maulik:resnick:2005}), and is closely related to the coefficient of tail dependence $\eta$,   which is defined in \cite{ledford:tawn:1996} under the assumption $\alpha_1=1$. If $\alpha_1=1$ and $Y_1,Y_2$ are tail-equivalent then we have the relation $\eta=1/\alpha_2=\ECI(Y_1|Y_2)/(1+\ECI(Y_1|Y_2))$. An important difference is as well that we only require the MRV of $(Y_1,F_1^{\leftarrow}\circ F_2(Y_2))^\top$ on $\E_2^{(i)}$ for $i=1,2$ and not of $(Y_1,Y_2)^\top$.} %The ECI value obtained in \Cref{Lemma 4.6} uses values of both $\alpha_1$ and $\alpha_2$.} % and thus we have a one-one correspondence between $\eta$ and ECI under $\MRV$ assumptions on appropriate cones}. %The ECI value obtained in \Cref{Lemma 4.6} uses the value of both $\alpha_1$, the primary regularly varying parameter of the (heaviest) marginal and $\alpha_2$.
\end{remark}
\begin{remark}
 Suppose $(Y_1^\perp,Y_2^\perp)^\top$ is a bivariate random vector {with independent components but with the same marginal distribution as $(Y_1,Y_2)^\top$, which satisfies the assumptions of \Cref{Lemma 4.6}.} Then by the definition of the ECI and the function $g(\cdot)$ in \Cref{def:eci},  
 and \Cref{Remark 3.6(a)}(a), we have
 \begin{eqnarray*}
    \CoVaR_{g(\gamma)|\gamma}(Y_1|Y_2)=O(\VaR_\gamma(Y_1))
    =O(\CoVaR_{\gamma|\gamma}(Y_1^\perp|Y_2^\perp)), \quad \gamma \downarrow 0.
 \end{eqnarray*}
  Thus the value of $g(\gamma)$ helps in assessing the probability or confidence with which 
     $\CoVaR(Y_1|Y_2)$ will have the same asymptotic behavior as the CoVaR of the independent model, i.e., $\CoVaR(Y_1^\perp|Y_2^\perp)$. In other words, higher values of ECI (and hence, $g(\gamma)$) imply a higher  risk contagion of $Y_2$  on $Y_1$; {by definition this also} means that with probability $1-g(\gamma)$
    the value of $ \CoVaR_{g(\gamma)|\gamma}(Y_1|Y_2)$ is sufficient to cover the losses of $Y_1$ if the value of $Y_2$ is already above $\VaR_\gamma(Y_2)$. Clearly, larger values of $\alpha_2$ result in weaker tail dependence between $Y_1$ and $Y_2$, and a smaller ECI. 
    
    A similar phenomenon may be observed also for the popular systemic risk measures Marginal-Mean-Excess (MME)  and Marginal-Expected-Shortfall (MES) (\cite{acharya:petersen:philippon:richardson:2017,das:fasen:2018}) which are defined as     
    \begin{align*}
        \text{MME}_\gamma(Y_1|Y_2)&:=\E((Y_1-\VaR_\gamma(Y_2))_+|Y_2>\VaR_\gamma(Y_2)), 
        \intertext{and,}
       \text{MES}_\gamma(Y_1|Y_2)&:=\E(Y_1|Y_2>\VaR_\gamma(Y_2))
    \end{align*}
     respectively.
We know from \cite{das:fasen:2018,das:fasen:2019} that under some mild assumptions,
 \begin{eqnarray*}
    g(\gamma)^{-1}\text{MME}_\gamma(Y_1|Y_2)=O(\VaR_\gamma(Y_1))=O(\gamma^{-1}\text{MME}_\gamma(Y_1^\perp|Y_2^\perp)), \quad \gamma\downarrow 0.
\end{eqnarray*}
{We can check that} for example, the P-GC and the P-MOC model satisfy these assumptions.
 Moreover, under some stronger (moment) assumptions, we also have that 
\begin{eqnarray*}
    g(\gamma)^{-1}\text{MES}_\gamma(Y_1|Y_2)=O(\VaR_\gamma(Y_1))=O(\gamma^{-1}\text{MES}_\gamma(Y_1^\perp|Y_2^\perp)), \quad \gamma\downarrow 0.
\end{eqnarray*}
 These stronger assumptions are unfortunately not satisfied by the $\PGC$ model. To summarise, under some regularity conditions, for increasing values of ECI and hence $g(\cdot)$,
the MME (and MES) increase as well to attain the same confidence level $1-\gamma$ as in the independent model, implying higher risk contagion of $Y_2$ on $Y_1$. In particular, a conclusion of the asymptotic behavior of CoVaR, MES and MME is that although the underlying models exhibit asymptotic independence there is still a certain amount of dependence on the tails which will have an influence on risk contagion, and this strength of tail dependence is measured using the function $g(\cdot)$ and finally reflected in the ECI.
\end{remark}

%\begin{remark}~
% \VF{If $Y_1$ and $Y_2$ are not tail-equivalent the conclusion of \Cref{Lemma 4.6} does not hold anymore. In this case, the ECI depends on the tail behavior of $Y_1$, $Y_2$, and the properties of the function $h$ and a general assertion about the asymptotic behavior of the $\CoVaR$ is not possible anymore.}
%\end{remark}

\subsection{Measuring CoVaR under different model assumptions}

In this section, we show the direct consequences of \Cref{thm:bivCoVaRmain} for various underlying distributions discussed in this paper, in particular asymptotically dependent, Gaussian copula, and  Marshall-Olkin copula models. % which are relevant for systemic risk modeling as well. 
These models provide a flavor of expected results, although computations for complex networks are more complicated, as we will see in \Cref{sec:bipnet}.

First, we consider the asymptotically dependent case where $\alpha_1=\alpha_2$ and the result is a direct consequence of \Cref{thm:bivCoVaRmain} and \Cref{def: h}.
\begin{proposition}[Asymptotic dependence] \label{CoVar:asymptocially dependent}
Let $\bY=(Y_1,Y_2)^\top\in \R^2$ be a bivariate tail-equivalent random vector with $\bY\in\MRV(\alpha,b,\mu,\E_2^{(1)})\cap \MRV(\alpha,b,\mu,\E_2^{(2)})$ and  $h(y)={\mu}((y,\infty)\times (K^{1/\alpha},\infty))$ with
\begin{eqnarray*}
    K=\frac{{\mu}(\R_+\times(1,\infty))}{{\mu}((1,\infty)\times\R_+)} \quad\quad \text{ and } \quad\quad c={\mu}((1,\infty)\times\R_+)
\end{eqnarray*} %\marginpar{\VF{h might be zero}}
{Suppose $h:(l,\infty)\to(0,r)$ is strictly decreasing and continuous for some $r,l>0$.
Then for $0<\upsilon<\min(r/c,1)$ %, and  $h^{-1}(\upsilon c)\in(0,\infty)$ 
we have}
\begin{eqnarray*}
    \CoVaR_{\upsilon|\gamma}(Y_1|Y_2) \sim    {h^{-1}(\upsilon c)\VaR_{\gamma}(Y_1)}, \quad \gamma\downarrow 0.
\end{eqnarray*}
Moreover,  $\ECI(Y_1|Y_2)=\infty$.
    
\end{proposition}

For the bipartite network model of \Cref{sec:intro} where $\bZ$ has asymptotically dependent pairs we obtain directly the following result, some restricted versions of which have been shown in \cite{kley:kluppelberg:reinert:2016,kley:kluppelberg:reinert:2017}.

\begin{example}[Bipartite network with asymptotically dependent objects ]\label{example:bipartite:network:dependence}
Let  $\bZ\in \R_+^d$ be a random vector, $\bA \in \R_+^{2 \times d}$ be a  random matrix, and $\bX=\bA\bZ$.   Now also assume 
 $\bZ \in \MRV(\alpha,b, \mu,\E^{(1)}_d)\cap \MRV(\alpha,b, \mu,\E^{(2)}_d)$ and $\bZ$ has completely tail equivalent margins. Then due to \citet[Proposition A.1]{basrak:davis:mikosch:2002a}, which generalizes Breiman's Theorem to the multivariate setup, we have
 $\bX=\bA\bZ \in \MRV(\alpha,\overline{\mu},\E_2^{(1)})\cap \MRV(\alpha,\overline{\mu},\E_2^{(2)})$ with
  $
    \overline{\mu}(\cdot)=\E\big[ \mu(\bA^{-1}(\cdot))\big].
  $
  Define $
    h(y)= \overline{\mu}((y,\infty)\times (K^{1/\alpha},\infty))$
where 
\begin{eqnarray*}
    K=\frac{\overline{\mu}(\R_+\times(1,\infty))}{\overline{\mu}((1,\infty)\times\R_+)} \quad\quad \text{ and } \quad\quad c=\overline{\mu}((1,\infty)\times\R_+)
\end{eqnarray*}
{Suppose $h:(l,\infty)\to(0,r)$ is strictly decreasing and continuous for some $r,l>0$. %, and  $h^{-1}(\upsilon c)\in(0,\infty)$ 
%Suppose $h^{-1}(\upsilon c)\in(0,\infty)$, 
Then a conclusion of \Cref{CoVar:asymptocially dependent} is that for $0<\upsilon<\min(r/c,1)$,}
\begin{eqnarray*}
    \CoVaR_{\upsilon|\gamma}(X_1|X_2) \sim    {h^{-1}(\upsilon c)\VaR_{\gamma}(X_1)}, \quad \gamma\downarrow 0.
\end{eqnarray*}
Moreover,  $\ECI(X_1|X_2)=\infty$.

 In  \cite{kley:kluppelberg:reinert:2016,kley:kluppelberg:reinert:2017}, the authors investigate the asymptotically co-monotone case  $\mu([0,\bx]^c)=\max_{j=1,\ldots,d}(K_jx_j^{-\alpha})$ for $\bx=(x_1,\ldots,x_d)^\top\in\R_+^d$ with positive constants $K_1,\ldots,K_d>0$, our result turns out to be more general.% \DAS{addressing any asymptotically independent behavior.} \VF{ But here we are in the strong dependent case, it is not related to the asymptotically independent case. Therefore I would skip your comment.}
\end{example}

In the next two results, we investigate the asymptotically independent case  $\alpha_1<\alpha_2$, where the dependence is modeled by a Marshall-Olkin copula and a Gaussian copula, respectively.

\begin{proposition}[Marshall-Olkin copula]  \label{corollary 4.4}
Let $\bY=(Y_1,Y_2)^\top\in \R^2$ be a random vector and   $0<\upsilon<1$.
\begin{enumerate}[(a)]
 \item  Let $\bY \in \PMOC(\alpha, \theta, \lambda^{=})$ for $\lambda>0$. 
 \begin{itemize}
    \item[(i)] Suppose either (i) $\beta= 1/2$, or (ii)  $\beta >1/2$ and as $\gamma\downarrow 0$, $ v^{\alpha} \ov F_1^{\leftarrow}(v \gamma)/\ov F_1^{\leftarrow}(\gamma)$ converges uniformly to $1$  on $\left(0,1\right]$. %\marginpar{\VF{It is $F_1$ instead of $F_2$. Changed it.}} 
    Then 
 \begin{eqnarray*}%\label{eq:covarPMOCeq2}
    \CoVaR_{\upsilon\gamma^\beta|\gamma}(Y_1|Y_2) \sim  (\upsilon \gamma^\beta)^{-\frac1{\alpha}}\gamma^{\frac{1}{2\alpha}}\VaR_{\gamma}(Y_1), \quad \gamma\downarrow0.
\end{eqnarray*}
%Moreover, with $\beta= 1/2$, we have $\CoVaR_{\upsilon \gamma^{1/2}|\gamma} = O(\gamma^{-1/\alpha})$, hence  $\text{ECI}(Y_1|Y_2)=2$.
    \item[(ii)] Suppose $0\leq \beta < 1/2$ 
          and as $\gamma\downarrow 0$, $v^{\alpha}\ov F_1^{\leftarrow}(v \gamma)/\ov F_1^{\leftarrow}(\gamma)$ converges uniformly to $1$  on $\left[1,\infty\right)$. 
 Then 
 \begin{eqnarray*}%\label{eq:covarPMOCeq2}
    \CoVaR_{\upsilon\gamma^\beta|\gamma}(Y_1|Y_2) \sim (\upsilon \gamma^\beta)^{-\frac{2}{\alpha}}\gamma^{\frac{1}{\alpha}}\VaR_{\gamma}(Y_1), \quad \gamma\downarrow0.
\end{eqnarray*}
 \end{itemize}
Moreover, $\ECI(Y_1|Y_2)=2$.
 \item Let $\bY \in \PMOC(\alpha, \theta, \lambda^{\propto})$ for $\lambda>0$. 
 %\marginpar{\DAS{I changed the index (in red). Can you check?}\VF{I changed again a few things because we might have a slowly varying function and in (ii) ECI=3}}
 \begin{itemize}
    \item[(i)] Suppose either (i) $\beta= 1/3$, or  (ii) $\beta >1/3$ and as $\gamma\downarrow 0$, $v^{\alpha}\ov F_1^{\leftarrow}(v \gamma)/\ov F_1^{\leftarrow}(\gamma)$ converges uniformly to $1$  on $\left(0,1\right]$.
    Then 
 \begin{eqnarray*}%\label{eq:covarPMOCeq2}
    \CoVaR_{\upsilon\gamma^\beta|\gamma}(Y_1|Y_2) \sim  (\upsilon \gamma^\beta)^{-\frac1{\alpha}}\gamma^{\frac{1}{3\alpha}}\VaR_{\gamma}(Y_1), \quad \gamma\downarrow0.
\end{eqnarray*}
%Moreover, with $\beta= 1/3$, we have $\CoVaR_{\upsilon \gamma^{1/3}|\gamma} = O(\gamma^{-1/\alpha})$, hence  $\text{ECI}(Y_1|Y_2)=3$.
    \item[(ii)]  Suppose $0\leq \beta < 1/3$ 
          and as $\gamma\downarrow 0$, $v^{\alpha}\ov F_1^{\leftarrow}(v \gamma)/\ov F_1^{\leftarrow}(\gamma)$ converges uniformly to $1$  on $\left[1,\infty\right)$. 
 Then 
 \begin{eqnarray*}%\label{eq:covarPMOCeq2}
    \CoVaR_{\upsilon\gamma^\beta|\gamma}(Y_1|Y_2) \sim (\upsilon \gamma^\beta)^{-\frac{3}{\alpha}}\gamma^{\frac{1}{\alpha}}\VaR_{\gamma}(Y_1), \quad \gamma\downarrow0.
\end{eqnarray*}
 \end{itemize}
{Moreover, $\ECI(Y_1|Y_2)=3$.}
\end{enumerate}  
\end{proposition}

\begin{remark} $\mbox{}$
\begin{itemize}
\item[(a)]
For $\bY\in \PMOC$ as in \Cref{corollary 4.4}, the uniform convergence of $v^{\alpha}\ov F_j^{\leftarrow}(v \gamma)/\ov F_j^{\leftarrow}(\gamma)$ on some set   $I\subset (0,\infty)$ is 
necessary and sufficient for the uniform convergence of $h_\gamma^{\leftarrow}(v)/h^{-1}(v)$ on $I$.
However, we can check that even if $\ov F_j^{\leftarrow}(v \gamma)/\ov F_j^{\leftarrow}(\gamma)$ converges uniformly to $v^{-\alpha}$ on $\left[1,\infty\right)$, it may not necessarily converge uniformly on $\left(0,1\right]$ (counterexamples exist).
Nevertheless, exactly Pareto$(\alpha)$ distributed margins satisfy
$v^{\alpha}\ov F_j^{\leftarrow}(v \gamma)/\ov F_j^{\leftarrow}(\gamma)=1$ for any $v>0$ and $\gamma$ small, and hence, the asymptotic behavior of the $\CoVaR$ holds for Pareto margins and any $\beta\geq 0$ without any additional assumption.
\item[(b)]  For bivariate random vectors $\bY^{=}\in \PMOC(\alpha,\theta,\lambda^{=})$ and $\bY^{\propto}\in \PMOC(\alpha,\theta,\lambda^{\propto})$, \Cref{corollary 4.4} implies that $\text{ECI}(Y_1^{\propto}|Y_2^{\propto})>\text{ECI}(Y_1^{=}|Y_2^{=})$. 
Thus, the risk contagion in the $\PMOC(\alpha,\theta,\lambda^{=})$ model is higher than  in the $\PMOC(\alpha,\theta,\lambda^{\propto})$ model if the $\CoVaR$ is used as risk measure.
Even though both models exhibit asymptotic independence, the result shows that there is still some dependence on the tails influencing the CoVaR.

%\VF{For the Marshal-Olkin copula with  proportional parameters the ECI is higher than for the Marshal-Olkin copula with equal parameters showing that for the Marshal-Olkin copula with proportional parameters the dependence in the tails is stronger than for equal parameters although both exhibit asymptotic independence. In particular, it shows that the dependence structure has still an influence on systemic risk measures.}
\end{itemize}
\end{remark}

\begin{proposition}[Gaussian copula]  \label{CoVaR Gaussian}
Let $\bY=(Y_1,Y_2)^\top \in \PGC(\alpha, \theta, \Sigma_{\rho})$ with $\rho\in (-1,1)$ and  $g(\gamma)=\gamma^{\frac{{1-\rho}}{1+\rho}}(-\log \gamma)^{-\frac{\rho}{1+\rho}}$. Then for $0<\upsilon<1$ and as $\gamma\downarrow 0$,
\begin{eqnarray*}%\label{eq:covarPGC2}
    \CoVaR_{\upsilon g (\gamma)|\gamma}(Y_1|Y_2) \sim B^*\,\cdot \upsilon ^{-\frac{1+\rho}{\alpha}} \VaR_\gamma(Y_1)
\end{eqnarray*} 
%\marginpar{\VF{I deleted $g(\gamma)=O()$}.}

where  $B^*=B^*({\rho,\alpha}) = (4\pi)^{-\frac{\rho}{\alpha}}   (1+\rho)^{\frac{3(1+\rho)}{2\alpha}} (1-\rho)^{-\frac{1+\rho}{2\alpha}}.$ Finally, $\text{ECI}(Y_1|Y_2)= \frac{1+\rho}{1-\rho}$.
\end{proposition}

\begin{remark} $\mbox{}$
\begin{itemize}
\item[(a)]
Although the logarithm is a slowly varying function, $\log(vt)/\log ( t)$  converges uniformly only on compact intervals. Hence, for the Gaussian dependence case, it seems that $h_\gamma^{\leftarrow}(v)/h^{-1}(v)$ may not converge uniformly on intervals of the form $\left(0,R\right]$ or $\left[L,\infty\right)$ for any  $L,R>0$, thus we do not attempt verifying conditions (b) or (c) in \Cref{thm:bivCoVaRmain}.
\item[(b)] The measure $\text{ECI}(Y_1|Y_2)=\frac{2}{1-\rho}-1$ is increasing in $\rho$ suggesting, not quite surprisingly that, as the Gaussian correlation $\rho$ increases, the risk contagion measured by the CoVaR increases as well; in fact as $\rho$ increases from $-1$ to $1$, ECI increases from 0 to $\infty$.
%\VF{$\text{ECI}(Y_1|Y_2)=\frac{2}{1-\rho}-1$ is increasing in $\rho$ suggesting, not surprisingly, that for $\rho$ increasing, i.e. the correlation of the Gaussian copula increasing, the tail dependence is getting stronger and hence, higher risk reserves are necessary.}
\end{itemize}
\end{remark}

\section{Risk contagion in a bipartite network with asymptotically independent objects} \label{sec:bipnet}

Recall the bipartite network structure defined in \Cref{sec:intro} where the risk exposure of $q$ entities of a financial system given by $\bX\in \R_+^q$ is captured using the risk exposure of the underlying assets $\bZ\in\R_+^d$ and the bipartite network is defined via the matrix $\bA\in \R_+^{q\times d}$.
In this section, we derive asymptotic tail probabilities of the risk exposure $\bX=\bA\bZ$ 
where the objects in $\bZ$ are asymptotically independent,  the case of pairwise asymptotically dependent objects was already covered in \Cref{example:bipartite:network:dependence}.
We are particularly interested in tail probabilities of \emph{rectangular sets} 
 \begin{align}\label{set:Azs}
       \Gamma_{\bx_S}^{(q)} = \{\bv\in \R_+^q: v_s > x_s, \forall s\in S\}
    \end{align} 
where $x_s>0$ for all $s\in S\subseteq \mathbb{I}_q$ and $\bx_S=(x_s)_{s\in S}$, 
as these help us first in computing conditional probabilities and eventually conditional risk measures like CoVaR. In the bivariate setup, the rectangular sets are of the form $\left[\bzero,\bx\right]^c$ and $(\bx,\binfty)$, $\bx\in\R_+^2$. The proofs for the results in this section are given in \Cref{Appendix D}.

Before stating the results, we need some definitions and notations following \cite{das:fasen:kluppelberg:2022}. Recall that we denote the set of agents/entities by $\mathcal{A}=\mathbb{I}_q$ and the set of assets/objects by $\mathcal{O}=\mathbb{I}_d$.
\begin{definition}
  For $k\in\mathbb{I}_q$ and $i\in\mathbb{I}_d$ the functions $\tau_{(k,i)}:\R_+^{q\times d}\to\R_+$ are defined as
\beam\label{mtau}
\tau_{(k,i)} (\bA)  =\sup_{\bz\in\E_d^{(i)}} \frac{(\bA\bz)_{(k)}}{z_{(i)}}. % = \sup_{\bz\in \partial \aleph_d^{(i)}} \tau^{(k)} (\bA \bz).
\eeam  
\end{definition}

The functions $\tau_{(k,i)}$ are meant to be like norms for the matrices $\bA\in \R_+^{q\times d}$. % where we are concerned with computing limit measures of subsets in $\E_d^{(i)}\subset \R^d_+$ under the map $\bA$ and we seek the appropriate subspace $\E_q^{(k)}$ contained in $\R_+^q$ where we may obtain non-null limit measures under multivariate regular variation. 
Although the $\tau_{(k,i)}$'s are not necessarily norms (or even semi-norms) on the induced vector space (see \cite[Section~5.1]{horn:johnson:2013}), they do admit some useful properties; cf. \cite[Lemma 3.4]{das:fasen:kluppelberg:2022}.
%In particular, suppose $\bA \in \R_+^{q\times d}$ is a deterministic matrix with all rows non-trivial and let $\mathds{1}_{\bA} = ((\mathds{1}_{\{A_{ki}>0\}}))\in \R_+^{q\times d}$;
%then for fixed $i\in\{1,\dots,d\}$ and fixed $k\in\{1,\dots,q\}$, the following equivalences hold (cf. \citet{das:fasen:kluppelberg:2022}, Lemma 3.2 and Lemma 3.3):
%\begin{align*}
%    \emptyset\not=\bA^{-1}(\E_q^{(k)}) \subseteq \E_d^{(i)}
%    \quad \Longleftrightarrow \quad
%        0<\tau_{(k,i)} (\bA) <\infty
%    \quad \Longleftrightarrow \quad 0<\tau_{(k,i)} (\mathds{1}_{\bA}) <\infty.
%\end{align*}

\begin{definition}\label{def:decom}
Let $\bA\in \R_+^{q\times d}$ be a random matrix.
For $k\in\mathbb{I}_q$ and  $\omega\in\Omega$, define \linebreak$\bA_\omega:=\bA(\omega)$ and
 $$ i_k(\bA_\omega):= {\max}\{j\in\mathbb{I}_d: \tau_{(k,j)}(\bA_\omega)<\infty\},$$
which creates a partition  $\Omega^{(k)}(\bA)=(\Omega_i^{(k)}(\bA))_{i=1,\ldots,d}$ of $\Omega$ given by
\[\Omega_i^{(k)}: = \Omega_i^{(k)} (\bA):= \{\omega\in \Omega: i_k(\bA_{\omega})=i\}, \quad i\in\mathbb{I}_d.\]
We write $\P_i^{(k)} (\,\cdot\,):=\P(\,\cdot\,\cap\,\Omega_i^{(k)})$ and $\bE_i^{(k)}[\,\cdot\,]:=\bE[\,\cdot\,\bone_{\Omega_i^{(k)}}]$.
\end{definition}

Now we are ready to characterize the asymptotic probabilities of $\bX=\bA\bZ$ for various tail sets $C$ of $\R_+^q$; cf. \cite[Theorem 3.4 and Proposition 3.2]{das:fasen:kluppelberg:2022} and the details in \Cref{Appendix D}.

 \begin{theorem} \label{thm:main:dfk}
Let  $\bZ\in \R_+^d$ be a random vector, $\bA \in \R_+^{q \times d}$ be a  random matrix, and $\bX=\bA\bZ$. Also for fixed $k\in \mathbb{I}_q$, let $C\subset \E_q^{(k)}$
be a  Borel set bounded away from \linebreak $\{\bx\in \R^q_+: x_{(k)}=0\}$.  Now also assume the following:
\begin{enumerate}[(i)]
    \item $\bZ \in \MRV(\alpha_{i},b_{i}, \mu_{i},\E^{(i)}_d)$ for all $i\in\mathbb{I}_d$ with $\lim_{t\to\infty} b_i(t)/b_{i+1}(t)= \infty$  for \linebreak $i=1,\ldots,d-1$.
    \item $\bA$ has almost surely no trivial rows and is independent of $\bZ$.
    \item $\bE_i^{(k)}[\mu_i(\partial\bA^{-1}(C))]=0$ for all $i\in \mathbb{I}_d$. 
    \item For all $i\in \mathbb{I}_d$ we have $ \bE_i^{(k)}\big[(\tau_{(k,i)}(\bA))^{\alpha_{i}+\delta}\big] <\infty$ for some $\delta=\delta(i,k)>0$.
    \end{enumerate}
Then the following holds:
\begin{itemize}
\item[(a)]  Define
%\begin{align*}\label{eq:ik*}
$i_k^* := \arg\min\{i\in \mathbb{I}_d: \P(\Omega_i^{(k)})>0\}.$
%    \end{align*}
    Then we have
 \begin{align*}%\label{eq:constant}
    b_{i_k^*}^{\leftarrow}(t)\P(\bX \in tC) \to \bE_{i_k^*}^{(k)}\big[ \mu_{i_k^*}(\bA^{-1}(C))\big] \, =\overline{\mu}_{i_{k}^{*},k}(C)=:
 \overline{\mu}_{k}(C),\quad\tto.
  \end{align*}
    Moreover, if $\overline{\mu}_{k}$ is a non-null measure then $\bA\bZ \in \MRV(\alpha_{i_{k}^{*}},\overline{\mu}_{k},\E_q^{(k)})$.
\item[(b)] Define
\begin{align*}%\label{def:iota}
\bar{\iota}= \bar{\iota}_C:= \min\{d, \inf\{i\in\{i_k^*,\ldots,d\}: \bE_i^{(k)} [\mu_i(\bA^{-1}(C))]>0\}\}.
\end{align*}
Suppose for all $i=i^*_k,\ldots,\bar{\iota}-1$ and  $\omega\in \Omega_i^{(k)}$ that $\bA_{\omega}^{-1}(C)=\emptyset$.
 Then we have
\begin{align*}%\label{eq:limiota}
\P(\bX\in tC)= (b_{\bar{\iota}}^{\leftarrow}(t))^{-1}\bE_{\bar{\iota}}^{(k)}[\mu_{\bar{\iota}}(\bA^{-1}(C))] + o((b_{\bar{\iota}}^{\leftarrow}(t))^{-1}),\quad t\to\infty.
\end{align*}
\end{itemize}
\end{theorem}

\begin{remark}~
\begin{itemize}
    \item[(a)] For both  $\bZ \in \PGC(\alpha, \theta, \Sigma)$ with $\Sigma$ positive definite (cf. \Cref{prop:bi}) and  $\bZ\in \PMOC(\alpha, \theta, \lambda^{=})$ or $\bZ\in \PMOC(\alpha, \theta, \lambda^{\propto})$, respectively (cf. \Cref{prop:pmocmrv})
    assumption (i) of \Cref{thm:main:dfk} is satisfied. 
    \item[(b)] Assumption~(i) in \Cref{thm:main:dfk} excludes the asymptotically dependent case \linebreak $\bZ \in \MRV(\alpha,b, \mu,\E^{(i)}_d)$ for $i=1,2$, where $\mu(\E^{(2)}_d)>0$. But in this case, we can use the well-known Breiman's Theorem generalized to the multivariate setup in \cite[Proposition A.1]{basrak:davis:mikosch:2002a},  see \Cref{example:bipartite:network:dependence}.
%\marginpar{\DAS{Remark (b) was a bit confusing so I rephrased (check) } \VF{ fine}}
 %   \item[(b)] If $\Gamma_{\bx_S}^{(q)}$ is a rectangular set then the intersection of $\bA^{-1}(\Gamma_{\bx_S}^{(q)})$ with the support of  $\mu_i$ is a finite  union of rectangular sets, however it  not necessarily a rectangular set itself (see Example 3.2 in \citet{das:fasen:kluppelberg:2022}). If $d$ is large the number of sets might be quite large as well. But from \Cref{thm:main:dfk} we already know that only the sets with $i_k^*$ components of $\bZ$ are important.
  %  But still in high dimensions they might be quite expensive to calculate.
\end{itemize}

\end{remark}

\begin{remark} $\mbox{}$ If $\bA=\bA_0$ is a deterministic matrix, we have a few direct consequences.
\begin{itemize}
    \item[(a)] First, $i_k^*=i_k(\bA_0)$ and also $\Omega_j^{(k)}=\emptyset$ for all $j\not=i_k^*$.
    \item[(b)] If $C=\Gamma_{\bx_S}^{(q)}\subset \E_q^{(k)}$ is a rectangular set then the intersection of $\bA_0^{-1}(\Gamma_{\bx_S}^{(q)})$ with the support of  $\mu_{i_k^*}$ is a finite union of rectangular sets; however clearly, it is not necessarily a rectangular set itself (\cite[Example 3.2]{das:fasen:kluppelberg:2022}). When $d$ is large, the number of sets in this union might be quite large and hence, it may be computationally expensive not only to get an explicit expression for $\overline{\mu}_k(\Gamma_{\bx_S}^{(q)})$, but even to assess whether $\overline{\mu}_k(\Gamma_{\bx_S}^{(q)})>0$. Nevertheless, it turns out that there are reasonable sufficiency conditions under which we can guarantee $\overline{\mu}_k(\Gamma_{\bx_S}^{(q)})>0$ where $|S|=k$; see the next lemma which is a conclusion of \cite[Proposition~3.1]{das:fasen:kluppelberg:2022}.
\end{itemize}
\end{remark}

\begin{lemma}\label{prop:tailbipnonull}
    Suppose the assumptions of \Cref{thm:main:dfk} hold. For any rectangular set $\Gamma_{\bx_S}^{(q)}$ with $|S|=k$ we have $\overline{\mu}_k(\Gamma_{\bx_S}^{(q)})>0$ with $|S|=k$ if $\mu_{i_k^*}$ has mass on all $\binom{d}{i_k^*}$ coordinate hyperplanes comprising $\E_d^{(i^*_k)}$. 
%\begin{itemize}
  % \item[(a)] The columns of $\bA$ have the same distribution \DAS{(do we need i.i.d.(cf. the network examples))}.
  %  For a deterministic matrix this is satisfied if every column of $\bA$ have the same entries. \DAS{SKIP (a)}
\end{lemma}

\begin{remark}
   With regard to  \Cref{prop:tailbipnonull}, the following examples will satisfy \linebreak $\overline{\mu}_k(\Gamma_{\bx_S}^{(q)})>0$ for any rectangular set  $\Gamma_{\bx_S}^{(q)}$ with $|S| =k\ge 2$ and random matrix $\bA$ with non-trivial rows.
   \begin{enumerate}[(i)]
      \item $\bZ$ is exchangeable (including independence),
            \item $\bZ\in \PGC(\alpha,\theta, \Sigma_\rho)$ where $\rho\in\left(-{1}/{(d-1)},1\right)$, %$\Sigma_{\rho}$,   is a positive-definite equicorrelation matrix.%, $\rho\in\left(-\frac{1}{(d-1)},1\right).$%; due to \Cref{prop:bi} are satisfied.
            % the assumptions of \Cref{thm:main:dfk}.
            \item $\bZ\in \PMOC(\alpha, \theta, \lambda^{=})$,
            \item $\bZ\in \PMOC(\alpha, \theta, \lambda^{\propto})$.
   \end{enumerate}
Note that the distributions of $\bZ$ in  (ii)-(iv)  are close to exchangeability (they become exchangeable if we assume that the marginals are identically distributed instead of completely tail equivalent).  
\end{remark}

For a $\PGC(\alpha,\theta,\Sigma)$ model, \Cref{prop:supportpgc}
gives sufficient criteria on the correlation matrix $\Sigma$ 
such that $\mu_{i}$ has mass on all $\binom{d}{i}$ coordinate hyperplanes comprising $\E_d^{(i)}$, { naturally indicating that there exist PG-C models which may not satisfy this criterion, see \Cref{Example 2.7}. Yet, if $\bZ\in \PGC(\alpha, \theta, \Sigma)$, we are often able to use a technique of dimension reduction, as exhibited in the proof of \Cref{prop:condprobGC} to obtain the correct tail probability rates.}

\subsection{{Risk contagion between two portfolios}}\label{subsec:bivprob}
%\subsection{Towards risk \VF{contagion} : joint tail risk of two portfolios}
In the context of risk contagion, an important task is to understand the probability of an extremely large loss for a single asset or a linear combination of assets given an extremely large loss for some other asset or a linear combination of assets. This allows us to concentrate on $\bA\in \R_+^{q\times d}$ where $q=2$. 
For a particular type of event, we would also concentrate on risk exposures of pairs of financial entities that may invest in disjoint sets of assets, yet because of the dependence of the underlying variables, the joint probability is not necessarily a product measure. First, we obtain a general result for such joint tail probabilities characterizing limit measures $\overline{\mu}_1$ and $\overline{\mu}_2$  as obtained in  \Cref{thm:main:dfk}.

\begin{proposition} \label{exchangable} 
Let  $\bZ\in \R_+^d$ be a random vector  with 
$\bZ \in \MRV(\alpha_{i},b_{i}, \mu_{i},\E^{(i)}_d)$ for all ${i=1,2}$ and
 ${b_1(t)/b_{2}(t)\to \infty}$ as $t\to\infty$. % for $i=1,\ldots,d-1$.
Let $\bA \in \R_+^{2 \times d}$ be a random matrix with almost surely no trivial rows independent of $\bZ$.  With the notations from \Cref{thm:main:dfk}, the following holds for ${\overline \mu}_1$ and  ${\overline \mu}_2$ :
 \begin{itemize}
    \item[(a)] Suppose $\E\|\bA\|^{\alpha_1+\epsilon}<\infty$
    for some $\epsilon>0$. Then  $i_1^*=1$ and for $\bx\in\R_+^2$,
    \begin{eqnarray*}
            \overline{\mu}_1([\bzero,\bx]^c)=\sum_{\ell=1 }^d \E\left[\mu_1\left(\left\{\bz\in\R_+^d:\,\max\left(\frac{a_{1\ell}}{x_1},\frac{a_{2\ell}}{x_2}\right)z_\ell>1\right\}\right)\right].
        \end{eqnarray*}
    \item[(b)]  Suppose  $\max_{\ell\in \mathbb{I}_d}\P(\min\{a_{1\ell},a_{2\ell}\}>0)>0$ and  $\E\|\bA\|^{\alpha_1+\epsilon}<\infty$ for some $\epsilon>0$. Then $i_2^*=1$ and
    for $\bx\in (\bzero,\binfty)$,
    \begin{eqnarray*}
            \overline{\mu}_2(\left(\bx,\binfty\right))
                =\sum_{\ell=1 }^d \E\left[\mu_1\left(\left\{\bz\in\R_+^d:\,\min\left(\frac{a_{1\ell}}{x_1},\frac{a_{2\ell}}{x_2}\right)z_{\ell}>1\right\}\right)\right].
        \end{eqnarray*}
    \item[(c)] Suppose  $\max_{\ell\in \mathbb{I}_d}\P(\min\{a_{1\ell},a_{2\ell}\}>0)=0$ and  $\E\|\bA\|^{\alpha_2+\epsilon}<\infty$ for some $\epsilon>0$. Then $i_2^*=2$  and for $\bx\in(\bzero,\binfty)$,
    \begin{eqnarray*}
        \overline{\mu}_2((\bx,\infty))=\sum_{\ell,j=1}^d 
        \E\left[\mu_2\left(\left\{\bz\in\R_+^d:
        \frac{a_{1\ell}}{x_1}z_{\ell}>1,
        \frac{a_{2j}}{x_2}z_{j}>1 \right\}\right)\right].
    \end{eqnarray*}
\end{itemize}
In particular, if $\bZ$ is exchangeable then each measure in (a), (b), (c) is non-zero.
\end{proposition}

Note in this proposition  it is sufficient to assume 
$\bZ \in \MRV(\alpha_{i},b_{i}, \mu_{i},\E^{(i)}_d)$ for $i=1,2$ instead of $i=1,\ldots,d$ as in \Cref{thm:main:dfk} because $\bX$ is a bivariate random vector.

 Next, we provide sufficient conditions for the limit measures in \Cref{exchangable} to be positive so that our probability approximations for $\bX=\bA\bZ$ belonging to some extreme rectangular set are non-trivial. These approximations turn out to be sufficient for obtaining the asymptotic behavior of CoVaR under the given assumptions as well.

\begin{proposition} \label{prop:condprobmain}
Let the assumptions of \Cref{exchangable} hold. Let $\alpha:=\alpha_1$. Suppose further that $\bZ$ is completely tail equivalent and $\E\|\bA\|^{\alpha+\epsilon}<\infty$
    for some $\epsilon>0$.
Then the following hold:
 \begin{itemize}
    \item[(a)] We have $\bX=(X_1,X_2)^\top\in\MRV(\alpha,b_1,\overline{\mu}_1,\E_2^{(1)})$ where
    \begin{eqnarray*}
            \overline{\mu}_1
            ([\bzero,\bx]^c)=\frac{\mu_1([\bzero,\bone]^c)}{d}\sum_{\ell=1 }^d \E\left[%\left(\frac{a_{1\ell}}{x_1}\right)^{\alpha}+\left(\frac{a_{2\ell}}{x_2}\right)^{\alpha}-
            \max\left\{\frac{a_{1\ell}}{x_1},\frac{a_{2\ell}}{x_2}\right\}^{\alpha}%\mathds{1}_{\{a_{1\ell}>0,a_{2\ell}>0\}}
            \right], \quad \bx\in\R_+^2.
        \end{eqnarray*}
    \item[(b)]  If  $\max_{\ell\in \mathbb{I}_d}\P(\min\{a_{1\ell},a_{2\ell}\}>0)>0$, then $\bX\in\MRV(\alpha,b_1,\overline{\mu}_2,\E_2^{(2)})$  where
    \begin{eqnarray*}
            \overline\mu_2(\left(\bx,\binfty\right))
                =\frac{\mu_1([\bzero,\bone]^c)}{d}\sum_{\ell=1}^d \E\left[\min\left\{\frac{a_{1\ell}}{x_1},\frac{a_{2\ell}}{x_2}\right\}^{\alpha}%\mathds{1}_{\{a_{1\ell}>0,a_{2\ell}>0\}}
                \right], \quad \bx\in(\bzero,\binfty).
        \end{eqnarray*}
        Moreover, as $t\to\infty$,
        \begin{eqnarray} \label{eq5}
        \P(X_1>tx_1|X_2>tx_2) \sim %x_2^{\alpha}\, \frac{\overline\mu_2(\left(\bx,\binfty\right))}{\overline{\mu}_1
          %  (\R\times(x_2,\infty))} =
             x_2^{\alpha}\,
           {\sum\limits_{\ell=1}^d\E\left[\min\left\{\frac{a_{1\ell}}{x_1},\frac{a_{2\ell}}{x_2}\right\}^{\alpha}\right]}\Bigg/{\sum\limits_{\ell=1}^d\E\big[ a_{2\ell}^{\alpha} \big]}.
          %  \P(X_2>tx_2|X_1>tx_1)\sim x_1^{\alpha}
         %   \frac{\sum\limits_{\ell=1}^d\E\left[\min\left\{\frac{a_{1\ell}}{x_1},\frac{a_{2\ell}}{x_2}\right\}^{\alpha}\right]}{\sum\limits_{\ell=1}^d\E\left[ a_{1\ell}^{\alpha} \right]}, \quad t\to\infty.
        \end{eqnarray}
{With 
\begin{eqnarray*}
    h(y)= \frac{\mu([\bzero,\bone]^c)}{d}\sum_{\ell=1}^d \E\left[\min\left\{\frac{a_{1\ell}}{y},\frac{a_{2\ell}}{K^{1/\alpha}}\right\}^{\alpha}%\mathds{1}_{\{a_{1\ell}>0,a_{2\ell}>0\}}
                \right]
\end{eqnarray*}
where 
\begin{eqnarray*}
    K=\sum\limits_{\ell=1}^d a_{2\ell}^\alpha\Bigg/{\sum\limits_{\ell=1}^d a_{1\ell}^\alpha} \quad\quad \text{ and } \quad\quad c=\frac{\mu([\bzero,\bone]^c)}{d}\sum\limits_{\ell=1}^d a_{1\ell}^\alpha
\end{eqnarray*}
we receive for $0<\upsilon<\min(h(0)/c,1)$, that}
\begin{eqnarray*}
    \CoVaR_{\upsilon|\gamma}(X_1|X_2) \sim    { h^{-1}(\upsilon c)\VaR_{\gamma}(X_1)}, \quad \gamma\downarrow 0.
\end{eqnarray*}
        Additionally, if the non-zero components of $\bA$ have a bounded support, bounded away from zero, then there exists  $\upsilon^*\in (0,1)$ such that for all $0<\upsilon<\upsilon^*$ we have as $\gamma\downarrow 0$,
        \begin{eqnarray} \label{ckr}
            \CoVaR_{\upsilon|\gamma}(X_1|X_2)\sim \upsilon^{-\frac{1}{\alpha}} %\left({\sum\limits_{\ell=1}^d\E\big[ a_{1\ell}^{\alpha} \big]}\Bigg/{\sum\limits_{\ell=1}^d\E\big[ a_{2\ell}^{\alpha} \big]}\right)^{\frac{1}{\alpha}}
            \VaR_{\gamma}(X_1).
        \end{eqnarray}
        Finally, $\text{ECI}(X_1|X_2)=\infty$. 
\end{itemize}
\end{proposition}

\begin{remark}~  
\begin{enumerate}[(a)]
\item The multivariate regular variation of $\bX$ on $\E_2^{(1)}$ is a      direct consequence of the multivariate version of Breiman's         Theorem~\cite{breiman:1965} in \citep[Proposition   
A.1]{basrak:davis:mikosch:2002a}; see as well 
    \cite[Proposition~3.1]{kley:kluppelberg:reinert:2016}. 
 \item If  $Z_1,\ldots,Z_d$ are completely tail equivalent with the marginals being exactly Pareto distributions and  $\bZ=(Z_1,\ldots, Z_d)^\top\in \MRV(\alpha_1=\alpha,b_1,\mu_1,\E_2^{(1)})$ with $\mu_1$ having only mass on the axes (which is satisfied if $\alpha_1<\alpha_2$ ),  statement \eqref{eq5} is a special case of \cite[Corollary 2.6]{kley:kluppelberg:reinert:2017}
  and \eqref{ckr} of \cite[Theorem 3.4]{kley:kluppelberg:reinert:2017}. The CoVaR results obtained in \cite{kley:kluppelberg:reinert:2017} address the specific case where  $\text{ECI}(X_1|X_2)=\infty$.
  \item[(c)]  If bank/agent 1 (with risk variable $X_1$) is connected to object $j$ (with risk variable $Z_j$), then by defining $X_2=Z_j$ we can directly apply \Cref{prop:condprobmain}
  to obtain the tail behavior of $(X_1,Z_j)$ and $(Z_j,X_1)$ respectively. Hence we can get the asymptotic behavior of $ \CoVaR_{\upsilon|\gamma}(X_1|Z_j)$ and  $\CoVaR_{\upsilon|\gamma}(Z_j|X_1)$; in other words, we are able to measure the risk contagion between agents and theirs connected objects. If additional new objects are introduced into the market (i.e., $d$ increases), this has no impact on the CoVaR behavior between the agent and the object even if it connects to this new object.
\end{enumerate}
\end{remark}
  
  Incidentally, if $\max_{\ell \in \mathbb{I}_d}\P(\min\{a_{1\ell},a_{2\ell}\}>0)=0$, reflecting that the two banks/agents are not connected to same asset/object at the same time,  the right-hand side of \eqref{eq5} turns out to be 0, $\text{ECI}(X_1|X_2)$ becomes finite leading to the asymptotically independent model and the $\CoVaR$ approximation of \Cref{prop:condprobmain} is not valid anymore. %; cf  \citep[Theorem~3.2(b)]{kley:kluppelberg:reinert:2016} for a case where $\max_{\ell\in \mathbb{I}_d}\P(\min\{a_{1\ell},a_{2\ell}\}>0)=0$ will render a similar approximation negligible.
In the next few results, we concentrate on the case where $\max_{\ell\in \mathbb{I}_d}\P(\min\{a_{1\ell},a_{2\ell}\}>0)=0$, which relates to a scenario where the aggregate returns of financial entities are represented by disjoint sets of assets.
For example, this covers as well the case where we want to understand the influence of an asset/object on a bank/agent with which it is not connected. Note that although the two banks/agents may be represented by almost surely disjoint assets, yet they are related by the dependence assumption on the underlying set of assets whose risk is given by $\bZ$. We provide explicit tail rates and eventually also CoVaR computations under such a setup, assuming different dependence structures for the underlying random vector $\bZ$. % We will also find that in these cases $\text{ECI}(X_1|X_2)<\infty$.
%Next, we will study  in detail this interesting case where  $\max_{j\in \mathbb{I}_d}\P(\min\{a_{1j},a_{2j}\}>0)=0$ and in particular, we present the explicit tail rate.

%\marginpar{\VF{How do we define tail-equivalent? Do we use again completely tail equivalent?}\DAS{Perhaps complete is better to avoid more confusion with notation}}

\begin{proposition}[i.i.d. case] \label{prop:condprobiid}
Let $\bZ\in\R_+^d$ be a random vector with i.i.d. components $Z_1,\ldots,Z_d$ with distribution function $F_\alpha$ where $\ov F_\alpha \in\RV_{-\alpha}$, $\alpha>0$, $b_1(t)=F_{\alpha}^{\leftarrow}\left(1-1/t\right)$  and $b_i^{\leftarrow}(t)={(b_1^{\leftarrow}(t))^{i}}$.
Further, let $\bA \in \R_+^{2 \times d}$ be a  random matrix  with almost surely no trivial rows, independent of $\bZ$ and
 suppose  $\max_{\ell\in \mathbb{I}_d}\P(\min\{a_{1\ell},a_{2\ell}\}>0)=0$ and  $\E\|\bA\|^{2\alpha+\epsilon}<\infty$ for some $\epsilon>0$. Then $(X_1,X_2)^\top\in\MRV(2\alpha,b_2,\overline{\mu}_2,\E_2^{(2)})$ where
    \begin{eqnarray*}
        \overline{\mu}_2((\bx,\infty))=%\frac{1}{\binom{d}{2}}
        (x_1x_2)^{-\alpha}\sum\limits_{\ell,j=1}^d  \E\big[a_{1\ell}^{\alpha}a_{2j}^{\alpha}\big], \quad \bx\in(\bzero,\binfty).
    \end{eqnarray*}
Moreover, as $t\to\infty$,
        \begin{eqnarray*}
            \P(X_1>tx_1|X_2>tx_2)\sim  (b_1^{\leftarrow}(t))^{-1}\, x_1^{-\alpha}\sum\limits_{\ell,j=1}^d  \E\big[a_{1\ell}^{\alpha}a_{2j}^{\alpha}\big]\Bigg/\sum\limits_{j=1}^d\E\left[
            a_{2 j}^{\alpha}\right].
        \end{eqnarray*}
        Additionally, for $0<\upsilon<1$ we have as $\gamma\downarrow 0$,
        \begin{eqnarray*}
            \CoVaR_{\upsilon\gamma|\gamma}(X_1|X_2)\sim \upsilon^{-\frac{1}{\alpha}} \frac{\left(\sum\limits_{\ell,j=1}^d  \E\big[a_{1\ell}^{\alpha}a_{2j}^{\alpha}\big]\right)^{\frac{1}{\alpha}}}{\left(\sum\limits_{\ell=1}^d\E\left[
            a_{1\ell}^{\alpha}\right]\sum\limits_{j=1}^d\E\left[
            a_{2j}^{\alpha}\right]\right)^{\frac{1}{\alpha}}}\VaR_{\gamma}(X_1).
        \end{eqnarray*}
       % \marginpar{\VF{I added the power $-2$}}
        Finally, $\text{ECI}(X_1|X_2)=1$. 
\end{proposition}

\begin{remark}
 The arrival of new independent assets/objects on the financial market has no influence on the ECI and {the asymptotic behavior of CoVaR. Regardless of whether the new asset/object connects to one of the banks/agents}, we still have $\CoVaR_{\upsilon\gamma,\gamma}(X_1|X_2)=O(\VaR_{\gamma}(X_1))$ as $\gamma \downarrow 0$. If at least one agent/bank happens to connect to this new object, only the finite positive limit $\lim_{\gamma\downarrow 0}\CoVaR_{\upsilon\gamma|\gamma}(X_1|X_2)/\VaR_{\gamma}(X_1)$ may change.
\end{remark}

\begin{proposition}[Marshall-Olkin dependence] \label{prop:condprobMO}
Let $\bZ \in \PMOC(\alpha,\theta,\Lambda)$
and \linebreak $\bA \in \R_+^{2 \times d}$ be a  random matrix  with almost surely no trivial rows, independent of $\bZ$ and
 $\max_{\ell\in \mathbb{I}_d}\P(\min\{a_{1\ell},a_{2\ell}\}>0)=0$.
\begin{itemize}
    \item[(a)] Suppose $\bZ\in\PMOC(\alpha,\theta,\lambda^{=})$   and  $\E\|\bA\|^{\frac{3\alpha}{2}+\epsilon}<\infty$ for some $\epsilon>0$. Then \linebreak $\bX=\bA\bZ\in\MRV(\alpha_2,b_2,\overline{\mu}_2,\E_2^{(2)})$ where $\alpha_2=
    \frac{3\alpha}{2}$, {$b_2(t)=\theta^{\frac{1}{\alpha}} t^{\frac{2}{3\alpha}}$} and
    \begin{eqnarray*}
        \overline{\mu}_2((\bx,\infty))=%\frac{1}{\binom{d}{2}}
       \sum\limits_{\ell,j=1}^d  \E\left[\min\left\{\frac{a_{1\ell}}{x_1}, \frac{a_{2j}}{x_2}\right\}^{\alpha}\max\left\{\frac{a_{1\ell}}{x_1}, \frac{a_{2j}}{x_2}\right\}^{\alpha/2}\right], \quad \bx\in(\bzero,\binfty).
    \end{eqnarray*}
    Moreover, as $t\to\infty$,
        \begin{eqnarray*}
        \P(X_1>tx_1|X_2>tx_2)\sim (\theta t^{-\alpha})^{\frac{1}{2}}x_2^{\alpha}\,\overline{\mu}_2((\bx,\infty))\,\Big(\sum\limits_{j=1}^d\E\left[
            a_{2j}^{\alpha}\right]\Big)^{-1}.
        %    \P(X_2>tx_2|X_1>tx_1)\sim (\theta t^{-\alpha})^{\frac{1}{2}}x_1^{\alpha}
        %{x_1}, \frac{a_{2l}}{x_2}\right)^{\alpha/2}\right]}{\sum\limits_{j=1}^d\E\left[%\left(\frac{a_{1j}}{x_1}\right)^{\alpha}+\left(\frac{a_{2j}}{x_2}\right)^{\alpha}-
        %     a_{1j} ^{\alpha}%\mathds{1}_{\{a_{1j}>0,a_{2j}>0\}}
        %    \right]}.
        \end{eqnarray*}
       Additionally, if the non-zero components of $\bA$ have bounded support,  bounded away from zero, then there exists a $0<\upsilon^*_1<\upsilon^*_2<\infty$ such that for all $0<\upsilon<\upsilon_1^*$ we have as $\gamma\downarrow 0$, 
        \begin{eqnarray*}
         \CoVaR_{\upsilon\gamma^{\frac{1}{2}}|\gamma}(X_1|X_2)\sim \upsilon^{-\frac{1}{\alpha}}\frac{\left(\sum\limits_{\ell,j=1}^d\E\big[ a_{1\ell}^{\alpha}  a_{2j}^{\alpha/2}\big]\right)^{\frac{1}{\alpha}}}{\left(\sum\limits_{\ell=1}^d\E\big[ a_{1\ell}^{\alpha} \big]\right)^{\frac{1}{\alpha}}\left(\sum\limits_{j=1}^d\E\big[ a_{2j}^{\alpha} \big]\right)^{\frac{1}{2\alpha}}}\VaR_{\gamma}(X_1),
            %\CoVaR_{\upsilon\gamma^{\frac{1}{2}}|\gamma}(X_1|X_2)\sim \upsilon^{-\frac{1}{\alpha}}\left({\sum\limits_{\ell=1}^d\E\big[ a_{1\ell}^{\alpha}  a_{2\ell}^{\alpha/2}\big]}{\left(\sum\limits_{j=1}^d\E\big[ a_{2j}^{\alpha} \big]\right)^{-\frac{3}{2}}}\right)^{\frac{1}{\alpha}}\VaR_{\gamma}(X_2),
        \end{eqnarray*} 
        and for all $\upsilon_2^*<\upsilon<\infty$ we have as $\gamma\downarrow 0$,
        \begin{eqnarray*}
        \CoVaR_{\upsilon\gamma^{\frac{1}{2}}|\gamma}(X_1|X_2)\sim\upsilon^{-\frac{2}{\alpha}}\frac{\left(\sum\limits_{\ell,j=1}^d \E\big[ a_{1\ell}^{\alpha/2}  a_{2j}^{\alpha}\big]\right)^{\frac{2}{\alpha}}}{\left(\sum\limits_{\ell=1}^d\E\big[ a_{1\ell}^{\alpha} \big]\right)^{\frac{1}{\alpha}}\left(\sum\limits_{j=1}^d\E\big[ a_{2j}^{\alpha} \big]\right)^{\frac{2}{\alpha}}}\VaR_{\gamma}(X_1).
           % \CoVaR_{\upsilon\gamma^{\frac{1}{2}}|\gamma}(X_1|X_2)\sim \upsilon^{-\frac{1}{2\alpha}}\left({\sum\limits_{\ell=1}^d\E\big[ a_{2\ell}^{\alpha}  a_{1\ell}^{\alpha/2}\big]}{\left(\sum\limits_{j=1}^d\E\big[ a_{2j}^{\alpha} \big]\right)^{-\frac{3}{2}}}\right)^{\frac{1}{2\alpha}}\VaR_{\gamma}(X_2).
        \end{eqnarray*}
        Finally, $\text{ECI}(X_1|X_2)=2$. 
    \item[(b)] Suppose $\bZ\in\PMOC(\alpha,\theta,\lambda^{\propto})$
    and    $\E\|\bA\|^{\alpha\frac{3d+2}{2(d+1)}+\epsilon}<\infty$ for some $\epsilon>0$. Then \linebreak $\bX=\bA\bZ\in\MRV(\alpha_2,b_2,\overline{\mu}_2,\E_2^{(2)})$ where $\alpha_2=
    \alpha\frac{3d+2}{2(d+1)}$,  {$b_2(t)=\theta^{\frac{1}{\alpha}} t^{\frac{2(d+1)}{(3d+2)\alpha}}$} and %\VF{$b_2^*=(d\theta t)^{1/\alpha_2^*}$ and}
    \begin{eqnarray*}
        \overline{\mu}_2((\bx,\infty))=%\frac{1}{\binom{d}{2}}
        \sum\limits_{\ell,j=1}^d  \E\left[\min\left\{\frac{a_{1\ell}}{x_1}, \frac{a_{2j}}{x_2}\right\}^{\alpha}\max\left\{\frac{a_{1\ell}}{x_1}, \frac{a_{2j}}{x_2}\right\}^{\alpha\frac{d}{2(d+1)}}\right], \quad \bx\in(\bzero,\binfty).
    \end{eqnarray*} 
    Moreover, as $t\to\infty$,
        \begin{eqnarray*}
           \P(X_1>tx_1|X_2>tx_2)\sim (\theta t^{-\alpha})^{\frac{d}{2(d+1)}}x_2^{\alpha}\,\overline{\mu}_2((\bx,\infty))\,\Big(\sum\limits_{j=1}^d\E\left[
            a_{2j}^{\alpha}\right]\Big)^{-1}.
           % \frac{\sum\limits_{j,l=1 }^d \E\left[\min\left(\frac{a_{1j}}{x_1}, \frac{a_{2l}}{x_2}\right)^{\alpha}\max\left(\frac{a_{1j}}{x_1}, \frac{a_{2l}}{x_2}\right)^{\alpha\frac{d}{2(d+1)}}\right]}{\sum\limits_{j=1}^d\E\left[
           % a_{1j}^{\alpha}\right]%\mathds{1}_{\{a_{1j}>0,a_{2j}>0\}}}.
        \end{eqnarray*}
         Additionally, if the non-zero components of $\bA$ have bounded support, bounded away from zero, then there exists a $0<\upsilon^*_1<\upsilon^*_2<\infty$ such that for all $0<\upsilon<\upsilon_1^*$ we have as $\gamma\downarrow 0$,
        \begin{eqnarray*}
            \CoVaR_{\upsilon\gamma^{\frac{d}{2(d+1)}}|\gamma}(X_1|X_2)\sim \upsilon^{-\frac{1}{\alpha}}  \frac{\left(\sum\limits_{\ell,j=1}^d \E\big[ a_{1\ell}^{\alpha}  a_{2j}^{\frac{d\alpha}{2(d+1)}}\big]\right)^{\frac{1}{\alpha}}}{\Big(\sum\limits_{\ell=1}^d\E\big[ a_{1\ell}^{\alpha} \big]\Big)^{\frac{1}{\alpha}}\Big(\sum\limits_{j=1}^d\E\big[ a_{2j}^{\alpha} \big]\Big)^{\frac{d}{2(d+1)\alpha}}}\VaR_{\gamma}(X_1),
        \end{eqnarray*} 
        and for all $\upsilon_2^*<\upsilon<\infty$ we have %with $\beta=\frac{d}{2(d+1)}$ 
        as $\gamma\downarrow 0$,
        \begin{eqnarray*}
            \CoVaR_{\upsilon\gamma^{\frac{d}{2(d+1)}}|\gamma}(X_1|X_2)\sim \upsilon^{-\frac{2(d+1)}{d\alpha}} \frac{\left({\sum\limits_{\ell,j=1}^d \E\big[ a_{1\ell}^{\frac{d \alpha}{2(d+1)}} a_{2j}^{\alpha}  }\big]\right)^{\frac{2(d+1)}{d\alpha}}}{\Big(\sum\limits_{\ell=1}^d\E\big[ a_{1\ell}^{\alpha} \big]\Big)^{\frac{1}{\alpha}}\Big(\sum\limits_{j=1}^d\E\big[ a_{2j}^{\alpha} \big]\Big)^{2\frac{d+1}{d\alpha}}}\VaR_{\gamma}(X_1).
        \end{eqnarray*} 
       Finally,  $\text{ECI}(X_1|X_2)=2+\frac2d$. 
   \end{itemize}
\end{proposition}

\begin{remark} $\mbox{}$
\begin{itemize}
\item[(a)] If $\bZ\in\PMOC(\alpha,\theta,\lambda^{\propto})$, then as the number of assets/objects, i.e. $d$, increases, $\text{ECI}(X_1|X_2)=2+2/d$ decreases and hence, the rate of convergence of $\CoVaR_{\gamma,\gamma}(X_1|X_2)$ to infinity as $\gamma\downarrow 0$ becomes slower. In this case, the dependence of $X_1$ and $X_2$ in the tails gets weaker as the network becomes more diversified. It is important to note here that the $\PMOC(\alpha,\theta,\lambda^{\propto})$ model does not have a nested structure in contrast to the $\PMOC(\alpha,\theta,\lambda^{=})$ which is nested. In the $\PMOC(\alpha,\theta,\lambda^{=})$ model, the increase in the value of $d$ has no influence on the ECI value.
\item[(b)] Under the assumption of independent objects  $\text{ECI}(X_1|X_2)=1$  is less than $\text{ECI}(X_1|X_2)=2$ if the objects have a Marshall-Olkin dependence structure with equal parameters (\Cref{prop:condprobMO}(a)) which is again less than the $\text{ECI}(X_1|X_2)=2+2/d$ for the Marshall-Olkin copula with proportional parameters (\Cref{prop:condprobMO}(b)), reflecting that the dependence in the tails gets progressively stronger. 
\item[(c)] Note that the degree of tail dependence plays a role only if $\bA$ contains a column with all entries zero, otherwise the entries of $\bA$ do not influence ECI. In other words, although it is important to know if agents are connected to the same object, the weights or magnitude of the connections are not essential for the value of ECI. However, this is not always the case as the following example of a Gaussian copula shows where the location (index) of the zero column is important as well. Note that if $\bZ\in \PGC(\alpha, \theta, \Sigma)$ and $\Sigma$ is not an equicorrelation matrix, then $\bZ$ is not exchangeable in contrast to the other examples above (assuming identical margins).
\end{itemize}
\end{remark}

\begin{proposition}[Gaussian copula] \label{prop:condprobGC}
Let $\bZ\in \PGC(\alpha, \theta, \Sigma)$ with $\Sigma=(\rho_{\ell j})_{1\leq \ell,j\leq d}$ positive definite.
Suppose $\bA \in \R_+^{2 \times d}$ is a  random matrix  with almost surely no trivial rows, independent of $\bZ$ and
 $\max_{\ell\in \mathbb{I}_d}\P(\min\{a_{1\ell},a_{2\ell}\}>0)=0$.
 Also, define
 \begin{align*}
     \rhom & =\max\left\{\rho_{\ell j}: \ell, j \in \mathbb{I}_d, \ell\neq j  \right\},\\
      \rho^* & =\max\left\{\rho_{\ell j}: \ell, j \in \mathbb{I}_d, \ell\neq j \text{ and }\, \P(\min(a_{1\ell}, a_{2j})>0)>0 \right\}.
   %  \rho^* & =\max\{\rho_{jl}: j,l \in \mathbb{I}_d, j\neq l \; \text{ and }\, \P(\min(a_{1j}, a_{2l}>0)>0}\}
 \end{align*}
 %$$\rhom=\max_{ j, l \in \mathbb{I}_d, j\neq l  }\rho_{jl} \quad \text{ and }\quad  \rho^*=\max_{\genfrac{}{}{0pt}{}{ j, l \in \mathbb{I}_d, j\neq l, }{\P(\min(a_{1j}, a_{2l}>0)>0}} \rho_{jl}.$$
 \begin{itemize}
    \item[(a)] Suppose $\rho^*=\rhom$ and $\E\|\bA\|^{\frac{2\alpha}{1+\rhom}+\epsilon}<\infty$ for some $\epsilon>0$.
    Then we have \linebreak $\bX=\bA\bZ\in\MRV(\alpha_2,b_2,\overline{\mu}_2,\E_2^{(2)})$ with 
    %\marginpar{\VF{I think we have to move the $\theta$ in $b_2$ as I did, bu then delete it then in the constant $C({\rhom, \alpha, \theta})$ (which I haven't done yet) because otherwise the later results are not fitting. The CoVaR does not depend on $\theta$.} \DAS{Agreed. I will check tomorrow.} \VF{I changed it on the way I think it is fitting}}
    \begin{align*}
        & \alpha_2 =\frac{2\alpha}{1+\rhom}, \quad\quad
        b_2^{\leftarrow}(t)= C({\rhom, \alpha}) (\theta t^{-\alpha})^{-\frac{2}{1+\rhom}}(\log t)^{\frac{\rhom}{1+\rhom}} ,\\
    & \overline{\mu}_2((\bx,\infty)) = D({\rhom,\alpha, \bA})(x_1x_2)^{-\frac{\alpha}{1+\rhom}}, \quad \bx\in(\bzero,\binfty),
    %\sum_{\genfrac{}{}{0pt}{}{j,l=1 }{\rho_{jl}=\rhom}}^d\E(a_{1j}a_{2l})^{\frac{\alpha}{1+\rhom}} >0, \quad \bx\in(\bzero,\binfty).
    \end{align*}
    where for $\rho\in(-1,1), \alpha>0, \theta>0$ and $\bA\in \R_+^{2\times d}$, we define
   \begin{equation}
         \begin{aligned}\label{eq:Cgauss}    
    C(\rho, \alpha) & = { ({2\pi})^{-\frac{1}{1+\rho}}} (2\alpha)^{\frac{\rho}{1+\rho}}, \\
    D(\rho, \alpha, \bA) & = \frac{1}{2\pi} \frac{(1+\rho)^{3/2}}{(1-\rho)^{1/2}}\sum_{(\ell,j):\, \rho_{\ell j}=\rho }\E\big[a_{1\ell}^{\frac{\alpha}{1+\rho}}a_{2j}^{\frac{\alpha}{1+\rho}}\big].  
    \end{aligned}
   \end{equation}
 Moreover, as $t\to\infty$,
        \begin{align*}
        &  \P(X_1>x_1t|X_2>x_2t) \\
      &\quad \quad\quad\sim   (\theta t^{-\alpha})^{\frac{1-\rhom}{1+\rhom}}(\log t)^{-\frac{\rhom}{1+\rhom}}x_1^{-\frac{\alpha}{1+\rhom}}x_2^{\frac{\alpha\rhom}{1+\rhom}}\frac{C(\rhom,\alpha)^{-1}D(\rhom,\alpha, \bA)}{\sum\limits_{j=1}^d\E[a_{2j}^{\alpha}]}.
    \end{align*}
     Additionally, with $g(\gamma)=\gamma^{\frac{1-\rhom}{1+\rhom}}(-\alpha^{-1}\log \gamma)^{-\frac{\rhom}{1+\rhom}}$ and $0<\upsilon<\infty$ we have as $\gamma\downarrow 0$,
     \begin{align*}
       \CoVaR_{\upsilon g(\gamma)|\gamma}& (X_1|X_2)  
       \sim \upsilon^{-\frac{1+\rhom}{\alpha}} \frac{\left({C(\rhom,\alpha)^{-1}D(\rhom,\alpha, \bA)}\right)^{\frac{1+\rhom}{\alpha}}}{\Big(\sum\limits_{\ell=1}^d\E\big[ a_{1\ell}^{\alpha} \big]\sum\limits_{j=1}^d\E[a_{2j}^{\alpha}]\Big)^{\frac{1}{\alpha}}}\VaR_{\gamma}(X_1).
      %  & \CoVaR_{\upsilon g(\gamma)|\gamma}(X_1|X_2)  \sim \upsilon^{-\frac{1+\rhom}{\alpha}} \left(\frac{C(\rhom,\alpha,\theta)^{-1}D(\rhom,\alpha, \bA)}{\VF{\left(\sum\limits_{j=1}^d\E[a_{2j}^{\alpha}]\right)^{\frac{2}{1+\rhom}}}}\right)^{\frac{1+\rhom}{\alpha}}\VaR_{\gamma}(X_2).
     \end{align*}
     Finally,  $g(t^{-1})\in \RV_{-\frac{1-\rhom}{1+\rhom}}$ and hence, $\text{ECI}(X_1|X_2)=\frac{1+\rhom}{1-\rhom}$.
     \item[(b)]  Suppose $\rho^*<\rhom$ and $\E\|\bA\|^{\frac{2\alpha}{1+\rho^*}+\epsilon}<\infty$ for some $\epsilon>0$.  Then  $\overline{\mu}_2$
     as defined in \Cref{thm:main:dfk} is identically zero.
     But we have  $\bX=\bA\bZ\in \MRV(\alpha_2^*,b_2^*,\overline{\mu}_2^*,\E_2^{(2)})$ with
        \begin{align*}
        & \alpha_2^* =\frac{2\alpha}{1+\rho^*}, \quad\quad
        b_2^{*\leftarrow}(t)= C({\rho^*, \alpha}) (\theta t^{-\alpha})^{-\frac{2}{1+\rho^*}}(\log t)^{\frac{\rho^*}{1+\rho^*}} ,\\
    & \overline{\mu}_2^*((\bx,\infty)) = D({\rho^*,\alpha, \bA})(x_1x_2)^{-\frac{\alpha}{1+\rho^*}}, \quad \bx\in(\bzero,\binfty),
    \end{align*}
  and, as $t\to\infty$,
        \begin{align*}
        & \P(X_1>x_1t|X_2>x_2t) \\ 
        &\quad \quad\quad\sim   (\theta t^{-\alpha})^{\frac{1-\rho^*}{1+\rho^*}}(\log t)^{-\frac{\rho^*}{1+\rho^*}}x_1^{-\frac{\alpha}{1+\rho^*}}x_2^{\frac{\alpha\rho^*}{1+\rho^*}}\frac{C(\rho^*,\alpha)^{-1}D(\rho^*,\alpha, \bA)}{\sum\limits_{j=1}^d\E[a_{2j}^{\alpha}]},
    \end{align*} 
    where $C(\cdot)$ and $D(\cdot)$ are as defined in \eqref{eq:Cgauss}.\\
     Additionally, with $g(\gamma)=\gamma^{\frac{1-\rho^*}{1+\rho^*}}(-\alpha^{-1}\log \gamma)^{-\frac{\rho^*}{1+\rho^*}}$ and $0<\upsilon<\infty$ we have as $\gamma\downarrow 0$,
     \begin{align*}
        \CoVaR_{\upsilon g(\gamma)|\gamma} & (X_1|X_2)  \sim \upsilon^{-\frac{1+\rho^*}{\alpha}} \frac{\left({C(\rho^*,\alpha)^{-1}D(\rho^*,\alpha, \bA)}\right)^{\frac{1+\rho^*}{\alpha}}} {\Big(\sum\limits_{\ell=1}^d\E\big[ a_{1\ell}^{\alpha} \big]\sum\limits_{j=1}^d\E[a_{2j}^{\alpha}]\Big)^{\frac{1}{\alpha}}} \VaR_{\gamma}(X_1).
     \end{align*}
    Finally,  $g(t^{-1})\in \RV_{-\frac{1-\rho^*}{1+\rho^*}}$ and hence, $\text{ECI}(X_1|X_2)=\frac{1+\rho^*}{1-\rho^*}$.
\end{itemize}
\end{proposition}

\begin{remark} \label{Remark 4.12} $\mbox{}$
\begin{itemize}
\item[(a)] Suppose $\rho^*<\rhom$. Then there exists $\ell^*,j^*\in \mathbb{I}$ with $\ell^*\not=j^*$ and $$\rho^*<\rho_{\ell^*m^*}=\rho_{m^*\ell^*}\leq \rhom.$$ By definition of $\rho^*$ we have
\begin{eqnarray*}
    \P(\min(a_{1\ell^*}, a_{2j^*})>0)=0=\P(\min(a_{1j^*}, a_{2\ell^*})>0)
\end{eqnarray*}
resulting in $a_{1\ell^*}=a_{1j^*}=a_{2\ell^*}=a_{2j^*}=0$ a.s. Hence, the $\ell$-th and $j$-th column of \linebreak $\bA\in\R_+^{2\times d}$
are a.s. zero columns.
%    \item[(b)] \VF{Usually we have a higher dimensional model and compare some components\ldots}
\item[(b)] Both values, $\rho^*$ and $\text{ECI}(X_1|X_2)=\frac{1+\rho^*}{1-\rho^*}$, depend on the one hand, on the location of the non-zero entries of $\bA$, and on the other hand, on the dependence structure of the underlying object $\bZ$ modeled by the correlation matrix $\Sigma$. Naturally, the dependence in the network becomes stronger if either the pairwise correlations in $\Sigma$ increases or if the zero components of $\bA$ are replaced by non-zero components, i.e.,  we have more connections between the agents and the objects. Both effects might increase as well the tail dependence of the agents and the ECI.  %although they exhibit asymptotic independence
\item[(c)]  If additional objects are introduced on the market, that are not connected to the agents, they have no influence on the CoVaR behavior. However, if one agent is connected to that object the change in the CoVaR behavior depends on the correlation of this object with the other objects connected to the other agent in the Gaussian copula model.
\end{itemize}
\end{remark}

\begin{remark}
    Note that the limit measures $\overline{\mu}_1$ and $\overline{\mu}_2$ found in Propositions \ref{prop:condprobiid}-\ref{prop:condprobGC} are all non-zero measures and hence, the conditional probabilities computed are asymptotically non-trivial as well.
\end{remark}

%\begin{remark}
%In this section we obtained asymptotic conditional tail probabilities with two portfolios in a bipartite structure. In understanding systemic risk it is often important to understand the risk of one financial entity given that at least one other entity in the system is under stress; alternatively given an entity has high negative returns, what is the effect of this to the system, i.e., how does it affect even the best performing entity? Interesting this can be addressed by using the bivariate structure developed previously. Results on such network and CoVaR measurements follow similar to results n this section, we have detailed this in \Cref{sec:app:morethantwo}.

%\end{remark}

\subsection{{Risk contagion between various aggregates}}
%\subsection{Towards risk \VF{contagion}: joint tail risk of aggregates} 
In \Cref{subsec:bivprob} we obtained asymptotic conditional tail probabilities with two portfolios in a bipartite structure.
For a financial institution with more than two portfolios, it is also interesting to assess systemic risk between the entire system and a part of the system in terms of \emph{aggregate risks}. Given that the aggregate of a few risks, say for a set $T\subseteq\mathbb{I}_q$, given by $\sum_{m\in T} X_m$ are above their $\VaR$, say at level $\gamma$, given a measurable function $g:(0,1)\to (0,1)$,  what is the  $\VaR$ of $\sum_{k\in S} X_k$  for some $S\subseteq \mathbb{I}_q$ at level $\upsilon g(\gamma)$:
    $$\CoVaR_{\upsilon g(\gamma)| \gamma}\Big( \sum_{k\in S} X_k\Big|\sum_{m\in T} X_m\Big).$$
Of course, the following special cases will be of particular interest:
\begin{itemize}
  \item CoVaR of the $k$-th entity at level $\upsilon g(\gamma)$ given the aggregate of the system at level $\gamma$:
     \[\CoVaR_{\upsilon g(\gamma)|\gamma}\Big(X_k\Big|\sum_{m=1}^q X_m\Big), \]
    \item  CoVaR of the aggregate of the system at level $\upsilon g(\gamma)$ given a particular entity $k$ at level $\gamma$:
    \[ \CoVaR_{\upsilon g(\gamma)|\gamma}\Big(\sum_{m=1}^q X_m\Big|X_k\Big).\]
\end{itemize}
Incidentally, the results obtained in the previous sections suffice for computations of both the CoVaR asymptotics as well as for computing ECI. Indeed by defining $\be^S\in \R^q$ as
 $\be^S_S:=\bone_S$ and $\be^S_{S^c}:=\bzero_{S^c}$, which is a vector containing only $1$ and $0$ (similarly  $\be^T\in \R^q$), and finally, by  defining $\bA^*:=(\be^S,\be^T)^\top\bA$, we obtain
 \begin{eqnarray*}
    (Y_1,Y_2)^\top:=\left(\sum_{k\in S} X_k, \sum_{m\in T} X_m\right)^\top =(\be^S,\be^T)^\top\bA\bZ=\bA^*\bZ.
 \end{eqnarray*}
Thus, if $\bA^*$ and $\bZ$ satisfy the assumptions of \Cref{subsec:bivprob}, we can directly apply the associated tail probability results and CoVaR asymptotics. Since
\begin{eqnarray*}
     \max_{\ell\in \mathbb{I}_d}\P(\min\{a_{1\ell}^*,a_{2\ell}^*\}>0)= \max_{\ell\in \mathbb{I}_d}\P\left(\min\left\{\sum_{k\in S}a_{k\ell}^*,\sum_{m\in T}a_{m\ell}^*\right\}>0\right),
\end{eqnarray*}
the essential characteristic  $\max_{\ell\in \mathbb{I}_d}\P(\min\{a_{1\ell}^*,a_{2\ell}^*\}>0)=0$ (or $>0$) for determining the MRV on $\mathbb{E}_2^{(2)}$  and hence, the CoVaR rate and ECI, is equivalent to
\begin{eqnarray*}
    \max_{\ell\in \mathbb{I}_d}\max_{k\in S}\max_{m\in T}\P\left(\min\left\{a_{k\ell},a_{m\ell}\right\}>0\right)=0 \,\, (\text{or} > 0).
\end{eqnarray*}
Thus, subject to the above conditions, the rest of the results follow directly.

%In the setting of \Cref{subsec:bivprob} with $\bX=\bA\bZ$ where    $\bZ\in \R_+^d$ is a completely tail equivalent random vector   and
%$\bZ \in \MRV(\alpha_{i},b_{i}, \mu_{i},\E^{(i)}_d), \forall \, i\in\mathbb{I}_d$, with
% $b_i(t)/b_{i+1}(t)\to \infty$ as $t\to\infty$ for $i=1,\ldots,d-1$ and 
% $\bA \in \R_+^{q \times d}$ is a  random matrix  with almost surely no trivial rows and independent of $\bZ$, this case can be handled as special case of \Cref{subsec:bivprob}. Indeed, define $\be^S\in \R^q$ as%
% $\be^S_S=\bone_S$ and $\be^S_{S^c}=\bzero_{S^c}$ and similarly  $\be^T\in \R^q$.  Defining $\bA^*=(\be^S,\be^T)^\top\bA$ we receive

%\subsection{Towards risk \VF{contagion}: joint tail risk of more than two portfolios}
\subsection{{Risk contagion in more than two portfolios}}\label{subsec:morethantwo} 
%\marginpar{\VF{Shall we move this section to the appendix?}}
 In this section, for a financial system with two or more portfolios, we assess the risk of one financial entity performing poorly given that at least one other entity in the system is under stress; alternatively given an individual entity has high negative returns, what is the effect of this to the entire system? We use the ideas from the results in the bivariate structure developed previously in \Cref{subsec:bivprob} to obtain the asymptotics here. Clearly, the sparsity of the matrix defining the bipartite network given by $\bA$ has a role to play. For this section, we concentrate on $\bZ\in \R_+^d$ which are completely tail equivalent.
\begin{proposition}\label{prop:condprobwithoneandmax}
    Let  $\bZ\in \R_+^d$ be a  completely tail-equivalent random vector and
$\bZ \in \MRV(\alpha_{i},b_{i}, \mu_{i},\E^{(i)}_d), \,\forall\,i\in\mathbb{I}_d$ with
 $b_i(t)/b_{i+1}(t)\to \infty$ as $t\to\infty$ for $i=1,\ldots,d-1$. 
Let $\bA \in \R_+^{q \times d}$ be a random matrix with almost surely no trivial rows and independent of $\bZ$.  Moreover, let $\E||\bA||^{\alpha+\epsilon}<\infty$ where $\alpha:=\alpha_1$. Define $\bX=\bA\bZ$. Then the following holds for a fixed $k\in \mathbb{I}_q$ and $\bY:=(Y_1,Y_2)^\top:=(X_k, \max_{m\in \mathbb{I}_q\setminus\{k\}} X_m )$.
\begin{itemize}
    \item[(a)] $\bY\in\MRV(\alpha,b_1,{\mu}_1^*,\E_2^{(1)})$ where
    \begin{eqnarray*}
            {\mu}_1^*
            ([\bzero,\bx]^c)=\frac{\mu_1([\bzero,\bone]^c)}{d}\sum_{\ell=1 }^d \E\left[%\left(\frac{a_{1\ell}}{x_1}\right)^{\alpha}+\left(\frac{a_{2\ell}}{x_2}\right)^{\alpha}-
            \max\left\{\frac{a_{k\ell}}{x_1}, \max_{m\in \mathbb{I}_q\setminus\{k\}} \left\{\frac{a_{m\ell}}{x_2}\right\}\right\}^{\alpha}%\mathds{1}_{\{a_{1\ell}>0,a_{2\ell}>0\}}
            \right], \quad \bx\in\R_+^2.
        \end{eqnarray*}
    \item[(b)]  If  $\max_{m\in\mathbb{I}_q\setminus\{k\}} \max_{\ell\in\mathbb{I}_d} \P(\min(a_{k\ell}, a_{m\ell})>0)>0,$ 
     then $\bY\in\MRV(\alpha,b_1,\mu_2^*,\E_2^{(2)})$  where
    \begin{eqnarray*}
            \mu_2^*(\left(\bx,\binfty\right))
                =\frac{\mu_1([\bzero,\bone]^c)}{d}\sum_{\ell=1}^d \E\left[\max_{m\in\mathbb{I}_q\setminus\{k\}} \left\{\min\left\{\frac{a_{k\ell}}{x_1},\frac{a_{m\ell}}{x_2}\right\}\right\}^{\alpha}%\mathds{1}_{\{a_{1\ell}>0,a_{2\ell}>0\}}
                \right], \quad \bx\in(\bzero,\binfty).
        \end{eqnarray*}
        Moreover, as $t\to\infty$,
        \begin{align*} 
        \P(Y_1>tx_1|Y_2>tx_2) %& \sim  x_2^{\alpha} \, \frac{\mu_2^*(\left(\bx,\binfty\right))}{{\mu}_1^*            (\R_+\times(x_2,\infty))}  \\
            & \sim  x_2^{\alpha}\,
           \frac{\sum\limits_{\ell=1}^d\E\left[\max_{m\in\mathbb{I}_q\setminus\{k\}} \left\{\min\left\{\frac{a_{k\ell}}{x_1},\frac{a_{m\ell}}{x_2}\right\}\right\}^{\alpha}%\mathds{1}_{\{a_{1\ell}>0,a_{2\ell}>0\}}
                \right]}{\sum\limits_{\ell=1}^d\E\big[\max_{m\in \mathbb{I}_q\setminus \{k\}} a_{m\ell}^{\alpha} \big]}, \\
            %    \intertext{and,}
          \P(Y_2>tx_2|Y_1>tx_1)% & \sim  x_1^{\alpha} \, \frac{\mu_2^*(\left(\bx,\binfty\right))}{{\mu}_1^*            ((x_1,\infty)\times\R_+)}  \\
            & \sim  x_1^{\alpha}\,
           \frac{\sum\limits_{\ell=1}^d\E\left[\max_{m\in\mathbb{I}_q\setminus\{k\}} \left\{\min\left\{\frac{a_{k\ell}}{x_1},\frac{a_{m\ell}}{x_2}\right\}\right\}^{\alpha}%\mathds{1}_{\{a_{1\ell}>0,a_{2\ell}>0\}}
                \right]}{\sum\limits_{\ell=1}^d\E\big[ a_{k\ell}^{\alpha} \big]}.     
        \end{align*}
           Additionally, if the non-zero components of $\bA$ have a bounded support, bounded away from zero, then there exist $0<\upsilon^*_1,\upsilon^*_2<1$ such that for all $0<\upsilon<\upsilon_1^*$ we have as $\gamma\downarrow 0$,
        \begin{eqnarray*}
            \CoVaR_{\upsilon|\gamma}(Y_1|Y_2)\sim \upsilon^{-\frac{1}{\alpha}} %\left({\sum\limits_{\ell=1}^d\E\left[a_{k\ell}\right]^{\alpha}%\mathds{1}_{\{a_{1\ell}>0,a_{2\ell}>0\}}
             %   }\Bigg/{\sum\limits_{\ell=1}^d\E\big[\max_{m\in\mathbb{I}_q\setminus\{k\}}a_{m\ell}^{\alpha} \big]}\right)^{\frac{1}{\alpha}}
             \VaR_{\gamma}(Y_1),
        \end{eqnarray*}
        and for all $0<\upsilon<\upsilon_2^*$ we have as $\gamma\downarrow 0$,
        \begin{eqnarray*}
            \CoVaR_{\upsilon|\gamma}(Y_2|Y_1)\sim \upsilon^{-\frac{1}{\alpha}}%\left({\sum\limits_{\ell=1}^d\E\big[\max_{m\in\mathbb{I}_q\setminus\{k\}}a_{m\ell}^{\alpha} \big]}\Bigg/{\sum\limits_{\ell=1}^d\E\left[a_{k\ell}\right]^{\alpha}%\mathds{1}_{\{a_{1\ell}>0,a_{2\ell}>0\}}
             %   }\right)^{\frac{1}{\alpha}}
            \VaR_{\gamma}(Y_2),
        \end{eqnarray*}
        Finally, $\text{ECI}(Y_1|Y_2)=\text{ECI}(Y_2|Y_1)=\infty$.
\end{itemize}
\end{proposition}

    The first part of the proof is a direct application of \Cref{thm:main:dfk} and the mapping theorem in \cite[Theorem~2.3]{lindskog:resnick:roy:2014}. The second part of the proof follows then from \Cref{thm:bivCoVaRmain}. Hence, the detailed proof of this proposition is omitted.%,  similarly we skip proofs of further results in this section.
\begin{remark}\label{rem:str}
   
Note that if for fixed $k\in \mathbb{I}_q$ we have 
\begin{align}\label{eq:strA}
\max_{m\in\mathbb{I}_q\setminus\{k\}} \max_{\ell\in\mathbb{I}_d} \P(\min(a_{k\ell}, a_{m\ell})>0)>0,
\end{align}
then the measure $\mu_2^*$ in \Cref{prop:condprobwithoneandmax} is defined via $\mu_1$, the limit measure of $\bZ$ on $\E_d^{(1)}$. On the other hand, if $$\max_{m\in\mathbb{I}_q\setminus\{k\}} \max_{\ell\in\mathbb{I}_d} \P(\min(a_{k\ell}, a_{m\ell})>0)=0,$$ then $\mu_2^*$ will as well involve $\mu_2$, the limit measure of $\bZ$ on $\E_d^{(2)}$; thus, the dependence structure of $\bZ$ plays a role here. %Since the proofs are directly obtained by applying \Cref{thm:main:dfk} with relevant assumptions on $\bZ$, they are omitted here.
The CoVaR and ECI results for $\bZ$ with a i.i.d, $\PGC$ and $ \PMOC$ dependence structure are provided in \Cref{sec:app:morethantwo} and follow a  similar pattern to the results obtained in \Cref{subsec:bivprob}. 

\end{remark}
% For the systemic events discussed above, we have stated a general result on CoVaR and ECI in this section, especially when $\bA$ satisfies \eqref{eq:strA}. When \eqref{eq:strA} is not satisfied we may obtain results  using particular choices of the vector $\bZ$ (i.e., i.i.d, $\PGC, \PMOC$); they follow a  similar pattern to the results obtained in \Cref{subsec:bivprob} and 

\section{Conclusion}\label{sec:concl}
%\textcolor{red}{we have to change the conclusion}

The goal of the paper has been to understand the tail risk behavior in complex financial network models and use this knowledge to provide asymptotic approximations for conditional risk measures like CoVaR. We have used the framework of a bipartite network to assess risk contagion in this regard. In our study, we have managed to accomplish a few things.
\begin{enumerate}[(i)]
%\item We have proposed the notion of \emph{mutual asymptotic independence} to assess tail dependence in multivariate models; this parallels mutual independence vis-a-vis pairwise independence for multivariate random vectors.
    \item  For modeling via bipartite networks, it is natural to assume that various banks will invest in different assets, possibly non-overlapping, yet the exposure of these assets to the market may still make them dependent {(regardless of whether they are asymptotically dependent or independent)}. We have shown that modeling via such distributions still allows us to compute conditional tail probabilities, which erstwhile computations have shown to be negligible. Moreover, the CoVaR measures computed can be drastically different depending on the nature of the behavior of the joint distribution of the assets. 
    %We have computed conditional risk probability and CoVaR in a complex financial network where underlying distribution is heavy-tailed and may exhibit asymptotic independence (either mutual or pairwise). We have shown explicit computation for various different CoVaR measures.
    \item We have proposed the \emph{Extreme CoVaR Index} which measures the {strength of risk} contagion between entities in a system, {this is especially useful when the risks of underlying objects are asymptotically independent}.
\end{enumerate}
%Regarding the notion of \emph{mutual asymptotic independence}, we believe this provides a new pathway for characterizing the extremal dependence for high dimensional problems relating to tail events. 
We have restricted to CoVaR measures but we surmise that other conditional risk measures like Mean Expected Shortfall (MES), Mean Marginal Excess (MME) and Systemic Risk (SRISK) can also be computed under such models; {we have briefly discussed some connections with MME and MES}. We have also particularly focused on Pareto or Pareto-like tails in this paper for convenience, naturally, some of the results can be extended to general regularly varying marginal tails as well. For Gaussian copulas with regularly varying tails, we still need certain restrictions on the marginal tail behavior (see \cite{das:fasen:2024} for details); on the other hand, for Marshall-Olkin copulas, assuming tail-equivalent regularly varying marginals will lead to similar results. 
In this paper, we focused on these particular copula models because of their flexibility and inherent connections with risk, reliability and specifically systemic risk.  Clearly, our results need not be restricted to these copulas only. Although bivariate copula families are usually popular, \cite[Chapter 4]{joe:2014} lists multiple extensions of bivariate copulas to general high dimensions. We believe many such copulas can be explored for creating models with particular levels of asymptotic independence as necessitated by the context.
%Moreover, we have only investigated a few popular and useful dependence structures like Gaussian, and Marshall-Olkin; there are other Archimedean copulas that can be explored using the methodology in this paper allowing for a rich class of models for such complex systems.
Finally, we have left model calibration and statistical estimation to be pursued in future work.

{
\subsection*{Acknowledgement}
We would like to thank the Editor, the Associate Editor and both referees for giving us valuable and informative comments which helped us in improving the quality of the paper.

%\newpage

\bibliographystyle{imsart-nameyear}
\bibliography{bibfilenew}

%\newpage

\appendix

\section{Proofs of Section~\ref{sec:mrvwithasyind}} \label{se2:proofs}

For the proof of \Cref{prop:bi} we use the following auxiliary result given in \cite[Lemma A.2]{das:fasen:2024b}
\begin{lemma} \label{lem:aux1}

Let $\Sigma\in\R^{d\times d}$ be positive definite and $\gamma(\Sigma)$, $I(\Sigma)$ be defined as in \Cref{lem:quadprog}.
\begin{enumerate}
    \item[(a)] Suppose $\Sigma^{-1}\bone>\bzero$. Then for any $S\subseteq \mathbb{I}_d$ with  $S\not= \mathbb{I}_d$, the following inequality holds:
    \begin{eqnarray*}
        \gamma(\Sigma)>\gamma(\Sigma_{S}).
    \end{eqnarray*}
    \item[(b)] Suppose $\Sigma^{-1}\bone\ngtr\bzero$. Then $I(\Sigma)\neq \mathbb{I}_d$ and for any set $  S\neq \mathbb{I}_d$ with $I(\Sigma)\subseteq S\subseteq \mathbb{I}_d$ we have
    \begin{eqnarray*}
        \gamma(\Sigma)=\gamma(\Sigma_{S}).
    \end{eqnarray*}
    For $  S\subseteq  \mathbb{I}_d$ with $S^c\cap I(\Sigma)\not=\emptyset $ we have  $I(\Sigma)=I(\Sigma_S)$ and
    \begin{eqnarray*}
        \gamma(\Sigma)>\gamma(\Sigma_{S}).
    \end{eqnarray*}
\end{enumerate}
\end{lemma}

\begin{proof}[Proof of \Cref{prop:bi}]
%From \Cref{prop:premain} we have
Note that
\begin{eqnarray*}
\gamma_i&=&\min_{S\subseteq \mathbb{I}_d, |S| \ge i} \min_{\bz_S \ge \bone_S} \bz_S^{\top} \Sigma_S^{-1} \bz_S
\end{eqnarray*}
by definition (see \Cref{prop:premain}).
Hence, there exists a set $S\subseteq \mathbb{I}_d$ with $ |S| \ge i$
and \linebreak $\gamma_i= \min_{\bz_S \ge \bone_S} \bz_S^{\top} \Sigma_S^{-1} \bz_S$. However, due to \Cref{lem:aux1} there exists as well a set $M\subseteq S$ with $|M|=|S|-1\geq i-1$ and
\begin{eqnarray*}
    \gamma_i=\gamma(\Sigma_S)>\gamma(\Sigma_M)\geq\gamma_{i-1}
\end{eqnarray*}
and finally, $\alpha_i=\alpha\gamma_i>\alpha\gamma_{i-1}=\alpha_{i-1}$ and $b_{i}(t)/b_{i-1}(t)\to\infty$ as $t\to\infty$ for $i=2,\ldots,d$.
\end{proof}

\begin{proof}[Proof of \Cref{prop:supportpgc}] $\mbox{}$\\
    (a) \, Suppose that for all $S\subseteq \mathbb{I}_d$ with $|S|=i$ we have $\Sigma^{-1}_S\bone_S >\bzero_S$ and $\bone_S^{\top}\Sigma_S^{-1}\bone_S=\gamma_i$. Then $I(\Sigma_S)=I_S=S$ and
    $\gamma(\Sigma_S)
    =\bone_S^{\top}\Sigma_S^{-1}\bone_S=\gamma_i$. Thus, 
    $\mathcal{S}_i=\{S\subseteq \mathbb{I}_d:|S|=i\}$ and the statement follows directly from \Cref{prop:premain}. \\
    (b) \, Suppose that for  some $S\subseteq \mathbb{I}_d$ with $|S|=i$,  $\Sigma^{-1}_S\bone_S >\bzero_S$ and $\gamma_i\not=\bone_S^{\top}\Sigma^{-1}_S\bone_S$. Then $I_S=S$ but $\gamma_i\not=\bone_{I_S}^{\top}\Sigma^{-1}_{I_S}\bone_{I_S}$ giving $S\notin\mathcal{S}_i$ and finally, the statement follows again from \Cref{prop:premain}.
\end{proof}

\section{Proofs of Section~\ref{sec:systemicrisk}}\label{sec:proofofsec5}

%We prove first \Cref{thm:bivCoVaRmain}, which requires the following Lemma.

 \begin{proof}[Proof of \Cref{thm:bivCoVaRmain}] ~\\
(a) \, \textbf{Case 1.} Suppose $F_1=F_2$. Then 
$\VaR_{\gamma}(Y_1)=\VaR_{\gamma}(Y_2)$ and $\bY=\bY'\in\MRV(\alpha_i,b_i,\mu_i,\E_2^{(i)})$ for $i=1,2$.
    %Then \marginpar{\VF{at the moment with equal margins}.}
 %$Y_1, Y_2$, a canonical choice for $b_1(t)=\overline{F}_1^{\leftarrow}(1/t)=\VaR_{1/t}(Y_2)$    
% To begin with assume $\VaR_{\gamma}(Y_1)=b_1(\gamma^{-1}), 0<\gamma<1$, and by complete tail independence of $Y_1$ and $Y_2$ we have $ \Var_{\gamma}(Y_1)\sim\VaR_{\gamma}(Y_1) = b_1(\gamma^{-1})$ as $\gamma\downarrow 0$. Note that this is the canonical choice of $b_1$ with identical marginals.
    For $0<\gamma<1$ and $0<\upsilon<1$,  we receive   
    \begin{align*}
        \CoVaR_{\upsilon g(\gamma)| \gamma}(Y_1|Y_2) & = \inf \{y\in \R: \P(Y_1>y|Y_2>\VaR_{\gamma}(Y_2))\leq \upsilon g(\gamma)\}\\
                             & = \VaR_{\gamma}(Y_1)\inf \{y\in \R: \P(Y_1>y\VaR_{\gamma}(Y_1)|Y_2>\VaR_{\gamma}(Y_1))\leq \upsilon g(\gamma)\}\\
                             & = \VaR_{\gamma}(Y_1) \inf\left\{y\in\R: h_\gamma(y)\leq \upsilon  g(\gamma)\gamma b_2^{\leftarrow}(\VaR_{\gamma}(Y_1))\right\}\\
                             &=\VaR_{\gamma}(Y_1)h_{\gamma}^{\leftarrow}(\upsilon  g(\gamma)\gamma b_2^{\leftarrow}(\VaR_{\gamma}(Y_1)).
    \end{align*}  
    Due to our assumption, there exist constants $0<a_1<a_2<r$  such that
    \begin{eqnarray*}
        a_1<\liminf_{\gamma\downarrow 0}g(\gamma)\gamma b_2^{\leftarrow}(\VaR_{\gamma}(Y_1))
        \leq \liminf_{\gamma\downarrow 0}g(\gamma)\gamma b_2^{\leftarrow}(\VaR_{\gamma}(Y_1))<a_2,
    \end{eqnarray*}
    and hence, there exists a $\gamma_0\in(0,1)$ such that
    \begin{eqnarray*}
        g(\gamma)\gamma b_2^{\leftarrow}(\VaR_{\gamma}(Y_1))\in\left[\frac{a_1}{2},\frac{a_2+\min(r,a_2)}{2}\right], \quad \forall\,\, 0<\gamma<\gamma_0.
    \end{eqnarray*}
    Due to \Cref{Lemma C.1}, $h_\gamma^{\leftarrow}(y)/h^{-1}(y)\to 1$  uniformly  as $\gamma\downarrow 0$ on $[a_1/2,(a_2+\min(r,a_2))/2]$ and thus,
    \begin{eqnarray*}
        \lim_{\gamma\downarrow 0}\frac{h_{\gamma}^{\leftarrow}(\upsilon  g(\gamma)\gamma b_2^{\leftarrow}(\VaR_{\gamma}(Y_1))}{h^{-1}(\upsilon  g(\gamma)\gamma b_2^{\leftarrow}(\VaR_{\gamma}(Y_1))}=1.
    \end{eqnarray*}
   Therefore, we have
    \begin{eqnarray*}
       \lim_{\gamma\downarrow 0} \frac{\CoVaR_{\upsilon g(\gamma)| \gamma}(Y_1|Y_2)}{\VaR_{\gamma}(Y_1)h^{-1}(\upsilon  g(\gamma)\gamma b_2^{\leftarrow}(\VaR_{\gamma}(Y_1))}&=&
       \lim_{\gamma\downarrow 0} \frac{h_{\gamma}^{\leftarrow}(\upsilon  g(\gamma)\gamma b_2^{\leftarrow}(\VaR_{\gamma}(Y_1))}{h^{-1}(\upsilon  g(\gamma)\gamma b_2^{\leftarrow}(\VaR_{\gamma}(Y_1))}=1.
    \end{eqnarray*}
    \textbf{Case 2.} Suppose $F_2$ is arbitrary. Define 
    $\bY':=(Y_1',Y_2'):=(Y_1,F_1^{\leftarrow}\circ F_2(Y_2))$; by assumption we have $\bY' \in \MRV(\alpha_i,b_i,\mu_i,\E_2^{(i)})$, $i=1,2$, and $Y_1'\eqd Y_2'$. Since $F_1^{\leftarrow}\circ F_2$ is an increasing function, we have
    \begin{eqnarray*}
        \CoVaR_{\upsilon g(\gamma)| \gamma}(Y_1|Y_2)
        =\CoVaR_{\upsilon g(\gamma)| \gamma}(Y_1'|Y_2')
    \end{eqnarray*}
   and we can directly apply Case 1 to $\CoVaR_{\upsilon g(\gamma)| \gamma}(Y_1'|Y_2')$ and obtain the statement.\\
    (b,c) \, By the assumption that $h_\gamma^{\leftarrow}(v)/h^{-1}(v)$ converges uniformly on appropriate intervals, the proof follows analogously as the proof of (a).
\end{proof}

\begin{proof}[Proof of \Cref{Lemma C.1}]~\\
From $(Y_1,F_1^{\leftarrow}\circ F_2(Y_2))\in\MRV(\alpha_1,b_1,\mu_1,\E_2^{(1)})\cap\MRV(\alpha_2,b_2,\mu_2,\E_2^{(2)})$  and the continuity of $h$, it follows that $\VaR_{\gamma}(Y_1)\to \infty$ as $\gamma\downarrow 0$ and
for any $y>0$,
\begin{eqnarray*} %\label{eqada}
h_\gamma(y)\to h(y), \quad \gamma\downarrow 0.
\end{eqnarray*}
Since $h$ is continuous and decreasing we know  from \cite[Lemma 1.1.1]{dehaan:ferreira:2006} (taking $-h$, which is increasing)
that for any $v\in(0,r)$,
\begin{eqnarray} \label{eqada}
h_\gamma^{\leftarrow}(v)\to h^{-1}(v), \quad \gamma\downarrow 0.
\end{eqnarray}
Note that $h^{-1}$ is a strictly decreasing, continuous function due to (i) such that 
\begin{eqnarray*}  
    0<h^{-1}(a_2)=\inf_{v\in[a_1,a_2]}h^{-1}(v)\leq \sup_{v\in[a_1,a_2]}h^{-1}(v)=h^{-1}(a_1)<\infty.
\end{eqnarray*}
Therefore, it is sufficient to prove that
\begin{eqnarray} \label{b1}
    \sup_{v\in[a_1,a_2]}\left|h_\gamma^{\leftarrow}(v)-h^{-1}(v)\right|\to 0, \quad \gamma\downarrow 0,
\end{eqnarray}
or equivalently with $v=w^{-1}$,
\begin{eqnarray*}
    \sup_{w\in[a_2^{-1},a_1^{-1}]}\left|h_\gamma^{\leftarrow}(w^{-1})-h^{-1}(w^{-1})\right|\to 0, \quad\gamma\downarrow 0.
\end{eqnarray*}
Define
\begin{eqnarray*}
    F_\gamma(w)=\frac{h_\gamma^{\leftarrow}(w^{-1})}{h_\gamma^{\leftarrow}(a_1)}
    \quad \text{ and } \quad F(w)=\frac{h^{-1}(w^{-1})}{h^{-1}(a_1)}  \quad \text{ for } w\in[a_2^{-1},a_1^{-1}],
\end{eqnarray*}
with $F_\gamma(w)=F(w)=1$ for $w> a_1^{-1}$ and  $F_\gamma(w)=F(w)=0$ for $w< a_2^{-1}$. Then $F_\gamma$ and $F$ are distribution functions (in particular $F_\gamma$ is right continuous with left limits since $h_\gamma^{\leftarrow}$ is left continuous with right limits) and due to \eqref{eqada} we have as well $\lim_{\gamma\downarrow 0}F_\gamma(w)=F(w)$ for every $w\in\R$. Since $F$ is continuous, Polya's Theorem \cite{polya:1920} actually gives the uniform convergence
\begin{eqnarray*}
    \lim_{\gamma\downarrow 0}\sup_{v\in[a_1,a_2]}\left|\frac{h_\gamma^{\leftarrow}(v)}{h_\gamma^{\leftarrow}(a_1)}-\frac{h^{-1}(v)}{h^{-1}(a_1)}\right|=\lim_{\gamma\downarrow 0}\sup_{w\in[a_2^{-1},a_1^{-1}]}|F_\gamma(w)-F(w)|=0.
\end{eqnarray*}
Finally,
\begin{eqnarray*}
    \lefteqn{\sup_{v\in[a_1,a_2]}\left|h_\gamma^{\leftarrow}(v)-h^{-1}(v)\right|}\\
    &&\quad\leq 
    \sup_{v\in[a_1,a_2]}\left|h_\gamma^{\leftarrow}(v)-h^{-1}(v)\frac{h_\gamma^{\leftarrow}(a_1)}{h^{-1}(a_1)}\right|
    +\sup_{v\in[a_1,a_2]}\left|h^{-1}(v)\frac{h_\gamma^{\leftarrow}(a_1)}{h^{-1}(a_1)}-h^{-1}(v)\right|\\
    &&\quad= 
    h_{\gamma}^{\leftarrow}(a_1)\sup_{v\in[a_1,a_2]}\left|\frac{h_\gamma^{\leftarrow}(v)}{h_\gamma^{\leftarrow}(a_1)}-\frac{h^{-1}(v)}{h^{-1}(a_1)}\right|
    +\left|\frac{h_\gamma^{\leftarrow}(a_1)}{h^{-1}(a_1)}-1\right|\sup_{v\in [a_1,a_2]}|h^{-1}(v)|\\
    &&\quad\stackrel{\gamma\downarrow 0}{\to}0,
\end{eqnarray*}
which gives the statement.
\end{proof}

%\newpage
\begin{comment}

\begin{lemma}\label{lem:homprop}
    Let bivariate random vector $\bX\in \MRV(\alpha_i,b_i, \mu_i,\E_2^{(i)}), i=1,2$ with $\lim_{t\to\infty} b_1(t)/b_{2}(t)=\infty$ and $X_1$ and $X_2$ are (completely) tail equivalent. Then
    \[\lim_{x\downarrow 0}\mu_2((1,\infty)\times(x,\infty)) = \lim_{x\downarrow 0}\mu_2((x,\infty)\times(1,\infty)) =\infty.\]
\end{lemma}
\marginpar{\VF{The assumption of asymptotic independence is essential. Otherwise it does not hold.}}
\begin{proof}
 By homogeneity property of $\mu_1, \mu_2$ we have that $\mu_i(tA_i)=t^{-\alpha_i}\mu_1(A_i), i=1,2$ for any set $A_i\subset \mathcal{B}(\E_2^{(i)})$; see \cite[Theorem 3.1]{lindskog:resnick:roy:2014}.
 \DAS{to be done} 
% Clearly $g(x)= \mu_2((x,\infty)\times(1,\infty))>0, x>0$ is a decreasing function of $x$ (since $\mu_2$ is a measure). Suppose $\lim_{x\downarrow0} g(x) =c <\infty$. 
 %This means
 %\begin{align}
 %  \lim_{t\to\infty} t\P\left(\frac{\bX}{b_2(t)}\in (0,\infty)\times(1,\infty)\right)  & = \mu_2((0,\infty)\times(1,\infty))=c.
 % \end{align}
 \end{proof}

\end{comment}

\begin{proof}[Proof of \Cref{corollary 4.4}] $\mbox{}$\\ 
(a)(i)  First of all,
\begin{align*}
    h_{\gamma}(y) &:= b_2^{\leftarrow}(\VaR_{\gamma}(Y_1))\P\left(Y_1>y\VaR_{\gamma}(Y_1),Y_2>\VaR_{\gamma}(Y_1))\right)
\end{align*}
and due to \Cref{prop:pmocmrv} we have 
 $h(y)=y^{-\alpha}$ for $y\geq 1$ and $h^{-1}(v)=v^{-1/\alpha}$ for $v\in\left(0,1\right]$. Define $\eta=\frac{1}{2}$. Since 
 $$b_2^{\leftarrow}(t)=\theta^{-(\eta+1)}t^{(\eta +1)\alpha}\sim (\ov F_1(t)\ov F_2(t)^\eta)^{-1}, \quad t\to\infty,$$ we take w.l.o.g. $b_2^{\leftarrow}(t)=(\ov F_1(t)\ov F_2(t)^\eta)^{-1}$.
 Then  we have for $y\geq 1$,
\begin{eqnarray*}
    b_2^{\leftarrow}(t)\P(Y_1>yt,Y_2>t)
    &=&b_2^{\leftarrow}(t)\widehat C_{\Lambda}^{\text{MO}}(\ov F_1(yt),\ov F_2(t))
    =b_2^{\leftarrow}(t)\ov F_1(yt)\ov F_2(t)^{\eta}\\
    &=&\frac{\ov F_1(yt)}{\ov F_1(t)}=:\widetilde h_t(y).
\end{eqnarray*}
Moreover,
\begin{eqnarray*}
   \frac{ \widetilde h_t^{\leftarrow}(v)}{v^{-\frac{1}{\alpha}}}=\frac{\ov F_1^{\leftarrow}(v\ov F_1(t))}{v^{-\frac{1}{\alpha}}t}=\frac{\ov F_1^{\leftarrow}(v\ov F_1(t))}{v^{-\frac{1}{\alpha}}\ov F_1^{\leftarrow}(\ov F_1(t))}\frac{\ov F_1^{\leftarrow}(\ov F_1(t))}{t}.
\end{eqnarray*}
The first factor converges uniformly to $1$ on $\left(0,1\right]$ as $t\to\infty$ by our assumption, and the second factor converges to $1$ as $t\to\infty$ as well. Hence, $\widetilde h_t^{\leftarrow}(v)/h^{-1}(v)$ converges uniformly to $1$ on $\left(0,1\right]$ as $t\to\infty$. Finally, $h_\gamma^{\leftarrow}(v)/h^{-1}(v)$ converges uniformly to $1$ on $\left(0,1\right]$ as $\gamma\downarrow 0$ such that the assumption in \Cref{thm:bivCoVaRmain}(b) is satisfied. 
%Thus,
%\begin{eqnarray*}
%    b_2^{\leftarrow}(t)\P(Y_1>yt,Y_2>t)y^{\alpha}
%    &=&\frac{\ov F_1(yt)}{\theta (ty)^{-\alpha}}\cdot \left(\frac{\ov F_2(t)}{\theta t^{-\alpha}}\right)^{\eta},
%\end{eqnarray*}
%which converges uniformly to $1$, since the first expression converges uniformly to $1$ by assumption and the second converges as well to $1$, since the tails of $F_j$ are tail equivalent to Pareto tails. But then as well $h_\gamma(y)/h(y)$ converges uniformly to $1$ on $\left[1,\infty\right)$.
Since $\upsilon\gamma^{\beta+1} b_2^{\leftarrow}(\VaR_{\gamma}(Y_1))\sim \upsilon\gamma^{\beta-\eta}$, we have  $h^{-1}(\upsilon\gamma^{\beta+1} b_2^{\leftarrow}(\VaR_{\gamma}(Y_1))\sim \upsilon^{-\frac{1}{\alpha}}\gamma^{-\frac{\beta-\eta}{\alpha}}$ as $\gamma\downarrow 0$. The final statement is then an application of \Cref{thm:bivCoVaRmain}(b).\\
(a)(ii) \, For $0<y\leq 1$ we have $h(y)=y^{-\alpha\eta}$
and for $v\in\left[1,\infty\right)$ we have $h^{-1}(v)=v^{-\frac{1}{\alpha\eta}}$. Furthermore,
\begin{eqnarray*}
 b_2^{\leftarrow}(t)\P(Y_1>yt,Y_2>t)
    &=&b_2^{\leftarrow}(t)\ov F_1(yt)^\eta\ov F_2(t).
\end{eqnarray*}
The rest of the proof follows analogously as the proof of  (i) using \Cref{thm:bivCoVaRmain}(c).\\
(b) \, The proof of (b) is analogous to that of part (a) by taking $\eta=1/3$.
\end{proof}

\begin{proof}[Proof of \Cref{CoVaR Gaussian}] 
In the bivariate case we have $I_2=\{1,2\}$, $\gamma_2=\frac{2}{1+\rho}$, $\alpha_2=\frac{2\alpha}{1+\rho}$ and $h_s^S= \frac{1}{1+\rho}$ for $S=\{1,2\}$; cf. \cite[Example 3.8]{das:fasen:2024}. Define 
\begin{eqnarray*}
    b_2^{\leftarrow}(t)&=&(\theta t^{-\alpha})^{-\frac{2}{1+\rho}}(2\alpha \log(t))^{\frac{\rho}{1+\rho}}(2\pi)^{\frac{\rho}{1+\rho}}\frac{(1-\rho)^{\frac{1}{2}}}{(1+\rho)^{\frac{3}{2}}}
    ,\\
    \mu_2((\bz,\infty))&=&(z_1z_2)^{-\frac{\alpha}{1+\rho}}, \quad \bz=(z_1,z_2)^\top\in\R_+^2,\\
    h^{-1}(v)&=&v^{-\frac{1+\rho}{\alpha}}, \quad v>0.
\end{eqnarray*}

A conclusion of \Cref{prop:premain} is that $(Y_1,Y_2)^\top\in \MRV(\alpha_2,b_2,\mu_2,\mathbb{E}^{(2)}_2)$ (note that the constant $ \Upsilon_{\{1,2\}}$ is moved to $b_2$).
The proof follows now by applying the forms of $ b_2,  \mu_2$   to \Cref{thm:bivCoVaRmain}(a).
\end{proof}

\section{Proofs of Section~\ref{sec:bipnet}} \label{Appendix D}

\begin{proof}[Proof of \Cref{thm:main:dfk}]
Define $b_i^*(t)=F_{Z_{(i)}}^{\leftarrow}\left(1-1/t\right)$. Then
\begin{eqnarray}
    \lim_{t\to\infty}\frac{b_i^{\leftarrow}(t)}{b_i^{*\leftarrow}(t)}
    &=& \lim_{t\to\infty} b_i^{\leftarrow}(t)\P(Z_{(i)}>t) \nonumber \\
    &=&\mu_i(\{\bz\in\R_+^d:\, z_{(i)}>1\})=:c_i\in(0,\infty). \label{eq3}
\end{eqnarray}
Thus,
$\bZ \in \MRV(\alpha_{i},b_{i}, \mu_{i},\E^{(i)}_d)$ is equivalent to
$\bZ \in \MRV(\alpha_{i},b_{i}^*, \mu_{i}^*,\E^{(i)}_d)$ with \linebreak 
$\mu_i^*=\frac{1}{c_i}\mu_i$. Furthermore, due to \eqref{eq3}
we have as well $b_i^*(t)/b_{i+1}^*(t)\to\infty$ as $t\to\infty$, $i=1,\ldots, d-1$. Thus, (a) is a consequence of \cite[Theorem 3.4]{das:fasen:kluppelberg:2022} and (b) of
\cite[Proposition~3.2]{das:fasen:kluppelberg:2022} in combination with \eqref{eq3} and $\bZ \in \MRV(\alpha_{i},b_{i}^*, \mu_{i}^*,\E^{(i)}_d)$ for all $i\in \mathbb{I}_d$.
\end{proof}

\begin{proof}[Proof of \Cref{exchangable}] $\mbox{}$\\
(a) \, First of all, $\tau_{(1,1)}(\bA)<\infty$ and $\tau_{(1,2)}(\bA)=\infty$ a.s. such that $i_1^*=1$ and $\P(\Omega_1^{(1)})=1$.
Since $[\bzero,\bx]^c\in\mathcal{B}(\mathbb{E}_2^{(1)})$ we obtain by the definition of $\overline{\mu}_1$ that
\begin{eqnarray} \label{eq1}
    \overline{\mu}_1([\bzero,\bx]^c)=\E(\mu_1(\bA^{-1}([\bzero,\bx]^c))).
\end{eqnarray}
Furthermore, from \cite[Remark 7]{das:fasen:kluppelberg:2022}, we already know that due to
$\bZ \in \MRV(\alpha_{i},b_{i}, \mu_{i},\E^{(i)}_d)$, $i=1,2$, with $b_1(t)/b_{2}(t)\to \infty$
the support of $\mu_1$ is restricted to
\begin{eqnarray} \label{support}
    \left\{\bz\in\R_+^d:\, z_{(2)}=0\right\}\backslash\{\bzero\}
    =\bigcup_{\ell=1}^d\left\{\bz\in\R_+^d:\,z_\ell>0, z_j=0 ,\,\forall\,j\in\mathbb{I}_d\backslash\{\ell\}\right\}=:\bigcup_{\ell=1}^d\mathbb{T}_\ell,\quad\quad\quad
\end{eqnarray}
which is a disjoint union. Together with \eqref{eq1} this implies that
\begin{eqnarray*}
    \overline{\mu}_1([0,\bx]^c)
    &=&\E(\mu_1(\bA^{-1}([0,\bx]^c)))= \sum_{\ell=1}^d\E(\mu_1(\bA^{-1}([0,\bx]^c)\cap \mathbb{T}_\ell)\\
    &=&\sum_{\ell=1 }^d \E\left[\mu_1\left(\left\{\bz\in\R_+^d:a_{1\ell}z_\ell>x_1 \text{ or }a_{2\ell}z_\ell>x_2\right\}\right)\right]\\
    &=&\sum_{\ell=1 }^d \E\left[\mu_1\left(\left\{\bz\in\R_+^d:\,\max\left(\frac{a_{1\ell}}{x_1},\frac{a_{2\ell}}{x_2}\right)z_\ell>1\right\}\right)\right].
\end{eqnarray*}
(b) \, The set $(\bx,\infty)\in\mathcal{B}(\mathbb{E}_2^{(2)})$. When $k=2$ we have
\begin{eqnarray*}
    \Omega_1^{(2)}=\left\{\min_{\ell\in\mathbb{I}_d}\{a_{1\ell},a_{2\ell}\}>0\right\}
    \quad \text{ and }  \quad
    \Omega_2^{(2)}=\left\{\min_{\ell\in\mathbb{I}_d}\{a_{1\ell},a_{2\ell}\}=0\right\}.
\end{eqnarray*}
Now by our assumption $\P(\Omega_1^{(2)})>0$ and hence,  $i_2^*=1$. A consequence of the definition of $\overline{\mu}_2$ and $(\bx,\binfty)\in\mathcal{B}(\mathbb{E}_2^{(2)})$ is then
\begin{eqnarray*}
    \overline{\mu}_2((\bx,\binfty))&=&
    \E(\mu_1(\bA^{-1}((\bx,\binfty))\cap \Omega_1^{(2)}))\\
    &=&\sum_{\ell=1}^d
    \E\left[\mu_1\left(\left\{\bz\in\R_+^d:a_{1\ell}z_\ell>x_1,a_{2\ell}z_\ell>x_2\right\}\cap \Omega_1^{(2)}\right)\right]\\
    &=&\sum_{\ell=1 }^d \E\left[\mu_1\left(\left\{\bz\in\R_+^d:\,\min\left(\frac{a_{1\ell}}{x_1},\frac{a_{2\ell}}{x_2}\right)z_\ell>1\right\}\cap \Omega_1^{(2)}\right)\right]\\
    &=&\sum_{\ell=1 }^d \E\left[\mu_1\left(\left\{\bz\in\R_+^d:\,\min\left(\frac{a_{1\ell}}{x_1},\frac{a_{2\ell}}{x_2}\right)z_\ell>1\right\}\right)\right].
\end{eqnarray*}
\item[(c)] By assumption we have $\P(\Omega_1^{(2)})=0$ such that $i_2^*=2$
and $\P(\Omega_2^{(2)})=1$. Hence, the definition of $\overline{\mu}_2$ implies that
\begin{eqnarray*}
     \overline{\mu}_2((\bx,\binfty))
    =\E(\mu_2(\bA^{-1}((\bx,\binfty))\cap \Omega_2^{(2)}))
    =\E(\mu_2(\bA^{-1}((\bx,\binfty)))).
\end{eqnarray*}
Again from \cite[Remark 7]{das:fasen:kluppelberg:2022}, we already know that
the support of $\mu_2$ is restricted to
\begin{eqnarray} \label{support2}
    \lefteqn{\hspace*{-1cm}\left\{\bz\in\R_+^d:\, z_{(3)}=0\right\}\Big\backslash
    \left\{\bz\in\R_+^d:\, z_{(2)}=0\right\}} \nonumber \\
    &&=\bigcup_{1\leq \ell< j\leq d}\left\{\bz\in\R_+^d:\,z_\ell>0, z_j>0, z_m=0,\, \forall\, m\in\mathbb{I}_d\backslash\{\ell, j\}\right\} \nonumber \\
    &&=:\bigcup_{1\leq \ell<j\leq d}\mathbb{T}_{\ell,j}, \nonumber
\end{eqnarray}
such that
\begin{eqnarray*}
     \overline{\mu}_2((\bx,\binfty))
       &=&\sum_{1\leq \ell<j\leq d}\E(\mu_2(\bA^{-1}((\bx,\binfty))\cap \mathbb{T}_{\ell, j}))\\
       &=&\sum_{\genfrac{}{}{0pt}{}{\ell, j=1}{\ell\neq j}}^d
        \E\left[\mu_2\left(\left\{\bz\in\R_+^d:
        \frac{a_{1\ell}}{x_1}z_{\ell}>1,
        \frac{a_{2j}}{x_2}z_j>1 \right\}\right)\right]
\end{eqnarray*}
completing the proof.
\end{proof}

\begin{proof}[Proof of \Cref{prop:condprobmain}]
Since $Z_1,\ldots,Z_d$ are completely tail equivalent, we have
\begin{eqnarray*}
        \mu_{1}(\{\bz\in\R_+^d: z_\ell>1\})
        =\mu_{1}(\{\bz\in\R_+^d: z_1>1\}), \quad \ell \in\mathbb{I}_d,
\end{eqnarray*}
and by  the choice of $b_1$
\begin{eqnarray*}
    \mu_1([\bzero,\bone]^c)&=&\mu_1\left([\bzero,\bone]^c\cap \bigcup_{\ell=1}^d\mathbb{T}_\ell\right)=\sum_{\ell=1}^d\mu_{1}(\{\bz\in\R_+^d: z_\ell>1\})
    =d\mu_{1}(\{\bz\in\R_+^d: z_1>1\}).
\end{eqnarray*}
This implies
\begin{eqnarray} \label{eq2b}
    \mu_{1}(\{\bz\in\R_+^d: z_\ell>1\})=\frac{\mu_1([\bzero,\bone]^c)}{d},\quad \ell \in\mathbb{I}_d.
\end{eqnarray}
Furthermore, $\mu_1$ is homogeneous of order $-\alpha$  such that
together with \Cref{exchangable}(a) we receive
\begin{eqnarray*}
            \overline{\mu}_1([\bzero,\bx]^c)&=&\sum_{\ell=1 }^d \E\left[\mu_1\left(\left\{\bz\in\R_+^d:\,\max\left(\frac{a_{1\ell}}{x_1},\frac{a_{2\ell}}{x_2}\right)z_\ell>1\right\}\right)\right]\\
            &=&\sum_{\ell=1 }^d \E\left[\max\left(\frac{a_{1\ell}}{x_1},\frac{a_{2\ell}}{x_2}\right)^\alpha\right]
            \mu_1\left(\left\{\bz\in\R_+^d:\, z_\ell>1\right\}\right)\\
           % &=&\mu_1\left(\left\{\bz\in\R_+^d:\, z_1>1\right\}\right)\sum_{j=1 }^d \E\left[\max\left(\frac{a_{1j}}{x_1},\frac{a_{2j}}{x_2}\right)^\alpha\right]\\
            &=&\frac{\mu_1([\bzero,\bone]^c)}{d}\sum_{\ell=1 }^d \E\left[\max\left\{\frac{a_{1\ell}}{x_1},\frac{a_{2\ell}}{x_2}\right\}^\alpha\right].
        \end{eqnarray*}
Finally, the almost surely non-zero rows of $\bA$ imply that the right-hand side above is strictly positive, and hence, $\overline\mu_1$ is a non-null measure. From $i_1^*=1$ and \Cref{thm:main:dfk} we then conclude
that  $(X_1,X_2)^\top\in\MRV(\alpha_1,b_1,\overline{\mu}_1,\E_2^{(1)})$.\\
(b) From \Cref{exchangable}(b) we already know that $i_2^*=2$
and
\begin{eqnarray*}
            \overline{\mu}_2(\left(\bx,\binfty\right))
                =\sum_{\ell=1 }^d \E\left[\mu_1\left(\left\{\bz\in\R_+^d:\,\min\left(\frac{a_{1\ell}}{x_1},\frac{a_{2\ell}}{x_2}\right)z_\ell>1\right\}\right)\right].
\end{eqnarray*}
Again the homogeneity of $\mu_1$ of order $-\alpha$
and \eqref{eq2b} gives
\begin{eqnarray*}
    \overline\mu_2(\left(\bx,\binfty\right))
                =\frac{\mu_1([\bzero,\bone]^c)}{d}\sum_{\ell=1}^d \E\left[\min\left\{\frac{a_{1\ell}}{x_1},\frac{a_{2\ell}}{x_2}\right\}^{\alpha}%\mathds{1}_{\{a_{1j}>0,a_{2j}>0\}}
                \right],
\end{eqnarray*}
which is strictly positive because of $\max_{\ell\in \mathbb{I}_d}\P(\min\{a_{1\ell},a_{2\ell}\}>0)>0$. Then a consequence of \Cref{thm:main:dfk}
is $(X_1,X_2)^\top\in\MRV(\alpha_1,b_1,\overline{\mu}_2,\E_2^{(2)})$. Finally,
\begin{eqnarray*}
    \P(X_1>tx_1|X_2>tx_2) =\frac{\P(X_1>tx_1,X_2>tx_2)}{\P(X_2>tx_2)}
    &\sim&\frac{\overline\mu_2((\bx,\infty))}{\overline\mu_1( \left[0,\infty\right)\times (x_2,\infty))}\\
    &=&x_2^{\alpha}
            \frac{\sum\limits_{\ell=1}^d\E\left[\min\left\{\frac{a_{1\ell}}{x_1},\frac{a_{2\ell}}{x_2}\right\}^{\alpha}\right]}{\sum\limits_{\ell=1}^d\E\left[
            a_{2\ell}^{\alpha}
            \right]},
\end{eqnarray*}
as $t\to\infty$ by the former results. The asymptotic behavior of the $\CoVaR$ is then a consequence of \Cref{thm:bivCoVaRmain}(a) and \eqref{def: h}.
\begin{comment}
 Indeed,
%\begin{eqnarray*}
%   \textcolor{red}{ h(y)=\frac{\mu_1([\bzero,\bone]^c)}{d}\sum_{\ell=1}^d \E\left[\min\left\{\frac{a_{1\ell}}{y},a_{2\ell}\right\}^{\alpha}
%                \right].}
%\end{eqnarray*}
by assumption, the non-zero components of $\bA$ have bounded support, bounded away from zero. Consequently there exists a $y_0>0$ such that
\begin{eqnarray*}
    h(y)=y^{-\alpha}\frac{\mu_1([\bzero,\bone]^c)}{d}\sum_{\ell=1}^d \E\left[a_{1\ell}^{\alpha}\right]
                , \quad \forall\, y\geq y_0,
\end{eqnarray*}
and $h^{-1}:(h({y_0}),\infty)\to(0,{y_0})$ is
\begin{eqnarray*}
    h^{-1}(v)=v^{-\frac{1}{\alpha}}\left[\frac{\mu_1([\bzero,\bone]^c)}{d}\sum_{\ell=1}^d \E\left[a_{1\ell}^{\alpha}\right]\right]^{\frac{1}{\alpha}}.
\end{eqnarray*}
Finally, by \Cref{thm:bivCoVaRmain}(a) as $\gamma\downarrow 0$,
\begin{eqnarray*}
        \CoVaR_{\upsilon|\gamma}(X_1|X_2)&\sim&{\VaR_{\gamma}(X_1) h^{-1}(\upsilon \gamma b_2^{\leftarrow}(\VaR_{\gamma}(X_1))}
    \sim \upsilon^{-\frac{1}{\alpha}} \VaR_{\gamma}(X_1).%\left({\sum\limits_{\ell=1}^d\E\big[ a_{1\ell}^{\alpha} \big]}\Bigg/{\sum\limits_{j=1}^d\E\big[ a_{2j}^{\alpha} \big]}\right)^{\frac{1}{\alpha}}.
\end{eqnarray*}
%we receive the last statement.
\end{comment}
\end{proof}
         
\begin{proof}[Proof of \Cref{prop:condprobiid}]
%\VF{As in \citet[Example 2.1]{das:fasen:kluppelberg:2022}} (who assumed Pareto tails) we can conclude that
First, note that
\begin{eqnarray*}
   \mu_2(\{\bv\in\R_+^d:v_\ell>z_\ell, v_j>z_j\})=%\VF{\frac{1}{\left(d\atop 2\right)}}
        (z_\ell z_j)^{-\alpha}, \quad \bz\in(\bzero,\binfty).
\end{eqnarray*}
Plugging this in \Cref{exchangable}(c) results in
    \begin{eqnarray*}
        \overline{\mu}_2((\bx,\infty))&=&\sum_{\ell,j=1}^d
        \E\left[\mu_2\left(\left\{\bz\in\R_+^d:
        \frac{a_{1{\ell}}}{x_1}z_{\ell}>1,
        \frac{a_{2j}}{x_2}z_j>1 \right\}\right)\right]%\\
       % &=&%\frac{1}{\binom{d}{2}}
       % \sum_{\genfrac{}{}{0pt}{}{j,\ell=1}{j\not=\ell}}^d \E\left[\left(\frac{a_{1j}}{x_1} \frac{a_{2\ell}}{x_2}\right)^{\alpha}\right]
       =%\frac{1}{\binom{d}{2}}
        \sum_{\ell,j=1}^d \E\left[\left(\frac{a_{1\ell}}{x_1} \frac{a_{2j}}{x_2}\right)^{\alpha}\right],
    \end{eqnarray*}
where we used in the last step the assumption $\max_{\ell\in \mathbb{I}_d}\P(\min\{a_{1\ell},a_{2\ell}\}>0)=0$.  Furthermore, $b_1^{\leftarrow}(t)=1/\ov F_\alpha(t)$ and $b_2^{\leftarrow}(t)=1/\ov F_\alpha(t)^2$ such that as $t\to\infty$,
\begin{eqnarray*}
    \P(X_1>tx_1|X_2>tx_2)&=&\frac{b_1^{\leftarrow}(t)}{b_2^{\leftarrow}(t)}\frac{b_2^{\leftarrow}(t)\P(X_1>tx_1,X_2>tx_2)}{b_1^{\leftarrow}(t)\P(X_2>tx_2)}\\
    &\sim& \ov F_\alpha(t)\frac{\ov\mu_2(\bx,\infty)}{\ov \mu_1((x_1,\infty)\times \R_+)}\\
    &=&  \ov F_\alpha(t)x_2^{\alpha}
            \frac{\sum\limits_{\ell, j=1 }^d \E\left[\left(\frac{a_{1\ell}}{x_1} \frac{a_{2j}}{x_2}\right)^{\alpha} \right]}{\sum\limits_{j=1}^d\E\left[
            a_{2j}^{\alpha}
            \right]}.
\end{eqnarray*}
Again the asymptotic behavior of the $\CoVaR$ follows then from \Cref{thm:bivCoVaRmain}(a) and \eqref{def: h} with similar arguments as in \Cref{prop:condprobmain}.
\end{proof}

\begin{proof}[Proof of \Cref{prop:condprobMO}]
Similarly, as in  \Cref{prop:condprobiid}, the proof of the MRV is a combination of \Cref{exchangable}(c) and \Cref{prop:pmocmrv}, and the asymptotic behavior of the CoVaR can be derived by \Cref{thm:bivCoVaRmain}(a).
\end{proof}

\begin{proof}[Proof of \Cref{prop:condprobGC}] 
 For any set $S=\{\ell,j\}\subseteq \mathbb{I}_d$ with $|S|=2$ we have
\begin{eqnarray*}
    \Sigma_S^{-1}=\frac{1}{1-\rho_{\ell j}^2}\left(\begin{array}{cc}
        1 & -\rho_{j\ell}\\ -\rho_{\ell j} & 1
    \end{array}\right),
\end{eqnarray*}
$\Sigma_S^{-1}\bone_S>\bzero_S$ and $\bone_S^\top\Sigma_S^{-1}\bone_S=\frac{2}{1+\rho_{\ell j}}$. Thus, the set $\mathcal{S}_2$ as defined in \Cref{prop:premain}
is equal to
\begin{eqnarray*}
    \mathcal{S}_2=\left\{\{\ell, j\}\subseteq \mathbb{I}_d:\frac{2}{1+\rho_{\ell j}}=\frac{2}{1+\rhom}\right\}
        =\{\{\ell, j\}\subseteq \mathbb{I}_d: \rho_{\ell j}=\rhom\},
\end{eqnarray*}
$|I(\Sigma_S)|=|I_S|=|S|=2$ and $\gamma_S=\frac{2}{1+\rhom}$ for any $S\in \mathcal{S}_2$. Finally, with $\alpha_2$, $b_2$ as given above, $z_{\ell}, z_j>0$, and
\begin{eqnarray*}
    \mu_2(\{\bv\in\R_+^d: v_{\ell}>z_{\ell},v_j>z_j\})=\left\{
    \begin{array}{ll}
             \frac{1}{2\pi} \frac{(1+\rhom)^{3/2}}{(1-\rhom)^{1/2}}(z_{\ell}z_j)^{-\frac{\alpha}{1+\rho_{\ell j}}},\quad\quad & \text{ if }\rho_{\ell j}=\rhom,\\[2mm]
            0, & \text{ otherwise},
    \end{array}
    \right.
\end{eqnarray*}
we have due to \Cref{prop:premain} that $\bZ\in\MRV(\alpha_2,b_2,\mu_2,\mathbb{E}_d^{(2)})$. %Note that in the bivariate case $\Upsilon_S$ does not depend on $S$ and is therefore included in the expression of $b_2$ in contrast to the original version of \Cref{prop:premain} (cf. proof of  \Cref{CoVaR Gaussian}).
The representation of $\mu_2$ and \Cref{exchangable}(c)
imply that for $\bx\in(\bzero,\binfty)$, we have
\begin{eqnarray} \label{eq4}
 \overline{\mu}_2((\bx,\infty))=  \frac{1}{2\pi} \frac{(1+\rhom)^{3/2}}{(1-\rhom)^{1/2}}\sum_{(\ell, j):\,\rho_{\ell j}=\rhom}\E\left[a_{1\ell}^{\frac{\alpha}{1+\rhom}}a_{2j}^{\frac{\alpha}{1+\rhom}}\right](x_1x_2)^{-\frac{\alpha}{1+\rhom}}.
    \end{eqnarray}
(a) \, If $\rho^*=\rhom$, then of course $\overline{\mu}_2((\bx,\infty))>0$ and from \Cref{thm:main:dfk}, it follows that $(X_1,X_2)^\top\in\MRV(\alpha_2,b_2,\overline{\mu}_2,\E_2^{(2)})$. Finally, the asymptotic behavior of the  conditional probability can be calculated as in the previous statements and the $\CoVaR$ asymptotics follow from \Cref{thm:bivCoVaRmain}(a).\\
(b) \,  If $\rho^*<\rhom$, then \eqref{eq4}
results in $\overline{\mu}_2((\bx,\infty))=0$. Then define
$${M}:=\{\ell\in\mathbb{I}_d: \rho_{\ell j}\leq \rho^* \,\, \forall\,\, j\in\mathbb{I}_d\}.$$
For any $\ell\in\mathbb{I}_d\backslash M$ there exists an $j\in\mathbb{I}_d$, $j\not=\ell$ with $ \rho_{\ell j}>\rho^*$  and hence, \Cref{Remark 4.12} gives that the $\ell$-th column of $\bA$ is a.s. a zero column. Therefore, we define $\bA_{\{1,2\},M}$ by deleting the $\ell$-th columns in $\bA$ for all $\ell\in\mathbb{I}_d\backslash M$ and similarly, we define $\bZ_{M}$. Then $$\bX=\bA\bZ=\bA_{\{1,2\},M}\bZ_{M}.$$ But $\bZ_{M}\in\PGC(\alpha, \theta, \Sigma_{M})$ and $\rho^*=\max_{\ell, j\in M, \ell\neq j} \rho_{\ell j}$. Hence, the model $\bX=\bA_{\{1,2\},M}\bZ_{M}$ satisfies the assumption of (a) and an application of (a) gives the statement.
\end{proof}

\section{Risk contagion with more than two portfolios}\label{sec:app:morethantwo} 
%\marginpar{\VF{Shall we move this section to the appendix?}}
In \Cref{subsec:morethantwo} we obtained asymptotic conditional tail probabilities and CoVaR asymptotics comparing the risk of high negative returns for one single entity in the system vs. the worst (or at least one other entity in the entire system) having poor returns for a general regularly varying underlying distribution $\bZ$. In this section, we detail results for particular choices of $\bZ$.
\begin{proposition}[i.i.d. case] \label{prop:condproboneandmaxiid}
Let $\bZ\in\R_+^d$ be a random vector with i.i.d. components $Z_1,\ldots,Z_d$ with distribution function $F_\alpha$ where $\ov F_\alpha \in\RV_{-\alpha}$, $\alpha>0$, $b_1(t)=F_{\alpha}^{\leftarrow}\left(1-1/t\right)$  and $b_i^\leftarrow(t)=b_1^{\leftarrow}(t)^{i}$.
Further, let $\bA \in \R_+^{q \times d}$ be a  random matrix  with almost surely no trivial rows, independent of $\bZ$ and
 suppose for fixed $k\in \mathbb{I}_q$ we have $\max_{m\in\mathbb{I}_q\setminus\{k\}} \max_{\ell\in\mathbb{I}_d} \P(\min(a_{k\ell}, a_{m\ell})>0)=0$ and  $\E\|\bA\|^{2\alpha+\epsilon}<\infty$ for some $\epsilon>0$. Define $\bX=\bA\bZ$. Then $\bY=(Y_1,Y_2)^\top=(X_k, \max_{m\in \mathbb{I}_q\setminus\{k\}} X_m )^\top \in\MRV(2\alpha,b_2,{\mu}_2^*,\E_2^{(2)})$ where
    \begin{eqnarray*}
        {\mu}_2^*((\bx,\infty))=%\frac{1}{\binom{d}{2}}
        (x_1x_2)^{-\alpha}\sum\limits_{\ell,j=1}^d  \E\Big[a_{k\ell}^{\alpha} \max_{m\in\mathbb{I}_q\setminus\{k\}} \left\{a_{mj}^{\alpha}\right\}\Big], \quad \bx\in(\bzero,\binfty).
    \end{eqnarray*}
Moreover, as $t\to\infty$,
        \begin{align*}
            \P(Y_1>tx_1|Y_2>tx_2)&\sim  (b_1^{\leftarrow}(t))^{-1}\, x_1^{-\alpha} \,\frac{\sum\limits_{\ell,j=1}^d  \E\Big[a_{kj}^{\alpha} \max_{m\in\mathbb{I}_q\setminus\{k\}} \left\{a_{m\ell}^{\alpha}\right\}\Big]}{\sum\limits_{\ell=1}^d\E\big[\max_{m\in \mathbb{I}_q\setminus \{k\}} a_{m\ell}^{\alpha} \big]}, \\ %\intertext{and,}
            \P(Y_2>tx_2|Y_1>tx_1)&\sim  (b_1^{\leftarrow}(t))^{-1}\, x_2^{-\alpha}\, \frac{\sum\limits_{\ell,j=1}^d  \E\Big[a_{kj}^{\alpha} \max_{m\in\mathbb{I}_q\setminus\{k\}} \left\{a_{m \ell}^{\alpha}\right\}\Big]}{ \sum\limits_{j=1}^d\E\big[ a_{k j}^{\alpha} \big]}.
        \end{align*}
        Additionally,  for $0<\upsilon<1$ we have as $\gamma\downarrow 0$,
        \begin{eqnarray*}
            \CoVaR_{\upsilon\gamma|\gamma}(Y_1|Y_2)&\sim& \upsilon^{-\frac{1}{\alpha}} \left(\frac{\sum\limits_{\ell,j=1}^d  \E\Big[a_{kj}^{\alpha} \max_{m\in\mathbb{I}_q\setminus\{k\}} \left\{a_{m\ell}^{\alpha}\right\}\Big]}{\left(\sum\limits_{j=1}^d  \E\Big[a_{kj}^{\alpha}\big]\sum\limits_{\ell=1}^d\E\big[\max_{m\in \mathbb{I}_q\setminus \{k\}} a_{m\ell}^{\alpha} \big]\right)}\right)^{\frac{1}{\alpha}}\VaR_{\gamma}(Y_1),\\
        \CoVaR_{\upsilon\gamma|\gamma}(Y_2|Y_1)&\sim& \upsilon^{-\frac{1}{\alpha}} \left(\frac{\sum\limits_{\ell,j=1}^d  \E\Big[a_{k j}^{\alpha} \max_{m\in\mathbb{I}_q\setminus\{k\}} \left\{a_{m \ell}^{\alpha}\right\}\Big]}{ \sum\limits_{j=1}^d\E\big[ a_{k j}^{\alpha} \big]\sum\limits_{\ell=1}^d\E\big[\max_{m\in \mathbb{I}_q\setminus \{k\}} a_{m\ell}^{\alpha} \big]}\right)^{\frac{1}{\alpha}}\VaR_{\gamma}(Y_2).
        \end{eqnarray*}
         Finally,  $\text{ECI}(Y_1|Y_2)=\text{ECI}(Y_2|Y_1)=1$.
\end{proposition}

\begin{proposition}[Marshall-Olkin dependence] \label{prop:condproboneandmaxMO}
Let $\bZ \in \PMOC(\alpha,\theta,\Lambda)$
and \linebreak $\bA \in \R_+^{q \times d}$ be a  random matrix  with almost surely no trivial rows, independent of $\bZ$ and for fixed $k\in \mathbb{I}_q$ we have
$\max_{m\in\mathbb{I}_q\setminus\{k\}} \max_{\ell\in\mathbb{I}_d} \P(\min(a_{k\ell}, a_{m\ell})>0)=0$.
Further, let $\bX=\bA\bZ$ and $\bY=(Y_1,Y_2)^\top=(X_k, \max_{m\in \mathbb{I}_q\setminus\{k\}} X_m )^\top$.
\begin{itemize}
    \item[(a)] Suppose $\bZ\in\PMOC(\alpha,\theta,\lambda^{=})$   and  $\E\|\bA\|^{\frac{3\alpha}{2}+\epsilon}<\infty$ for some $\epsilon>0$. Then \linebreak $\bY \in\MRV(\alpha_2,b_2,{\mu}_2^*,\E_2^{(2)})$  where $\alpha_2=
    \frac{3\alpha}{2}$, {$b_2(t)=\theta^{\frac{1}{\alpha}} t^{\frac{2}{3\alpha}}$} and for $\bx\in(\bzero,\binfty)$, we have
    \begin{eqnarray*}
        {\mu}_2^*((\bx,\infty))=%\frac{1}{\binom{d}{2}}
       \sum\limits_{\ell,j=1}^d\ \E\left[\min\left\{\frac{a_{k\ell}}{x_1}, \max_{m\in \mathbb{I}_q\setminus\{k\}} \left\{\frac{a_{m j}}{x_2}\right\} \right\}^{\alpha}\max\left\{\frac{a_{k\ell }}{x_1},  \max_{m\in \mathbb{I}_q\setminus\{k\}} \left\{\frac{a_{mj}}{x_2}\right\}\right\}^{\frac{\alpha}{2}}\right].
    \end{eqnarray*}
    Moreover, as $t\to\infty$,
        \begin{align*}
        \P(Y_1>tx_1|Y_2>tx_2) & \sim (\theta t^{-\alpha})^{\frac{1}{2}}x_2^{\alpha}\,\mu_2^*((\bx,\infty))\,\Big(\sum\limits_{j=1}^d\E\big[\max_{m\in \mathbb{I}_q\setminus \{k\}} a_{m j}^{\alpha} \big]\Big)^{-1}, \\%\intertext{and,}
         \P(Y_2>tx_2|Y_1>tx_1) & \sim (\theta t^{-\alpha})^{\frac{1}{2}}x_1^{\alpha}\,\mu_2^*((\bx,\infty))\,\Big( \sum\limits_{\ell=1}^d\E\big[ a_{k\ell}^{\alpha} \big]\Big)^{-1}.
        \end{align*}
     Additionally, if the non-zero components of $\bA$ have bounded support, bounded away from zero, then there exists $0<\upsilon^*_1<\upsilon^*_2<\infty$ such that for all $0<\upsilon<\upsilon_1^*$, as $\gamma\downarrow 0$ we have
        \begin{eqnarray*}
            \CoVaR_{\upsilon\gamma^{\frac{1}{2}}|\gamma}(Y_1|Y_2)\sim \upsilon^{-\frac{1}{\alpha}} \left(\frac{ \sum\limits_{\ell,j=1}^d  \E\left[a_{k\ell}^\alpha \max_{m\in \mathbb{I}_q\setminus\{k\}} a_{m j}^{\alpha/2}\right]}{\sum\limits_{\ell=1}^d\E\big[ a_{k\ell}^{\alpha} \big]{\left(\sum\limits_{j=1}^d\E\big[\max_{m\in \mathbb{I}_q\setminus \{k\}} a_{m j}^{\alpha}\big]\right)^{\frac{1}{2}}}}\right)^{\frac{1}{\alpha}}\VaR_{\gamma}(Y_1),
        \end{eqnarray*} 
        and for all $\upsilon_2^*<\upsilon<\infty$, as $\gamma\downarrow 0$, we have 
        \begin{eqnarray*}
            \CoVaR_{\upsilon\gamma^{\frac{1}{2}}|\gamma}(Y_1|Y_2)\sim \upsilon^{-\frac{2}{\alpha}} \left(\frac{ \sum\limits_{\ell,j=1}^d  \E\left[a_{k\ell}^{\alpha/2} \max_{m\in \mathbb{I}_q\setminus\{k\}} a_{m j}^{\alpha}\right]}{\left(\sum\limits_{\ell=1}^d\E\big[ a_{k\ell}^{\alpha} \big]\right)^{\frac{1}{2}}{\left(\sum\limits_{j=1}^d\E\big[\max_{m\in \mathbb{I}_q\setminus \{k\}} a_{m j}^{\alpha}\big]\right)}}\right)^{\frac{2}{\alpha}}\VaR_{\gamma}(Y_1).
        \end{eqnarray*} 
       Finally,  $\text{ECI}(Y_1|Y_2)=\text{ECI}(Y_2|Y_1)=2$.   
    \item[(b)] Suppose $\bZ\in\PMOC(\alpha,\theta,\lambda^{\propto})$
    and    $\E\|\bA\|^{\alpha\frac{3d+2}{2(d+1)}+\epsilon}<\infty$ for some $\epsilon>0$. Then  \linebreak $\bY\in\MRV({\alpha}_2,{b}_2,{\mu}_2^*,\E_2^{(2)})$ where ${\alpha}_2=
    \alpha\frac{3d+2}{2(d+1)}$, {${b}_2(t)=\theta^{\frac{1}{\alpha}} t^{\frac{2(d+1)}{(3d+2)\alpha}}$}and for $\bx\in(\bzero,\binfty)$, we have %{$b_2^*=(d\theta t)^{1/\widetilde{\alpha}_2}$ and}
    \begin{eqnarray*}
       {\mu}_2^*((\bx,\infty))=%\frac{1}{\binom{d}{2}}
        \sum\limits_{\ell,j=1}^d \E\left[\min\left\{\frac{a_{k\ell}}{x_1}, \max_{m\in \mathbb{I}_q\setminus\{k\}} \left\{\frac{a_{m j}}{x_2}\right\} \right\}^{\alpha}\max\left\{\frac{a_{k\ell}}{x_1},  \max_{m\in \mathbb{I}_q\setminus\{k\}} \left\{\frac{a_{m j }}{x_2}\right\}\right\}^{\alpha\frac{d}{2(d+1)}}\right],
    \end{eqnarray*}
    Moreover, as $t\to\infty$,
         \begin{align*}
        \P(Y_1>tx_1|Y_2>tx_2)&\sim (\theta t^{-\alpha})^{\frac{d}{2(d+1)}}x_2^{\alpha}\,{\mu}_2^*((\bx,\infty))\,\Big(\sum\limits_{j =1}^d\E\big[\max_{m\in \mathbb{I}_q\setminus \{k\}} a_{m j }^{\alpha} \big]\Big)^{-1},\\
         \P(Y_2>tx_2|Y_1>tx_1) & \sim (\theta t^{-\alpha})^{\frac{d}{2(d+1)}}x_1^{\alpha}\,{\mu}_2^*((\bx,\infty))\,\Big( \sum\limits_{\ell=1}^d\E\big[ a_{k\ell}^{\alpha} \big]\Big)^{-1}.
        \end{align*}
   \end{itemize}
    Additionally, if the non-zero components of $\bA$ have bounded support, bounded away from zero, then there exists $0<\upsilon^*_1<\upsilon^*_2<\infty$ such that for all $0<\upsilon<\upsilon_1^*$, as $\gamma\downarrow 0$, we have 
        \begin{eqnarray*}
            \CoVaR_{\upsilon\gamma^{\frac{d}{2(d+1)}}|\gamma}(Y_1|Y_2)\sim \upsilon^{-\frac{1}{\alpha}}\left(\frac{ \sum\limits_{\ell,j=1}^d  \E\left[a_{k\ell}^\alpha \max_{m\in \mathbb{I}_q\setminus\{k\}} a_{m j}^{\frac{\alpha d}{2(d+1)}}\right]}{\sum\limits_{\ell=1}^d\E\big[ a_{k\ell}^{\alpha} \big]{\left(\sum\limits_{j=1}^d\E\big[\max_{m\in \mathbb{I}_q\setminus \{k\}} a_{m j}^{\alpha}\big]\right)^\frac{d}{2(d+1)}}}\right)^{\frac{1}{\alpha}}\VaR_{\gamma}(Y_1),
        \end{eqnarray*} 
        and for all $\upsilon_2^*<\upsilon<\infty$, as $\gamma\downarrow 0$, we have 
        \begin{eqnarray*}
            \CoVaR_{\upsilon\gamma^{\frac{d}{2(d+1)}}|\gamma}(Y_1|Y_2)\sim \upsilon^{-\frac{2(d+1)}{d\alpha}} \left(\frac{ \sum\limits_{\ell,j=1}^d  \E\left[a_{k\ell}^{\frac{\alpha d}{2(d+1)}} \max_{m\in \mathbb{I}_q\setminus\{k\}} a_{m j}^{\alpha}\right]}{\left(\sum\limits_{\ell=1}^d\E\big[ a_{k\ell}^{\alpha} \big]\right)^{\frac{1}{2(d+1)}}{\sum\limits_{j=1}^d\E\big[\max_{m\in \mathbb{I}_q\setminus \{k\}} a_{m j}^{\alpha}\big]}}\right)^{\frac{2(d+1)}{d\alpha}}\VaR_{\gamma}(Y_1).
        \end{eqnarray*} 
       Finally, $\text{ECI}(X_1|X_2)=\text{ECI}(Y_2|Y_1)=2+\frac2d$.    
\end{proposition}

\begin{proposition}[Gaussian copula] \label{prop:condproboneandmaxGC}
Let $\bZ\in \PGC(\alpha, \theta, \Sigma)$ with $\Sigma=(\rho_{\ell j})_{1\leq \ell,j\leq d}$ positive definite.
Suppose $\bA \in \R_+^{q \times d}$ is a  random matrix  with almost surely no trivial rows, independent of $\bZ$, and for fixed $k\in \mathbb{I}_q$ we have
$$\max_{m\in\mathbb{I}_q\setminus\{k\}} \max_{\ell\in\mathbb{I}_d} \P(\min(a_{k\ell}, a_{m\ell})>0)=0.$$
 Also, define
 \begin{align*}
     %\rhom & =\max\left\{\rho_{\ell j}: \ell, j \in \mathbb{I}_d, \ell\neq j  \right\},\\
      \rho^* & =\max\left\{\rho_{\ell j}: \ell,j \in \mathbb{I}_d, \ell\neq j \text{ and }\, \max_{m\in \mathbb{I}_q\setminus\{k\}}\P(\min(a_{k\ell}, a_{mj})>0)>0 \right\}.
   %  \tilde{\rho} & =\max\{\rho_{jl}: j,l \in \mathbb{I}_d, j\neq l \; \text{ and }\, \P(\min(a_{1j}, a_{2l}>0)>0}\}
 \end{align*}
    Suppose 
   { $\E\|\bA\|^{\frac{2\alpha}{1+\rho^*}+\epsilon}<\infty$} for some $\epsilon>0$ and let $\bX=\bA\bZ$. 
    Then we have \linebreak $\bY=(Y_1,Y_2)^\top=(X_k, \max_{m\in \mathbb{I}_q\setminus\{k\}} X_m )^\top\in\MRV(\alpha_2^*,b_2^*,{\mu}_2^*,\E_2^{(2)})$ with
    \begin{align*}
        & \alpha_2^* =\frac{2\alpha}{1+\rho^*}, \quad\quad
        b_2^{*\leftarrow}(t)= C^*({\rho^*, \alpha}) (\theta t^{-\alpha})^{-\frac{2}{1+\rho^*}}(\log t)^{\frac{\rho^*}{1+\rho^*}} ,\\
    & {\mu}_2^*((\bx,\infty)) = D^*({\rho^*,\alpha, \bA})(x_1x_2)^{-\frac{\alpha}{1+\rho^*}}, \quad \bx\in(\bzero,\binfty),
    %\sum_{\genfrac{}{}{0pt}{}{j,l=1 }{\rho_{jl}=\rho^*}}^d\E(a_{1j}a_{2l})^{\frac{\alpha}{1+\rho^*}} >0, \quad \bx\in(\bzero,\binfty).
    \end{align*}
    where for $\rho\in(-1,1), \alpha>0, \theta>0$ and $\bA\in \R_+^{2\times d}$, we define
    \begin{eqnarray*}
    C^*(\rho, \alpha) & = &{ ({2\pi})^{-\frac{1}{1+\rho}}} (2\alpha)^{\frac{\rho}{1+\rho}}, \label{eq:Cgaussoneandmax}\\
    D^*(\rho, \theta, \bA) & = &\frac{1}{2\pi} \frac{(1+\rho)^{3/2}}{(1-\rho)^{1/2}}\sum_{(\ell,j):\, \rho_{\ell j}=\rho }\E\left[a_{k\ell}^{\frac{\alpha}{1+\rho}}  \max_{m\in \mathbb{I}_q\setminus\{k\}} \{a_{m j}\}^{\frac{\alpha}{1+\rho}}\right].\label{eq:Dgaussoneandmax}        
    \end{eqnarray*}
Moreover, as $t\to\infty$,
        \begin{align*}
          \P(Y_1>x_1t|Y_2>x_2t) & \sim   (\theta t^{-\alpha})^{\frac{1-\rho^*}{1+\rho^*}}{(\log t)^{-\frac{\rho}{1+\rho}}}x_1^{-\frac{\alpha}{1+\rho^*}}x_2^{\frac{\alpha\rho^*}{1+\rho^*}}\frac{C^*(\rho^*,\alpha)^{-1}D^*(\rho^*,\alpha, \bA)}{\sum\limits_{j=1}^d\E\big[\max_{m\in \mathbb{I}_q\setminus \{k\}} a_{m j}^{\alpha} \big]}, \\ % \intertext{and,}
         \P(Y_2>tx_2|Y_1>tx_1) & \sim (\theta t^{-\alpha})^{\frac{1-\rho^*}{1+\rho^*}}{(\log t)}^{-\frac{\rho}{1+\rho}}x_1^{\frac{\alpha\rho^*}{1+\rho^*}}x_2^{-\frac{\alpha}{1+\rho^*}}\frac{C^*(\rho^*,\alpha)^{-1}D^*(\rho^*,\alpha, \bA)}{\sum\limits_{\ell=1}^d\E\big[ a_{k\ell}^{\alpha} \big]}.
        \end{align*}
    Additionally, with $g(\gamma)=\gamma^{\frac{1-\rho^*}{1+\rho^*}}(-\alpha^{-1}\log \gamma)^{-\frac{\rho^*}{1+\rho^*}}$ and $0<\upsilon< 1$ we have as $\gamma\downarrow 0$,
     \begin{align*}
     \CoVaR_{\upsilon g(\gamma)|\gamma}(Y_1|Y_2) & \sim \upsilon^{-\frac{1+\rho^*}{\alpha}} \frac{\left(C^*(\rho^*,\alpha)^{-1}D^*(\rho^*,\alpha, \bA)\right)^{\frac{1+\rho^*}{\alpha}}}{\left(\sum\limits_{\ell=1}^d\E\big[ a_{k\ell}^{\alpha} \big]\sum\limits_{j=1}^d\E\big[\max_{m\in \mathbb{I}_q\setminus \{k\}} a_{mj}^{\alpha} \big]\right)^{\frac{1}{\alpha}}}\VaR_{\gamma}(Y_1),
        %\CoVaR_{\upsilon g(\gamma)|\gamma}(Y_1|Y_2) & \sim \upsilon^{-\frac{1+\rho^*}{\alpha}} \left(\frac{C^*(\rho^*,\alpha,\theta)^{-1}D^*(\rho^*,\alpha, \bA)}{{\left(\sum\limits_{\ell=1}^d\E\big[\max_{m\in \mathbb{I}_q\setminus \{k\}} a_{m\ell}^{\alpha} \big]\right)^\frac{2}{1+\rho^*}}}\right)^{\frac{1+\rho^*}{\alpha}}\VaR_{\gamma}(Y_2),
        \intertext{and}
         \CoVaR_{\upsilon g(\gamma)|\gamma}(Y_2|Y_1) & \sim \upsilon^{-\frac{1+\rho^*}{\alpha}}\frac{\left(C^*(\rho^*,\alpha)^{-1}D^*(\rho^*,\alpha, \bA)\right)^{\frac{1+\rho^*}{\alpha}}}{{\left(\sum\limits_{\ell=1}^d\E\big[ a_{k\ell}^{\alpha}\big]\sum\limits_{j=1}^d\E\big[\max_{m\in \mathbb{I}_q\setminus \{k\}} a_{m j}^{\alpha} \big] \big]\right)^\frac{1}{\alpha}}}\VaR_{\gamma}(Y_2).
         %\CoVaR_{\upsilon g(\gamma)|\gamma}(Y_2|Y_1) & \sim \upsilon^{-\frac{1+\rho^*}{\alpha}} \left(\frac{C^*(\rho^*,\alpha,\theta)^{-1}D^*(\rho^*,\alpha, \bA)}{{\left(\sum\limits_{\ell=1}^d\E\big[ a_{k\ell}^{\alpha} \big]\right)^\frac{2}{1+\rho^*}}}\right)^{\frac{1+\rho^*}{\alpha}}\VaR_{\gamma}(Y_1).
     \end{align*}
      {Finally, $g(t^{-1})\in \RV_{-\frac{1-\rho^*}{1+\rho^*}}$ and hence, } $\text{ECI}(Y_1|Y_2)=\text{ECI}(Y_2|Y_1)=\frac{1+\rho^*}{1-\rho^*}$.
\end{proposition}

\end{document}